\documentclass{article}
\usepackage[lang = british]{ems-aihpd} 

\usepackage{bbm}
\usepackage{kbordermatrix}
\usepackage{soul} 
\usepackage{adjustbox}
\theoremstyle{plain}
\newtheorem{theorem}{Theorem}[section]
\newtheorem{lemma}[theorem]{Lemma}
\newtheorem{corollary}[theorem]{Corollary}

\theoremstyle{definition}
\newtheorem{definition}{Definition}[section]
\newtheorem{example}[theorem]{Example}
\newtheorem{remark}[theorem]{Remark}

\newtheorem*{theorem*}{Theorem}
\newtheorem*{definition*}{Definition}

\newcommand{\monoid}[1]{\langle \underline{#1} \rangle}
\newcommand{\bsym}[1]{ {\boldsymbol{ #1 }} }
\newcommand*{\Tr}[1]{\mathop{}\!\mathrm{Tr}{\left(#1\right)}}
\newcommand*{\tr}[2]{\mathop{}\!\mathrm{Tr}_{#1}{\left(#2\right)}}
\newcommand{\Span}[1]{\mathrm{Span}\left\lbrace #1 \right\rbrace}

\def\id{\mathbbm{1}}

\newcommand{\ket}[1]{\ensuremath{\left|#1\right\rangle}}
\newcommand{\braket}[2]{\langle #1  |#2\rangle}
\newcommand{\ketbra}[2]{|#1\rangle\langle #2|  }
\newcommand{\sandwich}[3]{\langle #1|#2 |#3\rangle  }

\let\C\relax 
\newcommand{\cB}{\mathcal{B}}
\newcommand{\cH}{\mathcal{H}}
\newcommand{\A}{A_{a|x}}
\newcommand{\B}{B_{b|y}}
\newcommand{\C}{C_{c|z}}
\newcommand{\hA}{\hat{A}_{a|x}}
\newcommand{\hB}{\hat{B}_{b|y}}
\newcommand{\hC}{\hat{C}_{c|z}}

\newcommand{\D}{D_{d|w}}
\newcommand{\hD}{\hat{D}_{d|w}}

\newcommand{\Ainf}[1]{{A}^{#1}_{a^{#1}|x^{#1}}}
\newcommand{\Binf}[2]{{B}^{#1,#2}_{b^{#1,#2}|y^{#1,#2}}}
\newcommand{\Cinf}[1]{{C}^{#1}_{c^{#1}|z^{#1}}}
\newcommand{\hAinf}[1]{\hat{A}^{#1}_{a^{#1}|x^{#1}}}
\newcommand{\hBinf}[2]{\hat{B}^{#1,#2}_{b^{#1,#2}|y^{#1,#2}}}
\newcommand{\hCinf}[1]{\hat{C}^{#1}_{c^{#1}|z^{#1}}}

\newcommand{\balpha}{ \boldsymbol{\alpha}}
\newcommand{\bbeta}{ \boldsymbol{\beta}}
\newcommand{\bgamma}{ \boldsymbol{\gamma}}
\newcommand{\bomega}{ \boldsymbol{\omega}}
\newcommand{\bnu}{ \boldsymbol{\nu}}
\newcommand{\bkappa}{ \boldsymbol{\kappa}}
\newcommand{\halpha}{ \hat{\alpha}}
\newcommand{\hbeta}{ \hat{\beta}}
\newcommand{\hgamma}{ \hat{\gamma}}
\newcommand{\homega}{ \hat{\omega}}
\newcommand{\hnu}{ \hat{\nu}}
\newcommand{\hkappa}{ \hat{\kappa}}

\newcommand{\bmu}{ \boldsymbol{\mu}}

\newcommand{\bdelta}{ \boldsymbol{\delta}}
\newcommand{\hdelta}{ \hat{\delta}}

\numberwithin{equation}{section}
\hypersetup{
  pdftitle   = {Two Convergent NPA-like Hierarchies for the Quantum Bilocal Scenario},
  pdfauthor  = {Marc-Olivier Renou; Xiangling Xu; Laurens T. Ligthart},
  pdfsubject = {Quantum information theory; network nonlocality (bilocal scenario)},
  pdfkeywords= {generalized NPA hierarchy, bilocal networks, scalar extension, polarization technique, bilocal factorisation hierarchy, quantum correlations, semidefinite programming SDP, noncommutative polynomial optimization},
  pdfcreator = {LaTeX with hyperref},
  pdflang    = {en},
  pdfdisplaydoctitle = true
}

\begin{document}

\title{Two Convergent NPA-like Hierarchies for the Quantum Bilocal Scenario}

\emsauthor{1}{Marc-Olivier Renou$^{1,2,3,4}$}{M.-O.~Renou}
\emsauthor{2}{Xiangling Xu$^{1,2,3,5, \dagger}$}{X.~Xu}
\emsauthor{3}{Laurens T. Ligthart$^{5}$}{L.T.~Ligthart}


\emsaffil{1}{$^1$Inria Paris-Saclay, B\^atiment Alan Turing, 1 rue Honor\'e d’Estienne d’Orves, 91120 Palaiseau, France}
\emsaffil{2}{$^2$CPHT, Ecole polytechnique, Institut Polytechnique de Paris, Route de Saclay, 91128 Palaiseau, France}
\emsaffil{3}{$^3$LIX, Ecole polytechnique, Institut Polytechnique de Paris, Route de Saclay, 91128 Palaiseau, France}

\emsaffil{4}{$^4$ICFO-Institut de Ciencies Fotoniques, The Barcelona Institute of Science and Technology, Castelldefels (Barcelona), Spain}

\emsaffil{5}{$^5$Institute for Theoretical Physics, ETH Z\"urich, Switzerland}

\emsaffil{6}{$^6$Institute for Theoretical Physics, University of Cologne, Germany}

\emsaffil{7}{$^\dagger$\scriptsize\texttt{xu.xiangling@inria.fr}}






\begin{abstract}
Characterising the correlations that arise from locally measuring a single part of a joint quantum system is one of the main problems of quantum information theory.
The seminal work [M. Navascués \emph{et al}, NJP 10,7,073013 (2008)], known as the NPA hierarchy, reformulated this question as a polynomial optimisation problem over noncommutative variables and proposed a convergent hierarchy of necessary conditions, each testable using semidefinite programming.
More recently, the problem of characterising the quantum \emph{network} correlations, which arise when locally measuring \emph{several independent} quantum systems distributed in a network, has received considerable interest.
Several generalisations of the NPA hierarchy, such as the scalar extension [Pozas-Kerstjens \emph{et al}, Phys. Rev. Lett. 123, 140503 (2019)], were introduced while their converging sets remain unknown.
In this work, we introduce a new \emph{bilocal factorisation NPA hierarchy}, prove its equivalence to a \emph{modified} bilocal scalar extension NPA hierarchy, and characterise its convergence in the case of the simplest network, the bilocal scenario.
We further explore its relations with the other known generalisations.
\end{abstract}

\maketitle


\tableofcontents
\section{Introduction}

Quantum correlations (the joint probability distribution established between independent parties measuring a single shared quantum state) are important both for foundations and applications of quantum theory.
They underlie the Bell theorem~\cite{BellTheorem,CHSH}, a theoretical physics milestone~\cite{BrunnerReview}. 
They are used in the device independent framework to obtain practical applications in terms of certification~\cite{Scarani2012device}, e.g. of quantum devices~\cite{Mayers1998, Arnon2016,SelfTestingReview}, randomness~\cite{Pironio2010,Colbeck2012} or cryptographic protocols~\cite{Acin2006,Vazirani2014}.
In all these approaches, no assumption is made on the state and measurements.
Hence, being able to characterise the set of probability distributions allowed by quantum theory is crucial. 
This can be done with the \emph{NPA hierarchy}~\cite{NPA2008,ParriloHierarchy,LasserreHierarchy}, which provides a converging hierarchy of Semi-Definite Programs (SDPs) to this set, using non-commutative polynomial optimisation theory. 

More recently, a generalised framework in which several independent quantum sources are distributed and measured by several parties in a \emph{network scenario} received considerable interest~\cite{NetworkNonlocReview,Branciard2010,Fritz2012,Renou2022a,Renou2019}. 
It offered new insights on foundations and applications of quantum theory~\cite{Henson2014,CoiteuxRoyPRL,CoiteuxRoyPRA,Supic2022,Pavel2022}, in particular from the simple bilocal scenario in which two independent sources are distributed to three parties in a line~\cite{Renou2021,Weilenmann2020,Lee2018,Renou2018,Bancal2018}.
This calls for a generalisation of the \emph{NPA hierarchy}. 
Several methods were proposed, but to our knowledge no proof of monotonic convergence exists. 

In this paper, we introduce two new hierarchies: the \emph{factorisation NPA hierarchy} and a modified \emph{scalar extension hierarchy}. The latter is based on the ideas presented in~\cite{ScalarExtension2019,klep2024state}. We prove that both hierarchies converge to the set of Projector Bilocal Quantum Distributions, a relaxation of the quantum correlations obtainable in the bilocal scenario.
We discuss the relation of these two hierarchies with the \emph{inflation-NPA hierarchy} of~\cite{QuantumInflation2021} and generalisation of our result to generic networks.

\medskip
\noindent\textbf{Note Added}---This work was first made available as a preprint in October 2022~\cite{renou2022bilocalv1} and inspired the works~\cite{ligthart2023inflation, klep2024state}.
In particular, our factorisation NPA hierarchy approach (Section~\ref{sec:TripartiteQDistribStandardNPA}) motivated the results of~\cite{ligthart2023inflation}, and the proposed scalar extension polynomial construction (Appendix~\ref{sec:ScaExtPoly}) is systematically studied by~\cite{klep2024state}.
However, the formulation of the initial scalar extension NPA hierarchy was incorrect, resulting in an incorrect proof of convergence.
In this new version, we propose a correction of it, inspired both by the works~\cite{ligthart2023inflation,klep2024state}.
LL joined the project at this second stage and contributed to the development of the corrected hierarchy and manuscript.

\subsection{Motivation}

\textbf{Bilocal quantum correlations}\\
In this article, we focus on the quantum bilocal scenario involving three parties, $A,B,C$ measuring two independent sources $\rho, \sigma$ distributed through the bilocal network (see Figure~\ref{fig:ThreePartiesScenarios}b).
We focus on the set of quantum correlations $\vec{P}=\{p(abc|xyz)\}$  which can be obtained in this scenario, where $p(abc|xyz)$ is the probability of measurement results $a, b, c$ given that $A, B, C$ respectively performed measurements labeled by $x, y, z$.
According to the Born rule, we write
\begin{align*}
   p(abc|xyz)=\tr{\rho\otimes\sigma}{A_{a|x}\otimes B_{b|y}\otimes C_{c|z}}, 
\end{align*}
where $\rho\in\cB(\cH_A\otimes\cH_{B_L}), \sigma\in\cB(\cH_{B_R}\otimes\cH_{C})$ are projectors over pure states, $A_{a|x}\in\cB(\cH_A), B_{b|y}\in\cB(\cH_{B_L}\otimes\cH_{B_R}), C_{c|z}\in\cB(\cH_C)$ are PVMs, and $\cH_A,\cH_{B_L},\cH_{B_R},\cH_C$ are arbitrary Hilbert spaces. (Throughout the text, when $\tau$ is a state, we use notation $\tr{\tau}{M}:=\Tr{\tau\cdot M}$.)

Such correlations, which we call \textit{Tensor Bilocal Quantum Distributions}, are the most basic examples of quantum network correlations. 
These arise when multiple independent quantum states are distributed among parties within a network~\cite{NetworkNonlocReview}. 

Arbitrarily good \emph{inner} approximations of this set of correlations can be obtained, e.g. by sampling Hilbert spaces, states, and measurement operators.
As we explain below, arbitrarily good \emph{outer} approximations of this set cannot be obtained. 
In this work, we will provide a method for obtaining arbitrarily good outer approximations of a natural relaxation of these correlations, which we introduced below and call \emph{Projector Bilocal Quantum Distributions}.

\begin{figure}
    \centering
    \includegraphics[width=\textwidth]{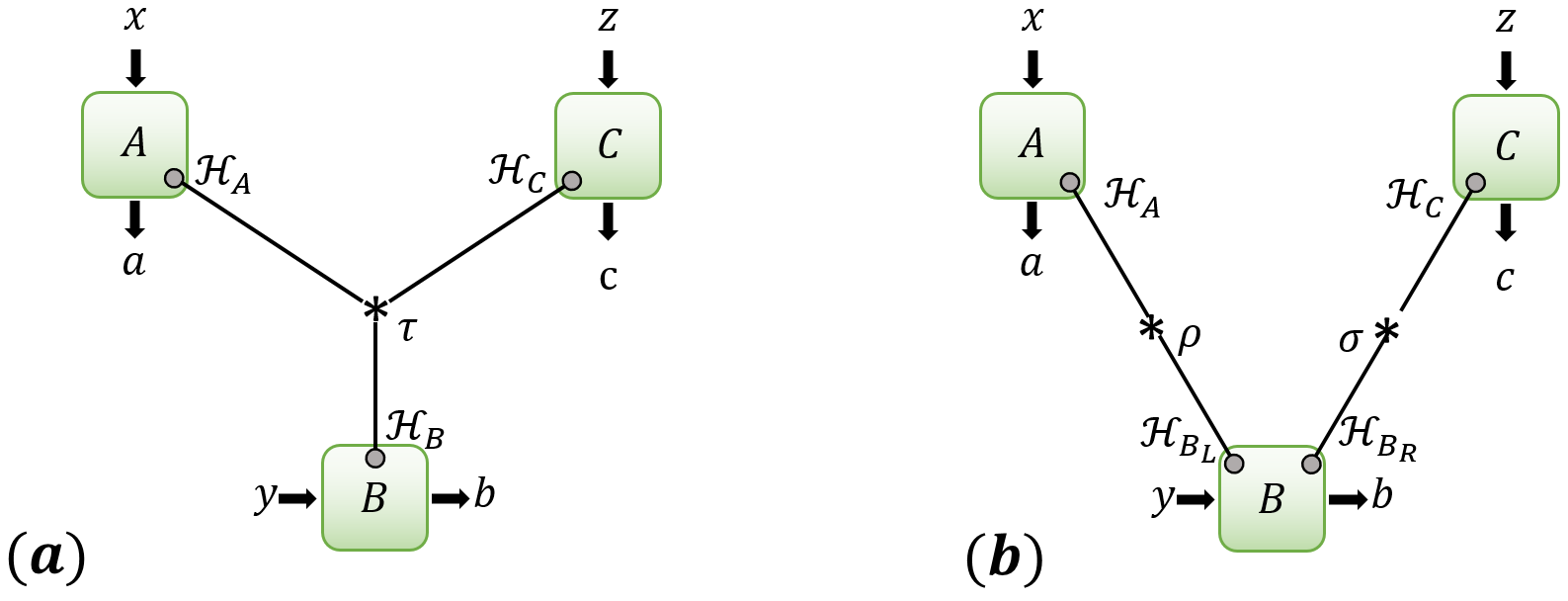}
    \caption{(a) Standard three-party Bell scenario. A three-particles quantum state, mathematically represented by a projector over a pure state $\tau\in\cB(\cH_A\otimes\cH_B\otimes\cH_C)$ (a positive operator such that $\Tr{\tau}=1$ and $\tau^2=\tau$), is created, each particle is sent to one of three separated parties $A, B, C$. $A$ measures the received particle according to some input $x$, obtaining an output $a$, mathematically represented by PVMs $A_{a|x}\in\cB(\cH_A)$ (a set of positive operators such that $\sum_a A_{a|x}=\id$), and $B,C$ do the same. The behavior of the experiment is described by a probability distribution $\vec{P}=\{p(abc|xyz)\}$ with $p(abc|xyz)=\tr{\tau}{A_{a|x}\otimes B_{b|y}\otimes C_{c|z}}$. \\ 
    (b) Bilocal scenario. Two two-particles quantum state (mathematically represented by two projectors over pure states $\rho\in\cB(\cH_A\otimes\cH_{B_L}), \sigma\in\cB(\cH_{B_R}\otimes\cH_{C})$ such that $\Tr{\rho}=\Tr{\sigma}=1$ and $\rho^2=\rho, \sigma^2=\sigma$) are created, $A$ (resp. $C$) receiving one particle from $\rho$ (resp. $\sigma$) and $B$ one of each state, as depicted. $A, B, C$ measurement operators are mathematically represented by positive operators $A_{a|x}\in\cB(\cH_A), B_{b|y}\in\cB(\cH_{B_L}\otimes\cH_{B_R}), C_{c|z}\in\cB(\cH_C)$ such that $\sum_a A_{a|x}=\sum_b B_{b|y}=\sum_c C_{c|z}=\id$). The behavior of the experiment is described by a probability distribution $\vec{P}=\{p(abc|xyz)\}$ with $p(abc|xyz)=\tr{\rho\otimes\sigma}{A_{a|x}\otimes B_{b|y}\otimes C_{c|z}}.$ }
    \label{fig:ThreePartiesScenarios}
\end{figure}

\medskip
\noindent\textbf{Standard Bell scenario quantum correlations and NPA hierarchy}\\
Let us first quickly recall an existing result already known for the more standard Bell scenario involving three parties, $A,B,C$ measuring \emph{a unique source} $\tau$ (see Figure~\ref{fig:ThreePartiesScenarios}a).
In this scenario, a quantum correlation $\vec{P}=\{p(abc|xyz)\}$ (a \emph{tensor quantum correlation}) is given through the Born rule
\begin{align*}
    p(abc|xyz)=\tr{\tau}{A_{a|x}\otimes B_{b|y}\otimes C_{c|z}},
\end{align*}
where the projector over a pure state $\tau\in\cB(\cH_A\otimes\cH_B\otimes\cH_C)$ represents the state, the PVMs $A_{a|x}\in\cB(\cH_A), B_{b|y}\in\cB(\cH_B), C_{c|z}\in\cB(\cH_C)$ represent the performed measurement, and $\cH_A,\cH_B,\cH_C$ are arbitrary Hilbert spaces.

An outer approximation was found by Navascués, Pironio, and Acín in~\cite{NPA2008}.
Their method, called the NPA hierarchy, is the noncommutative counterpart of the Parrilo-Lassere hierarchy~\cite{ParriloHierarchy,LasserreHierarchy}. 
It provides a converging hierarchy of Semi-Definite Programs (SDPs) to any noncommutative polynomial optimisation problem over an archimedean semialgebraic set (a condition always satisfied in our context).

For the set of quantum correlation $\vec{P}$ obtainable in the scenario of Figure~\ref{fig:ThreePartiesScenarios}a, for a hierarchy level $n$, the NPA method asks for the existence of a \emph{moment matrix} (or Hankel matrix) $\Gamma_n$ compatible with $\vec{P}$ (see Equation~\eqref{eq:MomentMatrix}).
The non-existence of such $\Gamma_n$ proves that $\vec{P}$ cannot be a tensor quantum correlation. 
Importantly, this existence problem is an SDP, which is easily solvable on a computer, with certificates of non-existence.
Increasing the hierarchy level $n$ (which indexes the size of $\Gamma_n$) gives stronger tests, and as $n$ increases, this sequence of tests singles out the larger set of \emph{commutator quantum correlations} $\vec{Q}=\{q(abc|xyz)\}$ writing
\begin{align*}
    q(abc|xyz)=\tr{\tau}{A_{a|x} B_{b|y} C_{c|z}},
\end{align*}
where $\tau$ is a projector over a pure state, $A_{a|x}, B_{b|y}, C_{c|z}\in\cB(\cH)$ are PVMs commuting together and all these operators are now part of the same global Hilbert spaces $\cH$.
It was recently proven that for arbitrary infinite dimensional Hilbert spaces, the set of tensor quantum correlations is strictly included in the set of commutator quantum correlations~\cite{ji2021mip}: informally, $\{\vec{P}\}\subsetneq\{\vec{Q}\}$.
However, in the case where the Hilbert spaces are restricted to be finite dimensional, the two sets coincide~\cite{tsirelsonproblem}. 

Moreover,~\cite{ji2021mip} proved that no algorithm can provide converging outer approximations of the set of tensor quantum correlations (hence of the set of tensor bilocal quantum correlations). 
Note also that the set of \emph{commutator quantum correlations} is considered as a potential \emph{alternative axiomatisation} of the set of quantum correlations in the standard Bell scenario.

\subsection{Main Contribution}

In this work, we introduce two hierarchies of outer approximations of the set of \emph{Projector Bilocal Quantum Distributions} which converge to that set.

Our first hierarchy, the \emph{factorisation bilocal NPA Hierarchy}, is introduced in Section~\ref{sec:nonSDPNPAhierarchy}.  
It is based on the same moment matrix $\Gamma$ as in the standard NPA hierarchy, to which we impose an additional nonlinear factorisation constraint corresponding to the relation \begin{equation}\label{eq:IntroFactorisation}
    \tr{\tau}{\halpha\hgamma} = \tr{\tau}{\halpha} \cdot \tr{\tau}{\hgamma},
\end{equation} 
where $\halpha$ (resp. $\hgamma$) is a monomial in Alice's PVMs (resp. Charlie's PVMs), and $\tau$ should be thought of as $\tau=\rho\otimes\sigma$: see Definition~\ref{def:FactorBMMatrix}.
This results in a hierarchy of (impractical) non-SDP problems.

Our second hierarchy, the \emph{scalar extension bilocal hierarchy}, is introduced in Section~\ref{sec:bilocal_hierarchy}. It is a modified version of what was already proposed in~\cite{ScalarExtension2019}, where the proof of convergence was absent.
To this, we add a scalar extension letter, $\kappa_{\bomega}$, for each operator-matching word $\bomega$ to the standard NPA hierarchy. These scalar extension letters should model the scalar operator $\hkappa_{\bomega}=\tr{\tau}{\homega}\id_{\cH}$, in particular, they are required to commute with all other letters.
Inspired by the \emph{polarisation} technique in~\cite{ligthart2023inflation}, the factorisation constraints are imposed through the quadratic version of Equation~\eqref{eq:IntroFactorisation}:
\begin{align}
    \tr{\tau}{ ( \hkappa_{\balpha\bgamma} - \hkappa_{\balpha} \hkappa_{\bgamma} )^2 }.
\end{align}
This results in a hierarchy of SDPs, although it is worth noticing that the hierarchy is not constructive: the existence of a solution to all NPO levels implies only the existence, but one cannot obtain a concrete quantum model via the GNS construction in general.
Lastly, the hierarchy is similar to the \emph{state polynomial optimisation} in~\cite{klep2024state}, with the main difference that the polynomial equality constraints are imposed through polarisation instead of localising matrices.

Our main result is the following theorem (see Theorems~\ref{thm:ConvnonSDPBilocNPAHierarchy} and~\ref{thm:ConvScalarBilocNPAHierarchy}) of convergence of these two hierarchies to the same set of \emph{projector bilocal quantum correlations}, introduced below after the theorem. 

\begin{theorem*}[Convergence of the factorisation bilocal and scalar extension bilocal hierarchies]
Let $\vec{Q}=\{q(abc|xyz)\}$ be a probability distribution. The following are equivalent:
\begin{enumerate}[(i)]
    \item $\vec{Q}$ is a projector bilocal quantum correlation,
    \item $\vec{Q}$ passes all factorisation bilocal NPA hierarchy tests,
    \item $\vec{Q}$ passes all scalar extension bilocal hierarchy tests,
\end{enumerate} 
\end{theorem*}
\noindent
where the set of projector bilocal quantum correlations (see Definition~\ref{def:CommutatorBilocQDistrib}) is a strict relaxation of the set of tensor bilocal quantum, defined as
\begin{definition*}[Projector Bilocal Quantum Distributions]
$\vec{Q}$ is a \emph{Projector Bilocal Quantum Distribution} iff there exist a Hilbert space $\cH$, projectors (possibly infinite trace) $\rho, \sigma$, and PVMs $\hA, \hB, \hC\in\cB(\cH)$ such that
\begin{enumerate}[(i)]
    \item $q(abc|xyz)=\Tr{\rho\sigma\cdot \hA\hB\hC}$
    \item $\forall \hA, \hB, \hC, [\hA, \hB]=[\hB, \hC]=[ \hC,\hA]=0$ 
    \item $\tau:=\rho\sigma =\sigma\rho$ is a projector over a pure state (i.e. $\Tr{\tau}=1$ and $\tau^2=\tau$)
    \item $\forall \hA, \hC, [\hA,\sigma] = [\rho, \hC] = 0$.
\end{enumerate}
\end{definition*}

In this definition, the projective property of $\tau$ is necessary: as we discuss in Appendix~\ref{sec:AppendixPureMixFormulationInequivalence}, it cannot be obtained by purification, contrary to tensor bilocal quantum distributions. 
We also provide a stopping criterion for our two hierarchies corresponding to the case where the reconstructed Hilbert space $\cH$ is finite-dimensional (see Theorem~\ref{thm:StopCriteriaFinite}).
Lastly, in Section~\ref{sec:OtherNoLoopNetworksInflationNPA} we compare our two hierarchies to the known generalisation, the inflation-NPA hierarchy~\cite{Inflation}, showing it is stricter than our two hierarchies (Theorem~\ref{thm:InflationNPAIsTighter} and explore generalisations to more general quantum network scenarios.

\subsection{First version of this manuscript, discussion, and open problems}\label{sec:discussion}
As explained above, two (generally nonequivalent) axiomatisations exist for quantum correlations in the standard Bell scenario of Fig.~\ref{fig:ThreePartiesScenarios}a.  
The most familiar is the set of \emph{tensor-product quantum correlations}; the other, standard in algebraic quantum field theory, is the set of \emph{commuting-observable quantum correlations}.  
The NPA hierarchy~\cite{NPA2008} provides a converging outer approximation to the latter set.

In this work we target the bilocal scenario of Figure~\ref{fig:ThreePartiesScenarios}b.  
We introduce two hierarchies, one of which can be cast as a SDP program, that converge to the set of \emph{Projector Bilocal Quantum Distributions} and compare them with the inflation-NPA hierarchy~\cite{Inflation}, showing that, in the bilocal case, the latter is generally stricter.  

\medskip
\noindent\textbf{First draft of this work~\cite{renou2022bilocalv1}, developments in~\cite{ligthart2023inflation,klep2024state}, and present work.}  
In the first version of this manuscript~\cite{renou2022bilocalv1}, two of us introduced the set of {Projector Bilocal Quantum Distributions} and claimed to have obtained a first hierarchy that cannot be cast as an SDP, and a second hierarchy that can be cast as an SDP program that both converge to the set of {Projector Bilocal Quantum Distributions}.
While our first claim was correct, the second was wrong (as noticed by Laurens T. Ligthart and David Gross).

Hence, this first draft left two questions open:  
(i) Does our projector-based bilocal set coincide with the natural $C^*$-algebraic one, thereby yielding a bilocal form of Tsirelson’s theorem (Appendix~\ref{sec:AppendixBilocalTsirelson})?  
(ii) Can our first hierarchy be adapted into a new one which can be cast as an SDP program, still converging to the same projector-based bilocal set?

Two subsequent papers~\cite{ligthart2023inflation,klep2024state}, both inspired by the approach followed in this first draft, resolved these issues:
\begin{enumerate}
  \item \emph{Ligthart \& Gross}~\cite{ligthart2023inflation} proved that the natural $C^*$-algebraic definition of bilocal correlations is equivalent to our projector definition (i) (see Appendix~\ref{sec:AppendixBilocalTsirelsonSolved}), thereby solving the bilocal Tsirelson problem and showing that, in finite dimensions, commutator-based and tensor-based models coincide.  
        They also introduced a \emph{polarisation hierarchy} whose convergence, together with that of the inflation-NPA hierarchy, settles the second open question (ii).  
        The quadratic constraints at the heart of their polarisation technique inspired the convergent scalar-extension hierarchy now given in Definition~\ref{def:BSE_hierarchy_Pol}. 
  \item \emph{Klep {et al.}}~\cite{klep2024state} pushed our scalar extension polynomial construction (Appendix~\ref{sec:ScaExtPoly}) further, casting it as a general theory of \emph{state polynomial optimisation}.  
        They proved powerful Positivstellens\"atze for state polynomials in full generality, obtained the associated convergent SDP hierarchies, independently solving (ii). They also demonstrated its numerical implementability.  
        Definition~\ref{def:BSE_hierarchy} is a concrete instance of their framework.
\end{enumerate}
Our present work corrects the incomplete convergence proof for the bilocal scalar extension hierarchy that appeared in our first preprint~\cite{renou2022bilocalv1}, thanks to the discussion with the authors of~\cite{ligthart2023inflation, klep2024state}.
The corrected scalar extension hierarchy (Definition~\ref{def:BSE_hierarchy_Pol}) can be seen as an instance of state polynomial optimisation where the factorisation constraints are imposed via the polarisation technique.

\medskip
\noindent\textbf{Outlook.}
Taken together, the factorisation NPA, polarisation, and scalar extension/state polynomial hierarchies now give a nearly complete SDP characterisation of quantum correlations in bilocal networks.
However, the case of more general quantum networks remains open.
For instance, in the triangle scenario (Figure~\ref{fig:GeneralisedScenarios}a), we note that all hierarchies discussed here are trivial, as no additional factorisation condition can be imposed beyond the standard NPA hierarchy.
The inflation-NPA hierarchy~\cite{QuantumInflation2021} therefore remains the only non-trivial general candidate (Appendix~\ref{sec:AppendixSharedRandBitInflationNPA}).  
Whether it converges to the $C^*$-algebraic set in arbitrary networks is an open problem: the bilocal proof of convergence in~\cite{ligthart2023inflation} fundamentally relies on Arveson’s Radon–Nikodym theorem~\cite{arveson1969subalgebras}, and extending that argument to general networks to ``split'' multiple parties simultaneously and consistently is far from straightforward.  
Establishing (or refuting) the completeness of the inflation-NPA hierarchy for general quantum networks thus remains a key challenge for future work.

\subsection{Content}

Section~\ref{sec:TripartiteQDistribStandardNPA} formally introduces the standard NPA hierarchy and the Projector Bilocal Quantum Distributions (Definition~\ref{def:CommutatorBilocQDistrib}). 
It then introduces the factorisation bilocal NPA hierarchy and demonstrates its convergence to the set of Projector Bilocal Quantum Distributions. 
Section~\ref{sec:ScalarExtHierarchyMainSec} first introduces the original scalar extension hierarchy of~\cite{ScalarExtension2019} and illustrates its failure to converge. 
It then proposes a modified scalar extension hierarchy and proves its convergence to the same set of Projector Bilocal Quantum Distributions.
Section~\ref{sec:OtherNoLoopNetworksInflationNPA} discusses the connection of our two hierarchies to the inflation-NPA hierarchy~\cite{QuantumInflation2021} and the potential generalisation of these results to larger networks in light of the inflation-NPA hierarchy.

\section{Generalising NPA hierarchy to quantum bilocal network scenario}\label{sec:TripartiteQDistribStandardNPA}
This section aims to extend the standard NPA hierarchy to the quantum bilocal scenario. Section~\ref{sec:standardNPA} provides a brief recap of the Bell scenario and the standard NPA hierarchy as preliminary background. In Section~\ref{sec:basicnotation}, the notation for noncommutative polynomials is introduced, which will be used throughout the paper. Section~\ref{sec:BilocQDistrib} delves into the bilocal scenario and its quantum models, namely the tensor product vs. commutator formulations. Lastly, Section~\ref{sec:nonSDPNPAhierarchy} presents the factorisation bilocal NPA hierarchy, a generalisation of the standard NPA hierarchy for bilocal scenarios, and demonstrates its convergence to the set of Projector Bilocal Quantum Distributions, marking the first main result of the paper.

\subsection{Tripartite Bell scenario and standard NPA hierarchy}\label{sec:standardNPA}
As a preliminary background and to fix the notation, we briefly discuss the standard NPA hierarchy introduced by Navascués, Pironio, and Acín~\cite{NPA2008} for the outer approximation of the quantum distribution obtainable in the standard tripartite Bell scenario in Figure~\ref{fig:ThreePartiesScenarios}.a. 

We first introduce the set of Tensor Tripartite Quantum Distributions, which are all the behaviours feasible in this scenario according to quantum theory:

\begin{definition}[Tensor Tripartite Quantum Distributions]\label{def:TensorTripartiteQDist}
Let $\vec{P}=\{p(abc|xyz)\}$ be a three-party probability distribution. 
We say that $\vec{P}$ is a \emph{Tensor Tripartite Quantum Distribution} iff there exist a composite Hilbert space $\cH = \cH_A \otimes \cH_B\otimes \cH_C$, a projector over a pure state $\tau$ acting over $\cH$, and PVMs $\{\hA\}$ over $\cH_A$, $\{\hB\}$ over $\cH_B$ and $\{\hC\}$ over $\cH_C$ such that 
\begin{align*}
    p(abc|xyz)=\tr{\tau}{\hA\otimes\hB\otimes\hC}.
\end{align*}
\end{definition}

The outer approximations of this set are given by the NPA hierarchy method, which we now sketch.
First, given that $\vec{P}$ is a Tensor Tripartite Quantum Distribution with $p(abc|xyz)=\tr{\tau}{\hA\otimes\hB\otimes\hC}$, one can construct the $n\geq 3$ order moment matrix $\Gamma^n$ as:
\begin{equation}
\kbordermatrix{
                          & \id & A_{a|x} & B_{b|y} & C_{c|z}              & A_{a|x}A_{a'|x'} & A_{a|x}B_{b|y}   &\cdots\\
\id                       & 1   &         &         &                      &                  & \tr{\tau}{\hA\hB}       \\
(A_{a|x})^\dagger         &     &         & \tr{\tau}{\hA \hB}             &                  &                  &      \\
(B_{b|y})^\dagger         &     &         &  \tr{\tau}{\hB}                &                  &                  &      \\
(C_{c|z})^\dagger         &     &         &                                &                  &                  &      \\
(A_{a|x}A_{a'|x'})^\dagger&     &         &\tr{\tau}{\hat{A}_{a'|x'}\hA\hB}&                  &                  &      \\
(A_{a|x}B_{b|y})^\dagger  &     &         &                                &                  &                  &      \\
\cdots                    &     &         &                                &                  &                  & 
},
\label{eq:MomentMatrix}
\end{equation}
where the rows and columns are labelled by all formal monomials in $\{A_{a|x},B_{b|y},C_{c|z}\}$ of length at most $n$, and the entries are the trace over $\tau$ of the product of the monomials labelling the entries (we only indicated some of the coefficients and omitted tensor products for readability, and used the relation $\hB^2=\hB$). $\Gamma^n$ satisfies several properties:
\begin{enumerate}[(i)]
\item $\Gamma^n_{1,1} = 1$.
\item $\Gamma^n$ is symmetric positive semidefinite.
\item $\Gamma^n$ satisfies several linear constraints, such as $\Gamma^n_{(A_{a|x})^\dagger, B_{b|y}}=\Gamma^n_{\id, A_{a|x}B_{b|y}}$.
\end{enumerate}
Moreover, if $\Gamma^n_{\id, A_{a|x}B_{b|y}C_{c|z}}=p(abc|xyz)$, we say that $\Gamma^n$ is compatible with $\vec{P}$.

Hence, if $\vec{P}$ is a Tensor Tripartite Quantum Distribution, then for all $n\geq 3$ there exists a moment matrix $\Gamma^n$ compatible with $\vec{P}$. 
This provides a hierarchy of outer approximations of the set of Tensor Tripartite Quantum Distributions.
Given $\vec{P}$ and some $n\geq 3$, if there is no moment matrix $\Gamma^n$ compatible with $\vec{P}$, then $\vec{P}$ is not a Tensor Tripartite Quantum Distributions.
Importantly, the problem of finding such $\Gamma_n$ is an SDP, which can be solved on a computer, with certificates of nonexistence.

A higher hierarchy level $n$ provides more accurate tests (as a compatible matrix $\Gamma^{n-1}$ can always be extracted from $\Gamma^n$). 
However, as this hierarchy is only sensitive to the algebraic relations between the operators (such as $\hB^2=\hB$ or $[\hA,\hC]=0$), this hierarchy of outer approximation does not obviously characterise the set of Tensor Tripartite Quantum Distributions.
Instead, it was shown to converge to the larger set of Commutator Tripartite Quantum Distributions (\cite{NPA2008}):

\begin{theorem}[Convergence of the standard NPA hierarchy]\label{thm:ConvStandardNPAHierarchy}
Suppose that $\vec{Q}=\{q(abc|xyz)\}$ is a probability distribution for Alice, Bob, and Charlie. Then $\vec{Q}$ passes all the standard NPA hierarchy tests iff $\vec{Q}$ is a Commutator Tripartite Quantum Distribution,
\end{theorem}
\noindent
where the Commutator Tripartite Quantum Distributions are defined as:
\begin{definition}[Commutator Tripartite Quantum Distributions]\label{def:CommutatorTripartiteQDist}
$\vec{Q}$ is a \emph{Commutator Tripartite Quantum Distribution} iff there exist a Hilbert space $\cH$, a projector over a pure state  $\tau$, and PVMs $\{\hA\}$, $\{\hB\}$ and $\{\hC\}$ over $\cH$ such that
\begin{enumerate}[(i)]
    \item $q(abc|xyz)=\tr{\tau}{\hA\hB\hC}$
    \item $\forall \hA, \hB, \hC,$ we have $[\hA, \hB]=[\hB, \hC]=[\hC, \hA]=0$.
\end{enumerate}
\end{definition}

Any Tensor Tripartite Quantum Distribution is trivially a Commutator Tripartite Quantum Distribution\footnote{The new operators $\hA, \hB, \hC$ are constructed as the original ones, padding with identities over the other local Hilbert spaces}. 
The converse, known as Tsirelson's problem, was first conjectured but turns out to be false (see~\cite{ji2021mip}, where it was also proven that the set of physical Tensor Tripartite Quantum Distributions is not computable).
Nevertheless, for finite-dimensional quantum systems, these two definitions are equivalent, thanks to an inductive result of Tsirelson's Theorem (see Theorem~\ref{thm:tsirelson} and Remark~\ref{rem:InductiveTsirelson} in Appendix~\ref{sec:AppendixBilocalTsirelson}).

\subsection{Noncommutative polynomial formulation}\label{sec:basicnotation}
We now present the standard \emph{canonical abstraction} procedure capable of performing this formalisation of the NPA hierarchy (based on noncommutative polynomials, see~\cite{BratelliRobinson}), introducing abstract \emph{letters} for each potential PVM elements of the parties \cite{pironio2010convergent}. 

We introduce the alphabets $\underline{a} = (a_1, \cdots, a_I)$, $\underline{b} = (b_1, \cdots, b_J)$, and $\underline{c} = (c_1, \cdots, c_K)$ respectively composed of the letters $\{\A\}$, $\{\B\}$, and $\{\C\}$, one letter for each PVM element $\{\hA\}$, $\{\hB\}$, and $\{\hC\}$ ($I,J,K$ are the respective cardinality of Alice, Bob, and Charlie PVMs operators).

The products of multiple measurement operators then correspond to combinations of letters $l_i\in \underline{a}\cup \underline{b}\cup \underline{c}$, referred to as words denoted by bold Greek letters $\bomega$, forming the set $\langle \underline{a}, \underline{b}, \underline{c} \rangle$. 
A word consisting of the letters in $\underline{a}$ is written as $\balpha$, forming the set $\monoid{a}$. Similarly, we write $\bbeta \in \monoid{b}$ associated with the alphabet $\underline{b}$ and $\bgamma \in \monoid{c}$ associated with $\underline{c}$. $1$ is the empty word with no letters.

We further impose the algebraic relations satisfied by the measurement operators on the letters. 
More precisely, as the measurement operators are projectors, measurement operators of distinct parties commute with each other and are self-adjoint, and hence we impose 
\begin{align}
    A_{a|x}A_{a'|x} = \delta_{a, a'} A_{a|x}, B_{b|y}A_{b'|y} = \delta_{b, b'} B_{b|y}, C_{c|z}C_{c'|z} = \delta_{c, c'} C_{c|z} \label{eq:projection_constr}  \\
    [a_i,b_j] = [a_i,c_k] = [b_j,c_k] = 0,  \label{eq:comm_constr} \\
    l_i^{\dagger} = l_i, (l_il_j)^{\dagger} = l_j^{\dagger}l_i^{\dagger} = l_jl_i, \label{eq:dagger_constr}
\end{align} 
where the operation $\dagger$ is an involution on the letters of a word, which stands for the conjugate transpose operation on operators.
These rules imply that the words $\balpha$, $\bbeta$, $\bgamma$ commute with each other. 
Therefore, any word $\bomega$ admits a unique minimal form $\bomega = \balpha\bbeta\bgamma$ where no letter is squared. 
Equality over words is checked in this minimal form. The length of $\bomega$, denoted by $|\bomega|$, is the number of letters in this minimal form. 

The entries in the moment matrix $\Gamma^n$ are indexed with the words $\bomega$ such that $|\bomega| \leq n$. 
Additionally, the ring of noncommutative polynomials $\mathbb{T}^{abc} = \mathbb{R} \langle \underline{a}, \underline{b}, \underline{c} \rangle$ arises; see Appendix~\ref{sec:ncPolyFormal}.

\subsection{Bilocal quantum distribution}\label{sec:BilocQDistrib}

In this subsection, we initially introduce the set of Bilocal Quantum Tensor Distributions (Definition~\ref{def:TensorBilocQDistrib}), representing the set of correlations attainable in the bilocal scenario where two independent quantum sources are distributed to three parties (see Figure~\ref{fig:ThreePartiesScenarios}b). Following this, we present the set of Bilocal Quantum Commutator Distributions (Definition~\ref{def:CommutatorBilocQDistrib}), serving as an outer approximation of the aforementioned set. These distributions will be further characterised by two hierarchies introduced in Section~\ref{sec:nonSDPNPAhierarchy} and Section~\ref{sec:ScalarExtHierarchyMainSec}.

According to quantum theory, the set of correlations that can be obtained in the bilocal scenario is
\begin{definition}[Tensor Bilocal Quantum Distributions]\label{def:TensorBilocQDistrib}
Let $\vec{P}=\{p(abc|xyz)\}$ be a three-party probability distribution. 
We say that $\vec{P}$ is a \emph{Tensor Bilocal Quantum Distribution} iff there exist a composite Hilbert space $\cH_A \otimes \cH_{B_L} \otimes \cH_{B_R} \otimes \cH_C$, two projectors over a pure state  $\rho_{AB_L}$ acting over $\cH_A \otimes \cH_{B_L}$ and $\sigma_{B_RC}$ acting over $\cH_{B_R} \otimes \cH_C$, and PVMs $\{\hA\}$ over $\cH_A$, $\{\hB\}$ over $\cH_{B_L} \otimes \cH_{B_R}$, and $\{\hC\}$ over $\cH_C$ such that 
\begin{align*}
    p(abc|xyz)=\tr{\rho_{AB_L}\otimes\sigma_{B_RC}}{\hA\otimes\hB\otimes\hC}.
\end{align*}
\end{definition}

Similar to the standard NPA hierarchy, we will only consider the abstract algebraic commutativity relations among the states and measurement operators. 
This leads to the following definition of Projector Bilocal Quantum Distributions.

\begin{definition}[Projector Bilocal Quantum Distributions]\label{def:CommutatorBilocQDistrib}
Let $\vec{Q}=\{q(abc|xyz)\}$ be a three-party probability distribution. 
We say that $\vec{Q}$ is a \emph{Projector Bilocal Quantum Distribution} iff there exist a Hilbert space $\cH$, a projector over a pure state $\tau$, two projectors (possibly not trace-class) $\rho$ and $\sigma$, and PVMs $\{\hA\}$, $\{\hB\}$, and $\{\hC\}$ over $\cH$ such that
\begin{enumerate}[(i)]
    \item $\tau = \rho\cdot\sigma =\sigma\cdot\rho$,
    \item $q(abc|xyz)=\tr{\tau}{\hA\hB\hC}$,
    \item $[\hA, \hB]=[\hB, \hC]=[ \hC,\hA]=0$,
    \item $[\hA,\sigma] = [\rho, \hC] = 0$,
\end{enumerate}
for all $\hA, \hB, \hC$. 
\end{definition}

The set of Tensor Bilocal Quantum Distributions of Definition~\ref{def:TensorBilocQDistrib} is trivially contained in the set of Projector Bilocal Quantum Distributions of Definition~\ref{def:CommutatorBilocQDistrib}. 
Due to~\cite{ji2021mip}, the converse is incorrect, and the physical set of Tensor Bilocal Quantum Distributions is not computable. 
Nonetheless, the equivalence between these two sets in the finite-dimensional setting is shown by~\cite[Corollary~9]{ligthart2023inflation} soon after the first draft of our manuscript.

\begin{remark}\label{rem:PureMixFormulationCommtutator}
In Definition~\ref{def:CommutatorBilocQDistrib}, we ask that $\tau$ is a projector over a pure state, which can be deduced from the fact that $\rho, \sigma$ are projectors and $\Tr{\tau}=1$. 
In Definition~\ref{def:TensorTripartiteQDist}, we already asked that $\rho, \sigma$, hence $\tau$, should be projectors over pure states. 
Note that with tensor products this condition is not necessary: due to purification, it is sufficient to ask that these states are density operators. Indeed, if $\vec{P}$ satisfies the conditions of Definition~\ref{def:TensorTripartiteQDist} but with $\rho, \sigma$ only positive and trace $1$, one can always enlarge the Hilbert spaces to obtain that $\vec{P}$ satisfies the conditions of Definition~\ref{def:TensorTripartiteQDist} with purifications $\rho', \sigma'$ of $\rho, \sigma$ which are now projectors over pure states.
In Appendix~\ref{sec:AppendixPureMixFormulationInequivalence}, we show that such a purification procedure is no longer possible in the commutator context.
In other words, the density operator (mixed state) formulation is \emph{not} equivalent to the pure state formulation in the commutator-based quantum model.
\end{remark}

\subsection{Factorisation bilocal NPA hierarchy}\label{sec:nonSDPNPAhierarchy}
We now introduce the factorisation bilocal NPA hierarchy (Definitions~\ref{def:FactorBMMatrix} and~\ref{def:nonSDPNPAtest}) and show its convergence to the set of Projector Bilocal Quantum Distributions (Theorem~\ref{thm:ConvnonSDPBilocNPAHierarchy}). 

We use the canonical abstraction formalism of Section~\ref{sec:basicnotation}, with alphabets $\underline{a}$, $\underline{b}$, and $\underline{c}$ in the letters $\{\A\}$, $\{\B\}$, and $\{\C\}$, which are in one-to-one correspondence with the sets of PVMs $\{\hA\}$, $\{\hB\}$, and $\{\hC\}$, respectively. Note that in the bilocal scenario, for $\halpha,\hgamma$ respectively monomials in the PVMs of Alice and Charlie, with $\tau=\rho_{AB_L}\otimes\sigma_{B_RC}$, we have nonlinear constraints of the form 
\begin{align}\label{eq:NonLinearCondition}
    \tr{\tau}{\halpha\hgamma} = \tr{\tau}{\halpha} \cdot \tr{\tau}{\hgamma}.
\end{align}
This motivates the following definition:

\begin{definition}[Factorisation Bilocal Moment Matrix]\label{def:FactorBMMatrix}
Fix $n \in \mathbb{N}$, let $\Gamma^n$ be a square matrix indexed by all words $\bomega \in \langle \underline{a}, \underline{b}, \underline{c} \rangle$ of length $|\bomega|\leq n$. We say that $\Gamma^n$ is a \emph{Factorisation Bilocal Moment Matrix of order $n$} if
\begin{enumerate}[(i)]
\item $\Gamma^n_{1,1} = 1$.
\item $\Gamma^n$ is positive.
\item it satisfies all the linear constraints 
\begin{align}\label{eq:LinearConstraints}
  \Gamma^n_{\bomega,\bnu}=\Gamma^n_{\bomega',\bnu'}
\end{align}
whenever $\bomega^{\dagger} \bnu = {\bomega'}^{\dagger} \bnu'$.
\item it satisfies the additional nonlinear factorisation constraints
\begin{align}\label{eq:nonSDPConstraints}
  \Gamma^n_{\balpha,\bgamma}=\Gamma^n_{\balpha, 1}\cdot\Gamma^n_{1,\bgamma},
\end{align}
where $\balpha \in \monoid{a}$ and $\bgamma \in \monoid{c}$.
\end{enumerate}
We further say that $\Gamma^n$ is compatible with the tripartite distribution $\vec{Q}$ iff $\Gamma^n_{1,\A\B\C}=q(abc|xyz)$.
An infinite matrix $\Gamma^\infty$ is said to be a \emph{Factorisation Bilocal Moment Matrix} iff all of its principal extracted matrices are Factorisation Bilocal Moment Matrices of some finite order. 
\end{definition}
Note that if one drops condition (iv), this definition becomes the standard NPA Moment Matrix of Section~\ref{sec:standardNPA}. 
It is straightforward to show that any Bilocal Tensor Quantum Distribution admits a Factorisation Bilocal Moment Matrix $\Gamma^n$ for all $n$. This yields the following hierarchy:

\begin{definition}[Factorisation bilocal NPA hierarchy] \label{def:nonSDPNPAtest}
Let $\vec{Q}=\{q(abc|xyz)\}$ be a tripartite probability distribution. We say that $\vec{Q}$ passes the \emph{factorisation bilocal NPA hierarchy} if for all integers $n\geq 3$, there exists a Factorisation Bilocal Moment Matrix $\Gamma^n$ of order $n$ that is compatible with $\vec{Q}$. 
\end{definition}
We are ready for our first main result.

\begin{theorem}[Convergence of the factorisation bilocal NPA hierarchy]\label{thm:ConvnonSDPBilocNPAHierarchy}
Let $\vec{Q}=\{q(abc|xyz)\}$ be a tripartite probability distribution.
Then $\vec{Q}$ passes all the factorisation bilocal NPA hierarchy tests iff $\vec{Q}$ is a Projector Bilocal Quantum Distribution.
\end{theorem}

\begin{proof}[Sketch proof]
The detailed version can be found in the Appendix~\ref{sec:nonSDPhardproof}. 
It is straightforward to show that if $\vec{Q}$ is a Projector Bilocal Quantum Distribution, then it admits a compatible Factorisation Bilocal 
Moment Matrix of any order (the nonlinear constraint (iv) is derived from the projectivity and commutativity of $\rho$ and $\sigma$ and the properties of the operator trace). 

For the converse direction, suppose there exists a compatible Factorisation Bilocal Moment Matrix $\Gamma^n$ for each $n \geq 3$. 
The first part of our proof is similar to the proof of convergence of the standard NPA hierarchy (see~\cite{NPA2008}). 
We first construct a compatible infinite-size factorisation bilocal moment matrix $\Gamma^\infty$ by extracting a convergent sequence from the set of compatible $\Gamma^n$. 
Then we construct the Hilbert space $\cH$ and the operators $\tau=\ketbra{\phi_1}{\phi_1}$, $\hA$, $\hB$, $\hC$ through a Gelfand--Naimark--Segal (GNS) representation from $\Gamma^\infty$.
The constructed operators are, respectively, a state and PVMs that automatically satisfy constraints (ii) and (iii) in Definition~\ref{def:CommutatorBilocQDistrib}.

In the last part of our proof, we construct the operators $\rho, \sigma$. 
Remarking that $\rho$ should commute with all polynomials in all operators $\hgamma$ associated with a word $\bgamma \in \monoid{c}$, we construct $\rho$ as the orthogonal projector on the subspace $V_{B_RC}$ generated by all these polynomials in $\hgamma$. 
We similarly define the operator $\sigma$ as the orthogonal projector on the subspace $V_{AB_L}$ generated by all polynomials in the operators $\halpha$. 
Note that both $V_{B_RC}$ and $V_{AB_L}$ contain $\ket{\phi_1}$. 
We orthogonally decompose them as $V_{B_RC}=\Span{\ket{\phi_1}}\stackrel{\bot}{\bigoplus} W_{B_RC}$ and $V_{AB_L}=\Span{\ket{\phi_1}}\stackrel{\bot}{\bigoplus} W_{AB_L}$.
Then, we show that the factorisation condition (iv) in the moment matrices imposes that the two complement subspaces are orthogonal to each other, that is, $W_{AB_L}\bot W_{B_RC}$.

As $\rho$ and $\sigma$ are, respectively, orthogonal projectors over $V_{B_RC}$ and $V_{AB_L}$, this directly implies that they commute and their product is the projector over $\Span{\ket{\phi_1}}$, that is, $\tau$, which proves constraints (i) of Definition~\ref{def:CommutatorBilocQDistrib} . The last constraint (iv) is proven similarly using this orthogonality relation, which concludes the proof.
\end{proof}

While Theorem~\ref{thm:ConvnonSDPBilocNPAHierarchy} gives a convergence to Projector Bilocal Quantum Distributions, as the name suggests, it does not support the application of semi-definite programming (SDP) due to the nonlinearity in the constraints (iv) of Definition~\ref{def:FactorBMMatrix}. 
In other words, while Theorem~\ref{thm:ConvnonSDPBilocNPAHierarchy} provides a theoretical characterisation, its practical usefulness to characterise bilocal quantum distributions is unclear.
We solve this problem in the next section, considering a hierarchy which additionally includes new commutative variables that can linearising these constraints.

\section{Linearising factorisation bilocal NPA hierarchy via scalar extension}\label{sec:ScalarExtHierarchyMainSec}

To address the nonlinearity issue of the factorisation bilocal NPA hierarchy, we turn to the scalar extension method that was initially proposed in~\cite{ScalarExtension2019} and later fully formalised in~\cite{klep2024state}. 

In the following, we first discuss (Section~\ref{Sec:MotivNewScalarExtensionHierarchy}) the intuition initially presented in~\cite{ScalarExtension2019} and clarify its limits. 
We show these limits motivate the need for new \emph{polarisation constraints} involving more scalar operators (see (iii$'$) and (iv$'$) bellow) that we introduce. 
We also explain why, although the existence of a solution to all NPO levels of our new hierarchy proves a quantum solution, this method is not constructive: one cannot directly extract the quantum representation from the hierarchy, but only a convex mixture of solutions.

After the motivation, we formally introduce in Section~\ref{sec:bilocal_hierarchy} our new hierarchy (Definition~\ref{def:BSE_hierarchy_Pol}), and demonstrate its convergence to the set of Projector Bilocal Quantum Distributions (Theorem~\ref{thm:ConvScalarBilocNPAHierarchy}), marking the second main result of this paper. 
Lastly, we provide the stopping criteria applicable to both hierarchies (Theorem~\ref{thm:StopCriteriaFinite}).

\subsection{Underlying intuition for the bilocal scalar extension hierarchy}\label{Sec:MotivNewScalarExtensionHierarchy}

This section discusses the limitations from the first proposed scalar extension method. 
It will lead us to the physical intuition underlying our Definition~\ref{def:BSE_hierarchy_Pol} of a bilocal scalar extension NPA hierarchy and its proof of convergence in Theorem~\ref{thm:ConvScalarBilocNPAHierarchy}.
The aim of this section is to motivate the definitions and guide readers through (or act as an alternative for) the rigorous mathematical proofs in Appendix~\ref{sec:AppendixBilocScaExtHierarchy}.

\subsubsection{Limitations to the first proposed approach to scalar extension method}  

In~\cite{ScalarExtension2019}, the authors propose to include new `scalar'  letters $\kappa_{\balpha}, \kappa_{\bgamma}$ to the standard NPA hierarchy which aim to represent the scalar operators $\hkappa_{\balpha}=\tr{\tau}{\halpha}\id, \hkappa_{\bgamma}=\tr{\tau}{\hgamma}\id$.
The authors consider a modified NPA hierarchy constructed out of the alphabet $\balpha, \bbeta, \bgamma, \kappa_{\balpha}, \kappa_{\bgamma}$, associated with a moment matrix $\Omega^n$, imposing all standard NPA constraints and the additional conditions:
\begin{enumerate}
    \item[(i)] $\kappa_{\balpha}, \kappa_{\bgamma}$ commute with all monomials,
    \item[(ii)] $\tr{\tau}{\halpha}=\tr{\tau}{\hkappa_{\balpha}}$ for all $\balpha$, and similarly for all $\bgamma$,
    \item[(iii)] $\tr{\tau}{\halpha\hgamma}=\tr{\tau}{\hkappa_{\balpha}\hkappa_{\bgamma}}$.
    \item[(iv)] $\tr{\tau}{\A\B\C}=q(abc|xyz)$.
\end{enumerate}

One can hope that, after reconstructing the operators from a converging sequence of moment matrices, the following statements on scalar extension operators hold:
\begin{enumerate}
    \item[(A1)] Condition (i) gives that both $\hkappa_{\balpha}, \hkappa_{\bgamma}\propto \id$. \emph{(Actually incorrect!)}
    \item[(A2)] By Statment (A1) and Condition (ii), it follows that $\hkappa_{\balpha} = \tr{\tau}{\halpha} \id$, and similarly for $\hkappa_{\bgamma}$.
    \item[(A3)] Statement (A2) and Condition (iii) implies Equation~\eqref{eq:NonLinearCondition}
    \begin{align*}
        \tr{\tau}{\halpha\hgamma}=\tr{\tau}{\hkappa_{\balpha}\hkappa_{\bgamma}}=\tr{\tau}{\tr{\tau}{\halpha} \id\tr{\tau}{\hgamma} \id}=\tr{\tau}{\halpha} \tr{\tau}{\hgamma}.
    \end{align*}
\end{enumerate} 
Had this worked, we would be done. Indeed, if there existed a modified NPA hierarchy $\{\Omega^n\}_n$ constructed out of the alphabet $\balpha, \bbeta, \bgamma, \kappa_{\balpha}, \kappa_{\bgamma}$, one could first construct some limiting $\Omega^\infty$ in which Equation~\eqref{eq:NonLinearCondition} holds. 
Then, one could extract from $\Omega^\infty$ (by only keeping letters $\balpha, \bbeta, \bgamma$) a factorisation bilocal NPA Hierarchy $\{\Gamma^n\}_n$.

However, this is not the case, as Statement (A1) is false. In fact, Condition (i) only leads to a block diagonal representation for $\hkappa_{\balpha},\hkappa_{\bgamma}$, where they act as scalar operators within each block only.

\subsubsection{The necessity of new constraints} 
Let us now give the correct deductions that follows from Conditions (i)-(iii). 
The standard GNS construction based on the modified NPA hierarchy provides a Hilbert space $\cH$ with direct sum decomposition $\cH=\oplus_i \cH_i$ such that $\tau=\oplus_i p_i\ketbra{\tau_i}{\tau_i}$ (with $\ket{\tau_i}$ a pure state, $p_i\geq 0$ and $\sum_i p_i =1$)\footnote{In the most generality, these Hilbert spaces in Section~\ref{Sec:MotivNewScalarExtensionHierarchy} admit instead direct integral decompositions, which introduces many technicalities. To focus on the main intuition, we oversimplify the argument.}, where
\begin{enumerate}
    \item[(B1)] Condition (i) implies that $\hkappa_{\balpha}=\oplus_i \lambda_i \id_i, \hkappa_{\bgamma}=\oplus_i \mu_i \id_i, \halpha=\oplus_i \halpha_i, \hgamma=\oplus_i \hgamma_i$.
    \item[(B2)] By Statement (B1) and Condition (ii), it follows that $\sum_i p_i\lambda_i = \sum_i p_i\sandwich{\tau_i}{\halpha_i}{\tau_i}$ and $\sum_i p_i\mu_i = \sum_i p_i\sandwich{\tau_i}{\hgamma_i}{\tau_i}$.
    \item[(B3)] Statement (B1) and Condition (iii) imply $\sum_i p_i\lambda_i\mu_i = \sum_i p_i\sandwich{\tau_i}{\halpha_i\hgamma_i}{\tau_i}$.
\end{enumerate} 
But Statement (B3) no longer implies the factorisation constraints of bilocal scenarios.
Note that Statements (B1)-(B3) reduce to (A1)-(A3) if there is exactly one unique block, that is if $\tau$ is pure. But this is generally not the case, and it will not help to consider a purification of $\tau$.

Hence, new constraints should be introduced to Conditions (i)-(iii). To this end, we will turn to the polarisation technique in~\cite{ligthart2023inflation}. 
The main intuition is to look for new constraints which would lead to the following quadratic versions of Statements (B2) and (B3):
\begin{align}
    0&=\sum_i p_i |\lambda_i - \sandwich{\tau_i}{\halpha_i}{\tau_i}|^2 \mathrm{~and~}
    0=\sum_i p_i |\mu_i - \sandwich{\tau_i}{\hgamma_i}{\tau_i}|^2,\label{eq:polarisation2}\\
    0&=\sum_i p_i |\lambda_i\mu_i - \sandwich{\tau_i}{\halpha_i\hgamma_i}{\tau_i}|^2.\label{eq:polarisation3}
\end{align}
Since $p_i \ge 0$, if the above equations could be imposed, we see that every block $i$ and the associated pure state $\ket{\tau_i}$ would correspond to a desired solution. 
We also observe that the method would become non-constructive: in the convergent limit, this hierarchy would enable the extraction of operators and the state as outlined in (B1), yet it would not directly provide access to the individual blocks $i$.

\subsubsection{Polarisation technique for a convergent scalar extension hierarchy} 
It is not immediately clear how Equations~\eqref{eq:polarisation2} and~\eqref{eq:polarisation3} can be enforced.
Indeed, to obtain Equation~\eqref{eq:polarisation3}, it is tempting to enforce $0=\sandwich{\tau}{(\halpha\hgamma-\hkappa_{\balpha}\hkappa_{\bgamma})^\dagger(\halpha\hgamma-\hkappa_{\balpha}\hkappa_{\bgamma})}{\tau_i}$ derived from Condition (iii). Then direct computation leads to
\begin{align*}
    0=\sum_i \sandwich{\tau_i}{(\halpha_i \hgamma_i)^\dagger(\halpha_i \hgamma_i)}{\tau_i} - 2 \lambda_i\mu_i \sandwich{\tau_i}{\halpha_i \hgamma_i }{\tau_i}+\lambda_i^2\mu_i^2,
\end{align*}
which does not necessarily imply Equation~\eqref{eq:polarisation3} unless $\sandwich{\tau_i}{(\halpha_i \hgamma_i)^\dagger(\halpha_i \hgamma_i)}{\tau_i}=|\sandwich{\tau_i}{\halpha\hgamma}{\tau_i}|^2$.
But this equation is generally not true. It would, however, hold if $\halpha_i \hgamma_i$ was replaced by a scalar extension operator $\hkappa_{\balpha \bgamma}$: our construction is based on this key observation.

For our new hierarchy, we need to include all scalar letters $\kappa_{\bomega}$ (not only the $\kappa_{\balpha}, \kappa_{\bgamma}$, in addition to $\alpha, \beta, \gamma$), with modified Conditions (i$'$)-(iv$'$) below. 
From now on, $\bomega$, $\bnu$ will refer to words that are composed both of noncommutative words and scalar extension letters (e.g., $\bomega=\A\B\kappa_{B_{b'|y'}\C}$). Moreover, we enforce $\kappa$-linearity in the moment matrix, that is $\kappa_{\kappa_{\omega} \nu} = \kappa_{\omega} \kappa_{\nu}$, to mimic the behaviour of $\tr{\tau}{\cdot}$.

Our construction proceeds in two steps.
First, we enforce the polarisation constraints at the level of the scalar extension letters $\kappa_{\balpha},\kappa_{\bgamma},\kappa_{\balpha\bgamma}$ only. 
This enables us to enforce the factorisation constraints for each individual block.
Second, we need to realise that the coefficients from one of these blocks (e.g., block $i=1$ only) provide a new factorization bilocal NPA matrix (where (iv$'$) ensures that the block is compatible with the targeted distribution). 

More precisely, we impose the modified conditions:
\begin{enumerate}
    \item[(i$'$)] $\kappa_{\bomega}$ commutes with all monomials,
    \item[(ii$'$)] $\tr{\tau}{{\bomega}} = \tr{\tau}{{\bnu}}$ whenever $\kappa_{\bomega} = \kappa_{\bnu}$,
    \item[(iii$'$)] $0=\tr{\tau}{(\kappa_{\balpha\bgamma}-\kappa_{\balpha}\kappa_{\bgamma})^\dagger(\kappa_{\balpha\bgamma}-\kappa_{\balpha}\kappa_{\bgamma})}$,
    \item[(iv$'$)] $0=\tr{\tau}{(\kappa_{\A\B\C}-q(abc|xyz))^\dagger(\kappa_{\A\B\C}-q(abc|xyz))}$.
\end{enumerate}
Note that Condition (ii$'$) corresponds to the linear constraints for standard NPA hierarchy (Equation~\eqref{eq:LinearConstraints}). 
The standard GNS construction leads to a direct sum decomposition $\cH = \oplus_i \cH_i$ and $\tau = \oplus p_i \ketbra{\tau_i}{\tau_i}$, where 
\begin{enumerate}
    \item[(C1)] Condition (i$'$) implies e.g. that $\hkappa_{\balpha}=\oplus_i \lambda_i \id_i, \hkappa_{\bgamma}=\oplus_i \mu_i \id_i, \hkappa_{\balpha\bgamma}=\oplus_i \nu_i \id $, and $\kappa_{\A\B\C}=\oplus_i \xi_i \id_i$,
    \item[(C2)] Condition (iii$'$) ensures the quadratic equation such as $\sum_i p_i|\lambda_i\mu_i -\nu_i|^2=0$,
    \item[(C3)] Condition (iv$'$) imposes that $\sum_i p_i | \xi_i - p(abc|xyz) |^2 = 0$.
\end{enumerate} 
Let us now pick one block, say the block $i=1$.
Then, using $\ket{\tau_1}$ we construct a matrix $\Gamma$ with the entries $\Gamma_{\bomega,\bnu}=\sandwich{\tau_1}{\kappa_{\bomega^\dagger\bnu}}{\tau_1}$. Along with Condition (ii$'$), it is straightforward\footnote{It is only straightforward for the direct sum simplification. In general, measure theory is required to deal with the subtlety.} to check that $\Gamma$ is a Factorisation Bilocal Moment Matrix associated with distribution $\Vec{Q}$. In particular, the factorisation condition Equation~\eqref{eq:nonSDPConstraints} corresponds to $\lambda_1\mu_1 -\nu_1=0$ and the compatibility condition follows from $ \xi_1 - p(abc|xyz) = 0$.

\subsection{The bilocal scalar extension hierarchy} \label{sec:bilocal_hierarchy}
To obtain a convergent version of the scalar extension hierarchy, we use the approach of~\cite{klep2024state}, where an NPA-like hierarchy \cite{pironio2010convergent}, i.e.~a non-commutative version of the Lasserre hierarchy \cite{lasserre2001global}, was used to solve the state polynomial optimisation problem.

To this end, we introduce the (infinite-size) alphabet $\mathcal{N} = \langle \underline{\kappa_{ABC}}, \underline{a}, \underline{b}, \underline{c} \rangle$, where $\underline{\kappa_{ABC}}$ is the alphabet composed of commuting letters $\kappa_{\bomega}$ for all $\bomega \in \langle \underline{a}, \underline{b}, \underline{c} \rangle$ and the empty word corresponds to $\id$.
The letters $\kappa_{\bomega}$ represent the scalar extensions, that is, they should mimic the behaviour of $\tr{\tau}{\bomega}$.
Indexing the rows and columns of the moment matrix $\Omega$ by $\bmu, \bnu \in \mathcal{N}$, we let the entry $(\bmu, \bnu)$ be given by a complex number
\begin{align}\label{eq:AssociatedFunctional}
    y_{\bmu^\dagger \bnu} = L_{\Omega}(\kappa_{\bmu^\dagger \bnu}) = y_{\kappa_{\bmu^\dagger \bnu}},
\end{align}
where $L_\Omega \circ \kappa$ can be thought of as a functional acting on the words in the variables $\mathcal{N}$, and $\kappa_{\bmu \kappa_{\bnu}} = \kappa_{\bmu}\kappa_{\bnu}$.
This leads to a matrix of the form
\begin{equation}
\Omega = \kbordermatrix{
    & \id & \A & \cdots & \kappa_{C_{c|z}} & \kappa_{A_{a|x}A_{a'|x'}} &\cdots\\
\id & 1 & & & &  \\
\A^\dagger & & & & & \\
A_{a'|x'}^\dagger & & y_{\hkappa_{A_{a'|x'}^\dagger \A}} & & & \\
\vdots & & & & & \\
\kappa_{A_{a|x}}^\dagger & & y_{\hkappa_{\A}^2} & & y_{\hkappa_{\A} \hkappa_{\C}} & \\
\kappa_{B_{b|y}}^\dagger & & & & & \\
\kappa_{C_{c|z}}^\dagger & & y_{\hkappa_{\A} \hkappa_{\C}} & & y_{\hkappa_{\C}^2} & \\
\kappa_{A_{a|x}A_{a'|x'}}^\dagger & & & & & \\
(\kappa_{A_{a|x}}\kappa_{A_{a'|x'}})^\dagger & & & & y_{\hkappa_{{A}_{a'|x'}} \hkappa_{\A} \hkappa_{\C}} & \\
\vdots & & & & & \\
(\kappa_{\A\C})^\dagger & y_{\hkappa_{\A\C}} & & & & \\
\vdots & & & & &                   
},
\label{eq:ScaExtMomentMatrix}
\end{equation}
where $y_{\bmu^\dagger \bnu} = y_{\bmu^{\prime\dagger}\bnu'}$ if $\bmu^\dagger \bnu = \bmu^{\prime\dagger}\bnu'$.
Assuming a valid assignment of complex numbers to $y$ exists, we call $y$ a moment representation.

Define the length of $\bmu = \left( \prod_i \kappa_{\bomega_i} \right) \bomega_0$ as $|\bmu| = |\bomega_0| + \sum_i |\bomega_i|$, and let $\mathcal{N}^n = \{\bmu \in \mathcal{N}: |\bmu| \leq n\}$.
We can construct a matrix $\Omega^n$ as the principal submatrix of $\Omega$ indexed in $\mathcal{N}^n$.

In addition, we define a \emph{localising matrix} of a polynomial $q = \sum_i q_i \bmu_i$ as the matrix $\Omega(q)$ where the entry $(\bomega, \bnu)$ is given by 
\begin{align}\label{eq:AssociatedFunctionalLocalising}
    \sum_i q_i y_{\bomega^\dagger \bmu_i \bnu} = \sum_i q_i L_\Omega (\kappa_{\bomega^\dagger \bmu_i \bnu}) .
\end{align}
Similarly, $\Omega^n(q)$ is defined as the submatrix indexed in $\mathcal{N}^{n-d_q}$, where $d_q = \lceil |q|/2 \rceil$.

Let $\mathcal{C}$ be the set of expressions that correspond to PVM properties, commutativity, and factorisation constraints Condition (iii$'$) (see Equation~\eqref{eq:FactorisationConstraintSetDef}).
Also, let $\mathcal{Q}$ be the set of compatibility requirements corresponding to Condition (iv$'$) (see Equation~\eqref{eq:fCompatibilityConstraint}).

We are now ready to define two closely related hierarchies of semidefinite programs for the bilocal compatibility problem.

\begin{definition}[Bilocal scalar extension hierarchy (Klep \emph{et al.})] \label{def:BSE_hierarchy}
For a distribution $\vec{Q}$, the $n$-th level of the hierarchy is given by the following semidefinite program. 

Find $\Omega^n$ such that
\begin{enumerate}[(i)]
\item $\Omega^n_{1,1} = 1$.
\item $\Omega^n$ is positive semidefnite.
\item it satisfies all the linear constraints 
\begin{align}
  \Omega^n_{\bomega,\bnu}=\Omega^n_{\bomega',\bnu'}
\end{align}
whenever $\kappa_{\bomega^{\dagger} \bnu} = \kappa_{{\bomega'}^{\dagger} \bnu'}$.
\item it satisfies the additional constraints
\begin{align} \label{eq:SDP_constraints}
    &\Omega^n(\pm c) \geq 0 \qquad \forall c \in \mathcal{C},
\end{align}
\item it satisfies the compatibility constraints
\begin{align} \label{eq:SDP_compatibility}
    &\Omega^n(\pm f) \geq 0 \qquad \forall f \in \mathcal{Q}.
\end{align}
\end{enumerate}
\end{definition}

\begin{definition}[Bilocal scalar extension hierarchy (polarisation)] \label{def:BSE_hierarchy_Pol}
For a distribution $\vec{Q}$, the $n$-th level of the hierarchy is given by the following semidefinite program. 

Find $\Omega^n$ such that
\begin{enumerate}[(i)]
\item $\Omega^n_{1,1} = 1$.
\item $\Omega^n$ is positive semidefnite.
\item $\Omega^n\left(N - \sum_{\bmu \in \mathcal{N}^n} \bmu^\dagger \bmu \right)$ is positive semidefinite, where $N= \left|\mathcal{N}^n\right|$.
\item it satisfies all the linear constraints 
\begin{align}
  \Omega^n_{\bomega,\bnu}=\Omega^n_{\bomega',\bnu'}
\end{align}
whenever $\kappa_{\bomega^{\dagger} \bnu} = \kappa_{{\bomega'}^{\dagger} \bnu'}$.
\item it satisfies the additional constraints
\begin{align} \label{eq:SDP_constraints_pol}
    &L_\Omega((\kappa_c)^2) = 0 \qquad \forall c \in \mathcal{C},
\end{align}
\item it satisfies the compatibility constraints
\begin{align} \label{eq:SDP_compatibility_pol}
    &L_\Omega((\kappa_f)^2) = 0 \qquad \forall f \in \mathcal{Q}.
\end{align}
\end{enumerate}
\end{definition}

Both hierarchies are quite similar, with the main difference being how some of the polynomial equality constraints are imposed.
In Definition~\ref{def:BSE_hierarchy}, such constraints are imposed through the use of localising matrices, which algebraically constrain the polynomials to be positive.
In the hierarchy of Definition~\ref{def:BSE_hierarchy_Pol}, such constraints are imposed through the quadratic formulation inspired by the \emph{polarisation technique}~\cite{ligthart2023inflation}, which only requires a linear constraint on the moment matrix.
The requirement (iii) in Definition~\ref{def:BSE_hierarchy_Pol} ensures that each of the variables $\kappa_{\bomega}$ is bounded.
In addition, the expressions $c$ and $f$ are squared in Equations~\eqref{eq:SDP_constraints_pol} and \eqref{eq:SDP_compatibility_pol} due to technicalities in the proof of convergence for the polarisation hierarchy.
We will elaborate on these details in Appendix~\ref{sec:AppendixBilocScaExtHierarchy}.

\begin{theorem}[Convergence of the scalar extension bilocal hierarchy] \label{thm:ConvScalarBilocNPAHierarchy}
Suppose that $\vec{Q}=\{q(abc|xyz)\}$ is a tripartite probability distribution. 
Then $\vec{Q}$ passes all the scalar extension bilocal hierarchy tests of Definition~\ref{def:BSE_hierarchy} and Definition~\ref{def:BSE_hierarchy_Pol} if and only if $\vec{Q}$ is a Projector Bilocal Quantum Distribution.
\end{theorem}
\begin{proof}
    See Appendix~\ref{sec:AppendixBilocScaExtHierarchy}.
\end{proof}

\subsection{Stopping criteria and finite-dimensionality}\label{sec:StopCriteria}
So far, we have shown that both the factorisation and scalar extension bilocal NPA hierarchies characterise the set of Projector Bilocal Quantum correlations (Definition~\ref{def:CommutatorBilocQDistrib}) in the asymptotic limit. 
The reconstructed Hilbert space $\cH$ is a priori infinite-dimensional, and there is no promise for the speed of convergence of these two hierarchies.
In this section, we give a stopping criterion for our hierarchies, and show that it is associated with a finite-dimensional reconstructed Hilbert space $\cH$.
Our stopping criterion is similar to the one for the NPA hierarchy given in~\cite{NPA2008}.

\begin{definition}[Rank loop]\label{def:rankloop}
For any $N \geq 3$, the Factorisation Moment Matrix $\Gamma^N$ is said to have a \emph{rank loop} if its extracted submatrix $\Gamma^{N-1}$ satisfies
\begin{align}
    \mathrm{rank}(\Gamma^{N-1}) = \mathrm{rank}(\Gamma^{N}).
\end{align}
Similarly, the Scalar Extension Bilocal Matrix $\Omega^N$ has a rank loop if its extracted submatrix $\Omega^{N-1}$ satisfies
\begin{align}
    \mathrm{rank}(\Omega^{N-1}) = \mathrm{rank}(\Omega^{N}).
\end{align}
\end{definition}

\begin{theorem}[Stopping criteria]\label{thm:StopCriteriaFinite}
Let $\vec{Q}=\{q(abc|xyz)\}$ be a tripartite probability distribution. The following are equivalent:
\begin{enumerate}[(i)]
    \item $\vec{Q}$ admits a finite-dimensional Projector Bilocal Quantum representation.
    \item $\vec{Q}$ admits a Factorisation Bilocal Moment Matrix $\Gamma^N$ which has a rank loop for some $N \in \mathbb{N}$.
    \item $\vec{Q}$ admits a Scalar Extension Bilocal Moment Matrix $\Omega^N$ which has a rank loop for some $N \in \mathbb{N}$, and the associated GNS-representation $(\cH, \pi, \ket{\phi_1})$ satisfies $\pi(T_{abc}) = \mathbb{R}\id_{\cH}$.
    \item $\vec{Q}$ admits a finite-dimensional Tensor Bilocal Quantum representation.
\end{enumerate}
\end{theorem}
\begin{proof}
This follows from essentially the same argument as in \cite[Theorem~10]{NPA2008}, \cite[Theorem~2]{pironio2010convergent}, \cite[Theorem~1.69]{burgdorf2016optimization}. We thus omit the details, and comment that the GNS-representation in (iii) is well-defined thanks to the rank loop condition.

Item (iv) was left as an open question in our first preprint~\cite{renou2022bilocalv1}; nonetheless, this is answered positively soon after by~\cite[Corollary~9]{ligthart2023inflation}.
\end{proof}

\section{Inflation-NPA hierarchy and generalisation to other Networks}\label{sec:OtherNoLoopNetworksInflationNPA}

In this section, we first introduce a third hierarchy, which can be used to approximate quantum distributions in the bilocal scenario: the inflation-NPA hierarchy of~\cite{QuantumInflation2021}. 
This hierarchy was obtained by combining the inflation technique~\cite{Inflation,NavascuesCVStandardInflation,VictorCVStandardInflation} to the NPA method~\cite{NPA2008}.
It is worth noting that an alternative Schmidt-rank-dependent version was proposed and proven to be convergent by~\cite{ligthart2023convergent}, but the Schmidt-rank-dependence makes their construction an inner approximation of outer approximation.
We prove in Theorem~\ref{thm:InflationNPAIsTighter} that any distribution compatible with this inflation-NPA method is a Projector Bilocal Quantum Distribution (i.e. it is compatible with the two hierarchies introduced in Section~\ref{sec:nonSDPNPAhierarchy} and~\ref{sec:bilocal_hierarchy}).
The converse question on the equality of sets of correlations was left open in the first draft of our manuscript, and is then proven by~\cite[Theorem~16]{ligthart2023inflation}.

Then, we introduce more general network scenarios (see Figure~\ref{fig:GeneralisedScenarios}) and discuss the generalisations of factorisation, scalar extension, and inflation-NPA hierarchies to these scenarios. 
We remark that for the triangle scenario, the factorisation and scalar extension hierarchies are trivial (they correspond to the standard NPA hierarchy), while the NPA inflation hierarchy is nontrivial, showing a clear gap in the two sets which can be identified by these hierarchies in the context of generic network scenarios.

\subsection{Inflation-NPA hierarchy}\label{sec:InflationNPAHierarchyBilocalScenario}

Consider a Tensor Bilocal Quantum Distribution $\vec{P}$ (Definition~\ref{def:TensorBilocQDistrib}). 
Following~\cite{QuantumInflation2021}, we now show that $\vec{P}$ can be associated with a new two-parameter sequence of increasing size moment matrices $\Xi^{n,m}$ satisfying some additional symmetry constraints.
This can be used to construct a new hierarchy capable of showing that some given distribution cannot be a Tensor Bilocal Quantum Distribution, called the inflation-NPA hierarchy (also called quantum inflation hierarchy).
We then prove that this method is more (or equally) precise than the two factorisation and scalar extension hierarchies. 

\subsubsection{Inflation-NPA Hierarchy for the Bilocal scenario}
Let $\vec{P}$ be a Tensor Bilocal Quantum Distribution:
there exist a composite Hilbert space $\cH_A \otimes \cH_{B_L} \otimes \cH_{B_R} \otimes \cH_C$ and operators such that 
\begin{align*}
    p(abc|xyz)=\tr{\rho_{AB_L}\otimes\sigma_{B_RC}}{\hA\otimes\hB\otimes\hC}.
\end{align*}

Consider two integers $n\geq 3, m\geq 1$.  
Introducing $m$ copies of the Hilbert spaces and copies of the operators in these Hilbert spaces, we now construct a new inflation-NPA moment matrix $\Xi^{n,m}$ that satisfies some SDP constraints, notably the linear symmetry conditions corresponding to the fact that the copies should be interchangeable. 

We introduce $m$ independent copies of each Hilbert spaces, called $\cH_A^i$, $\cH_{B_L}^j$, $\cH_{B_R}^k$, and $\cH_C^l$.
We consider copies of all PVMs measurement operators in every possible (combination of) Hilbert space. 
More precisely, we consider $m$ operators $\{\hAinf{i}\}$ (resp. $\{\hCinf{l}\}$), copies of $\hA$ (resp. $\hC$) acting over the Hilbert spaces $\cH_A^i$ (resp. $\cH_C^l$), and $m^2$ operators $\{\hBinf{j}{k}\}$ copies of $\hB$ acting on Hilbert spaces $\cH_{B_L}^j\otimes\cH_{B_R}^k$.
We also consider $m$ independent copies of $\rho_{AB_L}$ and $\sigma_{B_RC}$, denoted by $\rho_{AB_L}^i$ acting on $\{\cH_A^i\otimes\cH_{B_L}^i\}$ and $\sigma_{B_RC}^l$ acting on $\{\cH_{B_R}^l\otimes\cH_C^l\}$.
We introduce $\tau_m$ the global state tensor product of all states, $\cH^i=\cH_A^i \otimes \cH_{B_L}^i \otimes \cH_{B_R}^i \otimes \cH_C^i$ the $i$th diagonal spaces, $\hBinf{i}{i}$ the $i$th diagonal PVM and $\cH$ the global Hilbert space tensor product of all Hilbert spaces.

In the following, to simplify the notation, we leave the inputs and the outputs of the operators implicit (e.g. $A^{i}$ should be thought of as $\Ainf{i}$).
As in the standard NPA case, we also identify an operator over a local Hilbert space with the operator acting over the full Hilbert space padding with identity operators, which allows to replace all tensor products with simple products\footnote{e.g., $A^{i}$ acting over $\cH_A^i$ is identified with $A^{i}\otimes \id$ acting over $\cH$ where the identity acts over the tensor product of all Hilbert spaces except $\cH_A^i$. 
Then, for instance, the operator $A^{i}\otimes A^{i'}\otimes C^l$ is identified with $A^{i}\cdot A^{i'}\cdot C^l$}.

\begin{figure}
    \centering
    \includegraphics[width=\textwidth]{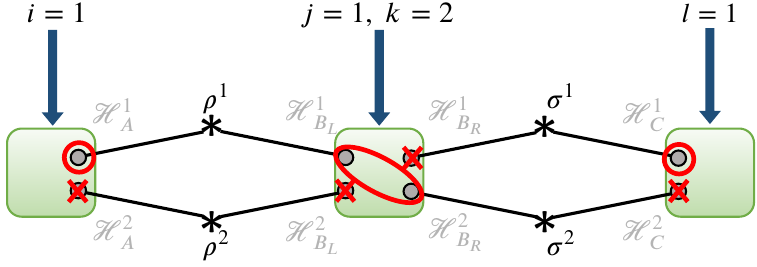}
    \caption{Physical scenario corresponding to the inflation technique (to simplify, inputs $x,y,z$ and outputs $a,b,c$ are omitted). 
    In case a distribution $\vec{P}$ is obtain from the bilocal scenario, one can consider the thought experiment in which the sources are duplicated, and each party is given \emph{additional} inputs $i$; $j,k$; and $l$. 
    The parties apply the measurement specified in the original scenario to the copies of the states given by these extra inputs. 
    Here, we represented the case where $A$ performs here measurement $A_{a|x}$ over her share of state $\rho^{1}$, corresponding to the new operator $A^1_{a|x}$.
    Similarly, here Bob measures operator $B^{12}_{b|y}$ over his shares of states $\rho^{1},\sigma^2$ and Charlie measures $C^{1}_{c|z}$ over his share of state $\sigma^2$. Note that in that case Charlie measures a state not considered by the other parties: his behavior factorises from the rest, which is imposed by the condition (iv) of the main text.
    Moreover, as the states and PVMs are copies of each other, the behavior should be invariant under $S_2\times S_2$ permutation group, as specified by condition (iii) of the main text.
    }
    \label{fig:InflationBilocalScenario}
\end{figure}

Physically, $\Xi^{n,m}$ corresponds to the ``standard'' moment matrix of a new hypothetical scenario in which the parties now have access to several copies of the sources and decide which one to measure according to the additional input $i$ for $A$, $(j,k)$ for $B$ and $l$ for $C$  (see Figure~\ref{fig:InflationBilocalScenario}), to which we additionally impose symmetry conditions deduced from the interchangeability of the copies of the states.
More precisely, $\Xi^{n,m}$ is indexed by words $\bomega,\bnu$ (of size $\leq n$) composed of letters $\{A^{i}\}, \{B^{j,k}\}, \{C^{l}\}$, and its coefficients are $\Xi^{n,m}_{\bomega,\bnu}=\tr{\tau_m}{\homega^\dagger\hnu}$. It satisfies four types of conditions:
\begin{enumerate}[(i)]
\item $\Xi^{n,m}_{1,1}=1$ and $\Xi^{n,m}$ is symmetric positive.
\item commutation of operators in different Hilbert spaces implies that $\Xi^{n,m}$ satisfies several linear constraints (e.g. $A^{1}$ commutes with $A^{2}$ and $B^{1,1}$ commutes with $B^{2,3}$ but not with $B^{1,2}$).
\item the invariance of $\tau_m$ under permutations $\theta,\theta'\in S_m\times S_m$ of the copies of $\rho_{AB_L}$ and $\sigma_{B_RC}$ implies new linear constraints on $\Xi^{n,m}$. 
Acting with $\theta$ (resp. $\theta'$) over $\tau_m$ is equivalent to permuting the index $i,j$ (resp. $k,l$) in $A^{i}, B^{j,k}, C^{l}$, hence, for instance, 
$\Xi^{n,m}_{1,A^{0}A^{1}B^{0,1}C^{1}}=\Xi^{n,m}_{1,A^{1}A^{0}B^{1,2}C^{2}}$ (here $\theta = (0~1)$ and $\theta'= (1~2)$).
\item Moreover, note that when only diagonal operators $A^{i}, B^{i,i}, C^{i}$ are involved, the trace over $\tau_m$ factorises as a product of $m$ traces over each diagonal space $\cH^i$. 
Hence $\Xi^{n,m}$ satisfies $\Xi^{n,m}_{1,\prod^m_{i=1} A^{i} B^{i,i} C^{i}}=\prod^m_{i=1} p(a^i b^{i,i} c^i \mid x^i y^{i,i} c^i)$ with $\vec{P}$\footnote{Physically, this terms corresponds to the parties $m$ parallel implementations of the initial experiment, one for each ($\rho^i, \sigma^i$).}.
We say that $\Xi^{n,m}$ is compatible with $\vec{P}$.
\end{enumerate}

As for the NPA and our previous hierarchies, the existence of a hierarchy of inflation-NPA moment matrices $\Xi^{n,m}$ compatible with any tensor bilocal distribution $\vec{P}$ can be used as a way to test whether some given distribution $\vec{Q}$ has such a model. If for some $n,m$, there does not exist any $\Xi^{n,m}$ compatible with $\vec{Q}$ (this is an SDP problem), then $\vec{Q}$ is not a Tensor Bilocal Quantum Distribution.

\begin{remark}\label{rem:inflationFactorisationDeFinetti}
Note that a moment matrix constructed from a Tensor Bilocal Quantum Distribution $\Xi^{n,m}$ satisfies additional nonlinear factorisation constraints. 
More precisely, for some words $\balpha^i, \bbeta^{ii}, \bgamma^{i}$ in the letters $\{A^{i}\}, \{ B^{i,i}\}$ and $\{C^{i}\}$,  we have the factorisation constraint $\Xi^{n,m}_{1,\prod_i \balpha^i, \bbeta^{ii}, \bgamma^{i}}=\prod_i\Xi^{n,m}_{1, \balpha^i, \bbeta^{ii}, \bgamma^{i}}$.
This condition can be imposed in the tests, but it is not practical, as it results in non-SDP problems. 
Nonetheless, the de Finetti Theorem~\cite{ligthart2023convergent} ensures the conditions are imposed anyway as $m, n$ go to infinity, and one can also use scalar extension hierarchy to encode this constraint easily.
\end{remark}

\subsubsection{Projector quantum distribution and inflation-NPA hierarchy}\label{sec:MainCommutatorInflationCompare}
We now prove that if $\vec{Q}$ is detected as infeasible by the scalar extension hierarchy (that is, if $\vec{Q}$ is not a Projector Bilocal Quantum Distribution), then the inflation-NPA hierarchy will also detect it. 

\begin{theorem}\label{thm:InflationNPAIsTighter}
Consider a distribution $\vec{Q}=\{q(abc|xyz)\}$.
The existence of $\Xi^{n,m}$ compatible with $\vec{Q}$, where $m$ is sufficiently large, implies the existence of a scalar extension moment matrix $\Omega^{n}$.
Hence, if $\vec{Q}$ passes the bilocal inflation-NPA hierarchy, then it also passes the factorisation and scalar extension bilocal NPA hierarchies, that is, it is a commutator bilocal distribution.
\end{theorem}
\begin{proof}
We first identify the PVMs of the first diagonal space with $\{\hA\}$, $\{\hB\}$, $\{\hC\}$ and suppose $m$ to be of order $nd^n$, where $d$ is the number of operators in $\{\hA, \hB, \hC\}$. For each distinct word $\bomega$, we assign a distinct index $i \geq 2$ and identify the associated scalar extension with $\bomega^i$ from the $i$-th diagonal space. Additional index assignment is needed when considering the powers of scalar extension words. Then, it is straightforward to construct a Scalar Extension Bilocal Moment Matrix by utilising the diagonal factorisation and symmetry of the inflation-NPA hierarchy. See Appendix~\ref{sec:formalInflationNPA} for more details.
\end{proof}
The converse of Theorem~\ref{thm:InflationNPAIsTighter} was left open in the first version of our manuscript, but was later shown to be true thanks to~\cite[Theorem~16]{ligthart2023inflation}.

It should be noted that, from the proof, the construction of a Scalar Extension Bilocal Moment Matrix of order $n$ requires an inflation level $m$ of order $\mathscr{O}(nd^n)$, representing a growth rate that exceeds exponential. But this serves only as a theoretical argument without direct implication to the computational efficiency.

\subsection{Generalisation to other network scenarios}\label{sec:OtherNetworkGeneralisation}

\begin{figure}
    \centering
    \includegraphics[width=\textwidth]{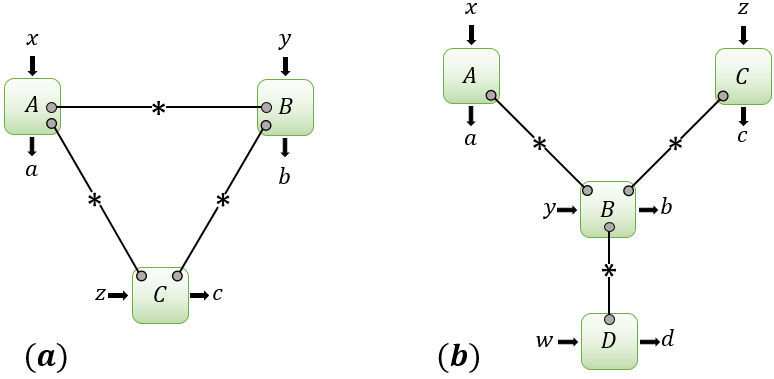}
    \caption{Some generalised network scenario: the triangle scenario (a) and the four party star network (b). }
    \label{fig:GeneralisedScenarios}
\end{figure}

Up to now, we have considered quantum correlations which can be obtained in the three-party Bell and Bilocal scenarios as in Figure~\ref{fig:ThreePartiesScenarios}.
Quantum theory can also be used to define what correlations can be obtained in more complex scenarios, such as the four party star network or the triangle scenario of Figure~\ref{fig:GeneralisedScenarios}, and all the three hierarchies we introduced can be generalised to these scenarios. 

Let us now discuss the generalisation of our results to these generalised scenarios.
We will first remark that the generalisation of our factorisation and scalar extension hierarchies is sometimes strictly less efficient to the inflation-NPA hierarchy, providing a clear gap in the case of the triangle scenario. 
Then, we evoke the potential generalisation of our main result to the star network scenario.

Let us first concentrate on the triangle scenario of Figure~\ref{fig:GeneralisedScenarios}c, where the correlations predicted by quantum theory are: 

\begin{definition}[Tensor Triangle Quantum Distributions]\label{def:TensorTriangleQDistrib}
Let $\vec{P}=\{p(abc|xyz)\}$ be a three-party probability distribution. 
We say that $\vec{P}$ is a \emph{Tensor Triangle Quantum Distribution} iff there exist a composite Hilbert space $\cH_{A_L} \otimes \cH_{A_R} \otimes \cH_{B_L} \otimes \cH_{B_R} \otimes \cH_{C_L} \otimes \cH_{C_R}$, projector over pure states $\rho_{A_RB_L}$ acting over $\cH_{A_R} \otimes \cH_{B_L}$, $\sigma_{B_RC_L}$ acting over $\cH_{B_R} \otimes \cH_{C_L}$, $\pi_{C_RA_L}$ acting over $\cH_{C_R} \otimes \cH_{A_L}$, and some PVMs $\{\hA\}$ over $\cH_{A_L} \otimes \cH_{A_R}$, $\{\hB\}$ over $\cH_{B_L} \otimes \cH_{B_R}$, and $\{\hC\}$ over $\cH_{C_L} \otimes \cH_{C_R}$ such that 
\begin{align*}
    p(abc|xyz)=\tr{\rho_{A_RB_L}\otimes\sigma_{B_RC_L}\otimes\pi_{C_RA_L}}{\hA\otimes\hB\otimes\hC}.
\end{align*}
\end{definition}

The generalisation of our factorisation or scalar extension hierarchies to this new set of correlation is trivial. 
Indeed, no factorisation condition can be imposed anymore, as any pair of parties is sharing a quantum source.
Hence, they consist of testing the existence of a hierarchy of moment matrices subject to the same conditions as in the standard NPA hierarchy.
Then, they are unable to distinguish triangle quantum correlations from the correlation that can be obtained in the standard tripartite Bell scenario of Figure~\ref{fig:ThreePartiesScenarios}a. 

In particular, the shared random bit distribution (in which the parties coordinately output 0 or 1 with probability 1/2, there are no inputs) is clearly feasible in the standard tripartite Bell scenario and hence cannot be detected as infeasible by the generalised factorisation or scalar extension hierarchies.
However, we prove in Appendix~\ref{sec:AppendixSharedRandBitInflationNPA} that the inflation-NPA shows that it is infeasible in the triangle scenario.
This proves that in the context of generic networks, the two generalised factorisation and scalar extension hierarchies can detect strictly fewer distributions as infeasible compared to the inflation-NPA hierarchy.

We also discuss in Appendix~\ref{AppendixStarNetwork} a potential generalisation of our main results to the case of the star network scenario of Figure~\ref{fig:GeneralisedScenarios}a.

\begin{ack}
We thank Antonio Acín, David Gross, Igor Klep, Victor Magron and Miguel Navascués for discussions.
\end{ack}

\begin{funding}
M.-O.R. is supported by the grant PCI2021-122022-2B financed by\\ MCIN/AEI/10.13039/501100011033 and by the European Union NextGenerationEU/PRTR.
This work was partially funded by the Government of Spain (FIS2020-TRANQI and Severo Ochoa CEX2019-000910-S), Fundacio Cellex, Fundacio Mir-Puig, Generalitat de Catalunya (CERCA, AGAUR SGR 1381 and QuantumCAT), the ERC AdG CERQUTE and the AXA Chair in Quantum Information Science.
X.X. and M.-O.R. acknowledge funding from the INRIA and the CIEDS in the Action Exploratoire project DEPARTURE.
X.X. and M.-O.R. acknowledge funding from the ANR through the JCJC grant LINKS (ANR-23-CE47-0003).
X.X. and M.-O.R. acknowledge funding from the European Union's Horizon 2020 Research and Innovation Programme under QuantERA Grant Agreement no. 731473 and 101017733.
\end{funding}

\paragraph{Data availability}
Data sharing is not applicable to this article as no new data were created or analysed in this study.

\clearpage
\appendix

\section{A formal construction of noncommutative polynomials}

In this section, we provide a more mathematically rigorous construction of noncommutative polynomials such as $\mathbb{T}^{abc}$ and $\mathbb{T}^{abc}_{\mathrm{SE}}$. We also discuss the relation between the moment matrix and the linear functional. 
Lastly, we discuss the difference between our work and the existing results.

\subsection{Noncommutative polynomials generated by finite alphabets}\label{sec:ncPolyFormal}
Fixed some $I \in \mathbb{N}$, we consider an alphabet $\underline{a} = (a_1, \cdots, a_I)$ consisting of letters $a_i$ for $i=1, \cdots, I$. The reduced monoid generated by the alphabet $\underline{a}$ is denoted by $\monoid{a}$ with \emph{reducibility condition}
\begin{align*}
    a_i^2 = a_i
\end{align*}
for each letter $a_i$. One can introduce the additional orthogonality condition for selected pairs of letters analogous to Equation~\eqref{eq:dagger_constr}. All elements of $\monoid{a}$ are called \emph{words} and are assumed to be in minimal forms; that is, there is no word with two equal consecutive letters. The \emph{empty word} is denoted by $1$. The \emph{length} of the word $\bomega$, denoted by $|\bomega|$, is the number of letters in the word. The ring of \emph{noncommutative polynomials} associated with $\underline{a}$ is then written as $\mathbb{T}^a = \mathbb{R}\monoid{a}$.  

We now endow a $^*$-structure on $\mathbb{T}^a$ by introducing an involution\footnote{In mathematical literature, this involution is usually denoted by $*$. Instead, we use $\dagger$ to keep the notation consistent with physics literature.} such that $\forall i, j = 1, \cdots, I$,
\begin{align*}
    \mathbb{R}^{\dagger} = \mathbb{R}, a_i^{\dagger} = a_i, \\
    (a_ia_j)^{\dagger} = a_j^{\dagger}a_i^{\dagger} = a_ja_i.
\end{align*}
The subset of polynomials that are invariant under involution is called \emph{symmetric} and is denoted by $\mathrm{Sym}\mathbb{T}^a = \{\bsym{p}\in\mathbb{T}^a \mid \bsym{p}^{\dagger}=\bsym{p}\}$. It follows that $\mathbb{T}^a = \mathbb{R}\monoid{a}$ is a $^*$-algebra.

Fix a Hilbert space $\cH$ with a set of variables $\{\underline{A}=(A_1,\dots,A_I) \mid A_i \in \mathcal{B}(\cH)\}$. Define \emph{evaluation} for any $\bsym{p} \in \mathbb{T}^{a}$ to be the matrix polynomial $\bsym{p}(\underline{A})$, by replacing each letter $a_i$ with the operator $A_i$, where $\halpha$ is the operator corresponding to the word $\balpha$.

\begin{example}
Consider $\underline{a} = (a_1, a_2, a_3)$, and $\bsym{p} = a_1a_2a_3 + a_2^2 \in \mathbb{T}^{a}$. For evaluation, we may choose the Hilbert space of a qubit with the variables being Pauli matrices $\Vec{\sigma} = (\sigma_x, \sigma_y, \sigma_z)$. Whence,
\begin{align*}
    \bsym{p}(\Vec{\sigma}) &= \sigma_x\sigma_y\sigma_z + \sigma_y^2 \\
\end{align*}
\end{example}

We can easily generalise the above definition to the case of having multiple alphabets. For our purposes, it suffices to consider only two more alphabets $\underline{b} = (b_1, \cdots, b_J)$ and $\underline{c} = (c_1, \cdots, c_K)$ for some $J,K \in \mathbb{N}$. In this case, the elements of the reduced monoid $\langle \underline{a}, \underline{b}, \underline{c} \rangle$ are again called \emph{words}, assuming the minimal forms due to the additional reducibility conditions
\begin{align*}
    b_j^2 = b_j, c_k^2 = c_k
\end{align*}
for each letter $b_j$ and $c_i$. Moreover, we require all letters from different alphabets to commute with each other, that is,
\begin{align*}
    [a_i,b_j] = [a_i,c_k] = [b_j,c_k] = 0,
\end{align*}
for any $i=1, \cdots, I$, $j=1, \cdots, J$, and $k=1, \cdots, K$. The ring of \emph{bilocal noncommutative polynomials} is defined as $\mathbb{T}^{abc} = \mathbb{R} \langle \underline{a}, \underline{b}, \underline{c} \rangle$ with the same $^*$-structure, where the superscripts of $\mathbb{T}$ are used to refer to the alphabets used. Given any Hilbert space $\cH$, the evaluation can be generalised by introducing more variables $\{\underline{B}=(B_1,\dots,B_J) \mid B_j \in \mathcal{B}(\cH)\}$ and $\{\underline{C}=(C_1,\dots,C_L) \mid C_k \in \mathcal{B}(\cH)\}$.
We remark that both $\mathbb{T}^{a}$ and $\mathbb{T}^{abc}$ can be made into topological vector spaces.

As a convention, we use bold Greek letters to refer to words, while bold letters such as $\bsym{p}$ refer to noncommutative polynomials. In particular, $\balpha$ refers to the words in $\monoid{a}$, $\bbeta$ for $\monoid{b}$, and $\bgamma$ for $\monoid{c}$. Incidentally, any word $\bomega \in \langle \underline{a}, \underline{b}, \underline{c} \rangle$ can be uniquely decomposed $\bomega = \balpha\bbeta\bgamma$. This coincides with our discussion in Section~\ref{sec:basicnotation}, and, by construction, the canonical abstraction of any Commutator Tripartite Quantum Distribution gives rise to the ring $\mathbb{T}^{abc}$.

\subsection{Scalar extension polynomials}\label{sec:ScaExtPoly}
Consider the alphabet $\underline{a}$. For any word $\balpha \in \monoid{a}$, we define \emph{scalar extension} of $\balpha$ as $\kappa_{\balpha}$, imposing the commutation condition:
\begin{align*}
    [\kappa_{\balpha},\balpha'] = [\kappa_{\balpha},\kappa_{\balpha'}] = 0 \quad \forall \balpha, \balpha' \in \monoid{a}.
\end{align*}
We define the scalar extension of the empty word $1$ to be itself, i.e. $\kappa_{1} = 1$.

Despite being fully commuting, note that $\kappa_{\balpha\balpha'} \neq \kappa_{\balpha'\balpha}$ in general. Denote the \emph{infinite} alphabet generated by the scalar extension of $\underline{a}$ by $\underline{\kappa_{A}}$, and let $\underline{\kappa_{A}^n}$ be the finite alphabet consisting of the letters $\kappa_{\balpha}$ with $|\balpha| \leq n$. The elements of the free monoid $\monoid{\kappa_{A}}$ are said to be \emph{pure scalar words} and are generally represented by $\bkappa$. The \emph{length} of the word $\bkappa$ is defined as the number of letters in the word, denoted by $|\bkappa|$. 

Observe also that, different from the reduced monoid $\monoid{a}$, in the free monoid $\monoid{\kappa_{A}}$ we have $\kappa_{\balpha}^2 \neq \kappa_{\balpha}$ in general.
Following the commutation relations, define the commutative ring of \emph{pure scalar extension polynomials} associated with $\underline{a}$ as $T_a = \mathbb{R}[\kappa_{\balpha}, \balpha \in \monoid{a}]$. We define the ring of \emph{scalar extension polynomials} as $\mathbb{T}^a_{\mathrm{SE}} = T_a\monoid{a}$.

Similarly to the previous section, we can generalise $\mathbb{T}^a_{\mathrm{SE}}$ to the case of having multiple alphabets. In the case of having $\underline{a}$, $\underline{b}$, and $\underline{c}$, we define the scalar extension $\kappa_{\bomega}$ of the word $\bomega \in \langle \underline{a}, \underline{b}, \underline{c} \rangle$, which generate an infinite alphabet $\underline{\kappa_{ABC}}$ and a finite alphabet $\underline{\kappa_{ABC}^n}$, analogously. We again use $\bkappa$ for pure scalar words.

Denote the commutative ring of \emph{bilocal pure scalar extension polynomials} by $T_{abc} = \mathbb{R}[\kappa_{\bomega}, \bomega \in \langle \underline{a}, \underline{b}, \underline{c} \rangle]$ and the ring of \emph{bilocal scalar extension polynomials} by $\mathbb{T}^{abc}_{\mathrm{SE}} = T_{abc} \langle \underline{a}, \underline{b}, \underline{c} \rangle$, while imposing commutativity
\begin{align*}
    [\kappa_{\bomega},\balpha] = [\kappa_{\bomega},\bbeta] = [\kappa_{\bomega},\bgamma] = 0.
\end{align*}
It is clear that $\mathbb{T}^a_{\mathrm{SE}}$ contains $\mathbb{T}^a$ and $\mathbb{T}^{abc}_{\mathrm{SE}}$ contains $\mathbb{T}^{abc}$. The $^*$-structure can then be extended to $\mathbb{T}^a_{\mathrm{SE}}$ and $\mathbb{T}^{abc}_{\mathrm{SE}}$ by requiring $T_{\mathrm{SE}}^{\dagger} = T_{\mathrm{SE}}$.

Note that any \emph{word} in $\langle \underline{a}, \underline{\kappa_{A}} \rangle$ is of the form $\bkappa\balpha$. Analogously, any word in $\langle \underline{a}, \underline{b}, \underline{c}, \underline{\kappa_{ABC}} \rangle$ is of the form $\bkappa\balpha\bbeta\bgamma$. In case of ambiguity in notation, we ask for $T_{abc}$-linearity of the $\kappa$ symbol, that is
\begin{align*}
    \kappa_{\bomega + \kappa_{\bomega}\bnu} = \kappa_{\bomega} + \kappa_{\bomega}\kappa_{\bnu},
\end{align*}
in the case that we might have written a ``scalar extension of some scalar extension polynomials''.

We can evaluate the polynomials in $\mathbb{T}^{a}_{\mathrm{SE}}$. Fix a Hilbert space $\cH$ with a set of variables $\{\underline{A}=(A_1,\dots,A_I) \mid A_i \in \mathcal{B}(\cH)\}$, and choose a normalised positive operator $\tau \in \mathcal{B}(\cH)$. Define \emph{$\tau$-evaluation} for any $\bsym{p} \in \mathbb{T}^{a}_{\mathrm{SE}}$ to be the matrix trace polynomial $\bsym{p}_{\tau}(\underline{A})$, by replacing each letter $a_i$ with operator $A_i$ and each $\kappa_{\balpha}$ with the number $\Tr{\tau \halpha} = \tr{\tau}{\halpha}$, where $\halpha$ is the operator corresponding to the word $\balpha$. The generalisation to $\mathbb{T}^{abc}_{\mathrm{SE}}$ is straightforward.

Here is an example of a $\tau$-evaluation.
\begin{example}
Consider $\underline{a} = (a_1, a_2, a_3)$. We have $\bsym{p} = \kappa_{a_1}\kappa_{a_1a_3} + \kappa_{a_2}^2 \in T_a$ and $\bsym{q} = a_1 \kappa_{a_1}\kappa_{a_2}a_3 \in \mathbb{T}^{a}_{\mathrm{SE}}$. For evaluation, we may choose the Hilbert space of a qubit with the variables being Pauli matrices $\Vec{\sigma} = (\sigma_x, \sigma_y, \sigma_z)$. Take $\tau$ as an arbitrary rank-$1$ projector. 
Whence,
\begin{align*}
    \bsym{p}_{\tau}(\Vec{\sigma}) &= \Tr{\tau\sigma_x}\Tr{\tau\sigma_x\sigma_z} + \Tr{\tau\sigma_y}^2 \\
    \bsym{q}_{\tau}(\Vec{\sigma}) &= \sigma_x \Tr{\tau\sigma_x} \Tr{\tau\sigma_y} \sigma_z.
\end{align*}
\end{example}

\subsection{Moment matrix and positive unital symmetric linear functional}\label{sec:AppendixMMatrixFunctionalEquivalence}
We discuss the equivalence between moment matrices, such as $\Gamma^n$ in Section~\ref{sec:standardNPA}, and linear functionals of $\mathbb{T}^{abc}$. The discussion here has an obvious generalisation to $\mathbb{T}^{abc}_{\mathrm{SE}}$.

Without loss of generality, we work with $\mathbb{T}^a = \mathbb{R}\monoid{a}$ for convenience. Fix $d \in \mathbb{N}$, denote $\mathbb{T}^a_{\leq 2d}$ as all polynomials with words of length $\leq 2d$.

A matrix $H$ is indexed by words of length $\leq d$ is said to satisfy \emph{noncommutative Hankel condition} if
\begin{align*}
    H_{\bomega,\bnu} = H_{\bomega',\bnu'} \text{, whenever } \bomega^{\dagger}\bnu = \bomega'^{\dagger}\bnu'.
\end{align*}
The matrix $H$ indexed by words of length $\leq d$ is a \emph{moment matrix of order $d$}, also known as \emph{Hankel matrix}, if it satisfies the noncommutative Hankel condition, is positive semi-definite, and $H_{1,1} = 1$. We note that this coincides with the properties of the moment matrix $\Gamma^n$ in the standard NPA hierarchy, which justifies the terminology ``moment matrix''.

To study the properties of $\mathbb{T}^{abc}$, a natural object to consider is its dual, which is the space of its linear functionals. In particular, we focus on \emph{unital symmetric positive linear functional} $L:\mathbb{T}^a_{\leq 2d} \to \mathbb{R}$. A linear functional $L$ is said to be \emph{unital} if $L(1) = 1$ and \emph{symmetric} if $L(\bsym{p}) = L(\bsym{p}^{\dagger})$ for all $\bsym{p}$ in the domain of $L$. If $L(\bsym{p}^{\dagger}\bsym{p}) \geq 0$ for all $\bsym{p} \in \mathbb{T}^a_{\leq 2d}$, we say that $L$ is positive.

These two objects are, in fact, equivalent. Given a unital symmetric positive linear functional $L:\mathbb{T}^a_{\leq 2d} \to \mathbb{R}$, we can define a moment matrix $H$ of order $d$ via formula $H_{\bomega,\bnu} = L(\bomega^{\dagger}\bnu)$, and vice versa. To conclude the equivalence, the only nontrivial thing that needs to be checked is the positivity.

Indeed, for any two polynomials $\bsym{p} = \sum_{\bomega}p_{\bomega}\bomega, \bsym{q} = \sum_{\bnu}q_{\bnu}\bnu \in \mathbb{T}^a_{\leq d}$, we obtain two vectors $\vec{p}$ and $\vec{q}$ given by the real coefficients $p_{\bomega}$ and $q_{\bnu}$. We then compute
\begin{align*}
    L(\bsym{p}^{\dagger}\bsym{q}) = \sum_{\bomega,\bnu} p_{\bomega}^{\dagger}q_{\bnu}L(\bomega^{\dagger}\bnu) = \sum_{\bomega,\bnu} p_{\bomega}q_{\bnu}H_{\bomega,\bnu} = \vec{p}^T H \vec{q},
\end{align*}
and the equivalence follows. 

We note that a general unital symmetric positive linear functional $L$ corresponds only to a moment matrix without extra constraints (other than the Hankel conditions). To recover constraints such as in Definition~\ref{def:FactorBMMatrix}, we require the positivity of $L$ on some suitable subsets of the symmetric polynomials.

\section{Proof for factorisation bilocal NPA hierarchy}\label{sec:nonSDPhardproof}
\begin{proof}[Proof of Theorem~\ref{thm:ConvnonSDPBilocNPAHierarchy}]
We start with the easier direction, given any $n \geq 3$ and suppose that we have a Commutative Bilocal Quantum Distribution $\vec{Q}$ as in Definition~\ref{def:CommutatorBilocQDistrib}. By purity we assume $\tau = \ketbra{\phi_1}{\phi_1}$ for some normalised vector $\ket{\phi_1} \in \cH$. We now wish to show the factorisation bilocal NPA hierarchy of order $n$ is satisfied for all $n \geq 3$. 

To this end, we use the canonical abstraction and define the $n$th order moment matrix $\Gamma^n$ through $\Gamma^n_{\bomega,\bnu} = \tr{\tau}{\homega\hnu}$, where $\homega$ and $\hnu$ are the operators corresponding to the words $\bomega, \bnu \in \langle \underline{a}, \underline{b}, \underline{c} \rangle$.

By construction, $\Gamma^n$ clearly satisfies the linear constraints and recovers $\vec{Q}$ with the appropriate choice of words. Positivity is ensured since $\Gamma^n$ is a Gram matrix. We only need to check the nonlinear factorisation constraint. Let $\balpha$ and $\bgamma$ be words in $\monoid{a}$ and $\monoid{c}$, respectively. We compute
\begin{align*}
    \tr{\tau}{\halpha^{\dagger}\tau\hgamma} &= \Tr{\tau\halpha^{\dagger}\tau\hgamma} \\
    &= \Tr{\rho\sigma\halpha^{\dagger}\rho\sigma\hgamma} \\
    &= \Tr{\rho\halpha^{\dagger}\sigma\rho\sigma\hgamma} \\
    &= \Tr{\halpha^{\dagger}\sigma\rho\sigma\hgamma\rho} \\
    &= \Tr{\halpha^{\dagger}\sigma\rho\sigma\rho\hgamma} \\
    &= \Tr{\halpha^{\dagger}\tau^2\hgamma} = \tr{\tau}{\halpha^{\dagger}\hgamma},
\end{align*}
where the commutativity and cyclicity of the trace are applied several times. Then
\begin{align*}
    \Gamma^n_{\balpha,\bgamma} &= \tr{\tau}{\halpha^{\dagger}\hgamma} \\
                    &= \tr{\tau}{\halpha^{\dagger}\tau\hgamma} \\
                    &= \sandwich{\phi_1}{\halpha^{\dagger}}{\phi_1}\sandwich{\phi_1}{\hgamma}{\phi_1} \\
                    &= \tr{\tau}{\halpha^{\dagger}}\tr{\tau}{\hgamma} = \Gamma^n_{\balpha,1}\Gamma^n_{1,\bgamma},
\end{align*}
as desired.

Conversely, we construct a GNS representation of $\mathbb{T}^{abc}$ given the existence of Factorisation Bilocal Moment Matrix $\Gamma^n$ of order $n$ for any natural number $n \geq 3$, exactly as in~\cite{NPA2008}. We then construct the operators $\rho$ and $\sigma$ with the help of nonlinear factorisation constraints.

We only provide a quick sketch for the GNS representation here since it was explicitly done in~\cite{NPA2008}. 
There exists a subsequence $\{\Gamma^{n_i}\} $ of $\Gamma^n$ that converges to a limit $\Gamma^{\infty}$ (in the weak-$^*$ topology, when each $\Gamma^{n}$ is seen as a vector in $\ell_\infty$). 
In particular, the convergence is pointwise, i.e. $\lim_{i \to \infty}\Gamma^{n_i}_{\bomega,\bnu} = \Gamma^{\infty}_{\bomega,\bnu}$ for all $\bomega$, $\bnu$, and hence the limit $\Gamma^{\infty}$ satisfies all linear constraints and nonlinear factorisation constraints. Let $\Gamma^{\infty}_N$ denote the extracted submatrix of $\Gamma^{\infty}$ indexed by words of length $\leq N$, it follows that $\Gamma^{\infty}_N$ is positive for all $N$. 

Now we apply a sequential Cholesky decomposition of the matrices $\Gamma^{\infty}_N$, which realises $\Gamma^{\infty}$ as a Gram matrix of an \emph{infinite} family of vectors $\{\ket{\phi_{\bomega}}\}$ indexed by all words $\bomega$ and
\begin{align*}
    \Gamma^{\infty}_{\bomega,\bnu} = \braket{\phi_{\bomega}}{\phi_{\bnu}}.
\end{align*}
The desired Hilbert space $\cH$ is then defined as the closure of the linear span of $\{\ket{\phi_{\bomega}}\}$, with respect to the inner product induced by the above formula. Let $\ket{\phi_1} = \ket{v^{00}}$ denote the vector corresponding to the null word $1$, we define the global state $\tau = \ketbra{\phi_1}{\phi_1} = \ketbra{v^{00}}{v^{00}}$ which is pure. 

Now, for each word $\bomega = \balpha\bbeta\bgamma \in \langle \underline{a}, \underline{b}, \underline{c} \rangle$, we can construct the corresponding operators with the same decomposition $\homega = \halpha\hbeta\hgamma$, where the commutation relation
\begin{align*}
    [\halpha, \hbeta] = [\halpha, \hgamma] = [\hbeta, \hgamma]
\end{align*}
is satisfied. For each letter $a_i \in \underline{a}$, the operator $\hat{a}_i$ is a Hermitian projector acting on $\cH$ via the formula
\begin{align*}
    \hat{a}_i\ket{\phi_{\bnu}} = \ket{\phi_{\hat{a}_i\bnu}},
\end{align*}
and meeting all standard NPA hierarchy constraints. Similarly to the letters $b_j \in \underline{b}$ and $c_k \in \underline{c}$.

To construct operators $\rho$ and $\sigma$, define $V_{AB_L} = \Span{P(\{\halpha\})}_P$ as the subspace generated by all polynomials that only evaluate operators $\halpha$. Clearly, the space $V_{AB_L}$ is not empty since it contains $\ket{v^{00}}$. Starting from $\ket{v^{00}}$, we use the Gram--Schmidt process to construct an orthogonal basis $\{P_i(\{\halpha\})\ket{v^{00}} = \ket{v^{i0}}, i \in I\}$ for some polynomials $P_i$ and some countable index set $I$. We similarly introduce $V_{B_RC} = \Span{Q(\{\hgamma\})}_Q$ with the orthogonal basis $\{Q_j(\{\hgamma\})\ket{v^{00}} = \ket{v^{0j}}, j \in J\}$ for some countable index set $J$. We may assume $0 \in I \cap J$. Now, we define
\begin{align}
    \rho &= \sum_j \ketbra{v^{0j}}{v^{0j}} \\
    \sigma &= \sum_i \ketbra{v^{i0}}{v^{i0}}.
\end{align}

It remains to check that both $\rho$ and $\sigma$ satisfy the desired properties in Definition~\ref{def:CommutatorBilocQDistrib}. Clearly we have $\rho^2 = \rho$, $\sigma^2 = \sigma$.

We claim that $\braket{v^{i0}}{v^{0j}} = 0$ for all $i \in I \setminus \{0\}$ and $j \in J \setminus \{0\}$. Indeed, if $i,j \neq 0$, we may assume $P_i(\{\halpha\}) = \sum_{\halpha}p^i_{\halpha}\halpha$ and $Q_j(\{\hgamma\}) = \sum_{\hgamma}q^j_{\hgamma}\hgamma$, for some real coefficients $p^i_{\halpha}$ and $q^j_{\hgamma}$. Then
\begingroup \allowdisplaybreaks
\begin{align*}
    \braket{v^{i0}}{v^{0j}} &= \sandwich{v^{00}}{P_i(\{\halpha\})^{\dagger} \cdot Q_j(\{\hgamma\})}{v^{00}} \\
    &= \sum_{\halpha, \hgamma}(p^i_{\halpha})^{\dagger}q^j_{\hgamma}\cdot\sandwich{v^{00}}{\halpha^{\dagger}\hgamma}{v^{00}} \\
    &= \sum_{\halpha, \hgamma}(p^i_{\halpha})^{\dagger}q^j_{\hgamma}\cdot\Gamma^{\infty}_{\balpha, \bgamma} \\
    &= \sum_{\halpha, \hgamma}(p^i_{\halpha})^{\dagger}q^j_{\hgamma}\cdot\Gamma^{\infty}_{\balpha, 1} \cdot \Gamma^{\infty}_{1, \bgamma} \\
    &= \sum_{\halpha, \hgamma}(p^i_{\halpha})^{\dagger}q^j_{\hgamma}\cdot\sandwich{v^{00}}{\halpha^{\dagger}}{v^{00}}\sandwich{v^{00}}{\hgamma}{v^{00}} \\
    &= \sandwich{v^{00}}{P_i(\{\halpha\})^{\dagger}}{v^{00}}\sandwich{v^{00}}{Q_j(\{\hgamma\})}{v^{00}} \\
    &= \braket{v^{i0}}{v^{00}} \braket{v^{00}}{v^{0j}} = 0,
\end{align*}
\endgroup
where both nonlinear factorisation constraint and orthogonality are applied. It follows that
\begin{align*}
    \rho \cdot \sigma = \sum_{ij} \ketbra{v^{0j}}{v^{0j}} \cdot \ketbra{v^{i0}}{v^{i0}} = \ketbra{v^{00}}{v^{00}} = \tau = \sigma \cdot \rho.
\end{align*}

Finally, we check $[\rho,\hgamma] = [\halpha,\sigma] = 0$. Observe that the support of $[\rho,\hgamma]$ is within $V_{B_RC}$ and
\begin{align*}
    \sandwich{v^{0j}}{[\hgamma,\rho]}{v^{0j'}}&=\sum_k\sandwich{v^{0j}}{\hgamma}{v^{0k}} \braket{v^{0k}}{v^{0j'}} -\sum_k\braket{v^{0j}}{v^{0k}}\sandwich{v^{0k}}{\hgamma}{v^{0j'}}\\
    &=\sandwich{v^{0j}}{\hgamma}{v^{0j'}}-\sandwich{v^{0j}}{\hgamma}{v^{0j'}}=0,
\end{align*}
for all $j, j'\in J$. It follows that $[\rho,\hgamma] = 0$. The analogous argument with the basis on $V_{AB_L}$ shows that $[\halpha,\sigma] = 0$. It is straightforward that Definition~\ref{def:CommutatorBilocQDistrib}.(iv) is true, which finishes the proof.
\end{proof}

\section{Scalar extension bilocal hierarchy in details}\label{sec:AppendixBilocScaExtHierarchy}
In this appendix, we aim to formalise the scalar extension bilocal hierarchy in Definition~\ref{def:BSE_hierarchy} and provide a convergence proof for Theorem~\ref{thm:ConvScalarBilocNPAHierarchy}. 

Initially, we consider the hierarchy detailed in Appendix~\ref{sec:ScaExtbyKlep}, which imposes constraints with localising matrices, a specific case of state polynomial optimisation as first introduced by~\cite{klep2024state}. 
Subsequently, we introduce the scalar extension bilocal hierarchy in Appendix~\ref{sec:DefinitionBilocScaExtHierarchy}, which allows for the expression of equality constraints using the polarisation technique of Refs.~\cite{ligthart2023inflation} as an alternative to localising matrices.
In Appendix~\ref{sec:ScaExthardproof}, building up from the proof of Theorem~\ref{thm:ConvScalarBilocNPAHierarchy} due to~\cite{klep2024state}, we demonstrate the convergence of this alternative hierarchy to Projector Bilocal Quantum Distributions.

\subsection{State polynomial framework in bilocal scenario}\label{sec:ScaExtbyKlep}
We show that the discussion in Section~\ref{Sec:MotivNewScalarExtensionHierarchy} leads to a special case of the more general state polynomial optimisation problem in~\cite{klep2024state}, which serves as a foundation to bilocal scalar extension hierarchy later introduced in Appendix~\ref{sec:DefinitionBilocScaExtHierarchy}.

We may denote the PVM and mutual commuting conditions of Alice, Bob, and Charlie by the set
\begin{align*}
    \mathcal{C} = \{ \A^2 = \A, \B^2 = \B, \C^2 = \C, \sum_a \A = \sum_b \B = \sum_c \C = \id, \\ 
    [\balpha, \bbeta] = [\balpha, \bgamma] = [\bbeta, \bgamma] = 0\}.
\end{align*}
Recall that $\mathcal{N} = \langle \underline{\kappa_{ABC}}, \underline{a}, \underline{b}, \underline{c} \rangle$ is the set of words of the form $\balpha\bbeta\bgamma\bkappa$, and $\mathcal{N}^n \subset \mathcal{N}$ is the set of words that are $\le n$ in length. Given a probability distribution $\Vec{Q}$, the set of constraints is defined to be
\begin{equation}\label{eq:ConstraintSetKlep}
\begin{aligned}
    \mathcal{Q} = \mathcal{R} \cup \{ \pm (\kappa_{\balpha}\kappa_{\bgamma} - \kappa_{\balpha\bgamma}) \} \cup \{ \pm (q(abc|xyz) - \kappa_{\A\B\C}) \}.
\end{aligned}
\end{equation}
For each $n \geq 3$, we have the truncated constraint set $\mathcal{Q}^{n} = \mathcal{Q} \cap \mathcal{N}^n$

Now, we construct the moment matrix $\Omega^n$ whose entries are indexed by all words in $\mathcal{N}^n$. 
In this construction, we only consider words that are unique up to commutation relations.  
Recall Equation~\eqref{eq:AssociatedFunctional}, there is an \emph{associated linear functional $L_{\Omega^n}: \mathbb{R}[\mathcal{N}^n] \to \mathbb{R} $} with which we can define localising matrices via Equation~\eqref{eq:AssociatedFunctionalLocalising}.

\begin{theorem}\label{thm:ScalExtHierachyKlep}
    Suppose that $\vec{Q}=\{q(abc|xyz)\}$ is a tripartite probability distribution. If for each $n \geq 3$, there exists a moment matrix $\Omega^n$ indexed by all words of $\mathcal{N}^n$ such that 
    \begin{enumerate}[(i)]
        \item $\Omega^n_{1,1} = 1$,
        \item $\Omega^n$ is positive,
        \item it satisfies all the linear constraints 
        \begin{align*}
            \Omega^n_{\bomega,\bnu}=\Omega^n_{\bomega',\bnu'}
        \end{align*}
        whenever $\kappa_{\bomega^{\dagger} \bnu} = \kappa_{{\bomega'}^{\dagger} \bnu'}$,
        \item for each $\bsym{p} \in \mathcal{Q}^n$, the localising matrix $\Omega^n(\bsym{p})_{\bomega, \bnu}$ is positive semi-definite,
    \end{enumerate}
    then $\Vec{Q}$ is a Projector Bilocal Quantum Distribution.
    The converse also holds.
\end{theorem}
\begin{proof}
    This follows from~\cite[Lemma 6.5]{klep2024state} as a special case of state polynomial optimisation. We note that an important ingredient of the proof is the Kadison--Dubois Theorem, which implies that $L_{\Omega^{\infty}} = \int_{\mathrm{ex}(\mathcal{S})} d\mu(K) K$ for the linear functional associated with the limit $\Omega^{\infty}$. 
    Here $\mathcal{S}$ is the set of symmetric unital linear functionals $\mathbb{R}[\mathcal{N}^n] \to \mathbb{R}$ that are nonnegative over the set $\mathcal{Q}$, and $\mathrm{ex}(\mathcal{S})$ denotes the extremal points. That is, every $K \in \mathrm{ex}(\mathcal{S})$ is a $*$-homomorphism in the sense that $K(\kappa_{\bomega} \kappa_{\bnu}) = K(\kappa_{\bomega})K(\kappa_{\bnu})$, i.e. a pure state on the abelian $C^*$-algebra generated by the scalar extension letters.
    Finally, it can be shown that there exists some extremal point $\tilde{K}$, and check that $\Gamma^{\infty}_{\bmu, \bnu} = \tilde{K}(\kappa_{\bmu^{\dagger} \bnu})$ defines a Factorisation Bilocal Moment Matrix. The claim then follows from Theorem~\ref{thm:ConvnonSDPBilocNPAHierarchy}.
\end{proof}

\subsection{Definition of bilocal scalar extension hierarchy}\label{sec:DefinitionBilocScaExtHierarchy}
We note that the hierarchy introduced in Theorem~\ref{thm:ScalExtHierachyKlep} imposes all constraints through localising matrices. 
In this section, we define the bilocal scalar extension hierarchy, which is an adaptation to the aforementioned hierarchy such that we impose all equality constraints using the polarisation technique of~\cite{ligthart2023inflation} in a systematic way.

Define the \emph{polarisation constraints}
\begin{equation}\label{eq:FactorisationConstraintSetDef}
    \begin{aligned}
        \mathcal{C} = \{ \kappa_{\balpha\bgamma} - \kappa_{\balpha} \kappa_{\bgamma} \mid \forall \balpha, \bgamma \},
    \end{aligned}
\end{equation}
which models the factorisation independence in bilocal scenario, and let $\mathcal{C}^n$ be the \emph{truncated polarisation constraints} given by the restriction to words in $\mathcal{N}^n$.
Given any tripartite probability distribution $\vec{Q}$, define
\begin{align}\label{eq:fCompatibilityConstraint}
    \mathcal{Q} = \{\kappa_{\A\B\C} - q(abc|xyz) \mid \forall a,b,c,x,y,z \},
\end{align}
which is associated with the compatibility to $\Vec{Q}$.
Note that each constraint in $\mathcal{Q}$ from Equation~\eqref{eq:ConstraintSetKlep} correspond to a polynomial in $\mathcal{C} \cup \mathcal{Q}$. 

Now, we restate Definition~\ref{def:BSE_hierarchy}, which constructs the Scalar Extension Bilocal Moment Matrix $\Omega^n$ whose entries are only indexed by scalar extension words in $\mathcal{N}^n$\footnote{This is different from the original scalar extension~\cite{ScalarExtension2019} which considers moment matrix indexed by $\langle \underline{a}, \underline{b}, \underline{c}, \underline{\kappa_{A}^n} \rangle$. However, the inclusion is not necessary for a convergence proof.}.  

\begin{definition}[Scalar Extension Bilocal Moment Matrix] 
Fix $n \in \mathbb{N}$, let $\Omega^n$ be a square matrix indexed by all words $\bomega \in \mathcal{N}^n$ of length $|\bomega|\leq n$. We say that $\Omega^n$ is a \emph{Scalar Extension Moment Matrix of order $n$} if 
\begin{enumerate}[(i)]
\item $\Omega^n_{1,1} = 1$.
\item $\Omega^n$ is positive semi-definite.
\item the localising matrix $\Omega^n(N - \sum_{\bomega \in \mathcal{N}^n} \kappa_{\bomega}^{\dagger}\kappa_{\bomega})$ is positive semi-definite, where $N= | \mathcal{N}^n |$.
\item it satisfies all the linear constraints 
\begin{align*}
  \Omega^n_{\bomega,\bnu}=\Omega^n_{\bomega',\bnu'}
\end{align*}
whenever $\kappa_{\bomega^{\dagger} \bnu} = \kappa_{{\bomega'}^{\dagger} \bnu'}$. 
\item For any polynomial $c \in \mathcal{F}$, we have
\begin{align}\label{eq:ScalExtFactConstraints}
    L_{\Omega^n}((\kappa_c)^2) = 0.
\end{align}
\end{enumerate}
We further say that $\Omega^n$ is compatible with the tripartite distribution $\Vec{Q}$ iff 
\begin{align}\label{eq:ScalExtCompatabilityConstraints}
    L_{\Omega^n}((\kappa_f)^2) = 0
\end{align}
for every $f \in \mathcal{Q}$.
An infinite matrix $\Omega^\infty$ is said to be a \emph{Scalar Extension Bilocal Moment Matrix} iff all of its principal extracted matrices are Scalar Extension Bilocal Moment Matrix of some finite order.
\end{definition}

\begin{definition}[Scalar extension bilocal hierarchy]
Let $\vec{Q}=\{q(abc|xyz)\}$ be a tripartite probability distribution.
We say that $\vec{Q}$ passes the \emph{scalar extension bilocal hierarchy} if for any integer $n\geq 3$, there exists a Scalar Extension Bilocal Moment Matrix $\Omega^n$ of size $n$ which is compatible with $\vec{Q}$.
\end{definition}

\subsection{Convergence of bilocal scalar extension hierarchy}\label{sec:ScaExthardproof}
In this section, we restate and prove Theorem~\ref{thm:ConvScalarBilocNPAHierarchy}.
\begin{theorem}[Convergence of the scalar extension bilocal hierarchy] 
Suppose that $\vec{Q}=\{q(abc|xyz)\}$ is a tripartite probability distribution. Then $\vec{Q}$ passes all the scalar extension bilocal hierarchy tests iff $\vec{Q}$ is a Projector Bilocal Quantum Distribution.
\end{theorem}
\begin{proof}
    We adapt the argument from~\cite{ligthart2023inflation} to Theorem~\ref{thm:ScalExtHierachyKlep}, building upon its proof. 
    We note that, a key assumption for this extension is the boundedness of words in $\mathcal{N}^n$, recognised as the archimedean property for quadratic modules. In Theorem~\ref{thm:ScalExtHierachyKlep}, boundedness is indirectly enforced through PVM properties, corresponding to the set $\{ \kappa_{\bnu^\dagger \bomega \bnu} \mid \bomega \in \mathcal{R}, \bnu \in \langle \underline{a}, \underline{b}, \underline{c} \rangle \}$ in Equation~\ref{eq:ConstraintSetKlep}, leveraging a standard argument for Positivestallens\"atze (e.g. \cite[Lemma~5.4]{klep2024state}). 
    However, in Definition~\ref{def:BSE_hierarchy}, PVM properties are instead enforced by polarisation technique, and Definition~\ref{def:BSE_hierarchy}.(iii) is needed to assure the adaptation.

    Now recall the proof for Theorem~\ref{thm:ScalExtHierachyKlep}, the Kadison--Dubois Theorem implies that $L_{\Omega^{\infty}} = \int_{\mathrm{ex}(\mathcal{S})} d\mu(K) K$ for the linear functional associated with the limit $\Omega^{\infty}$.
    Since every $K$ is $*$-homomorphic, we must have $K(\kappa_{c}^2) = K(\kappa_c)^2 \geq 0$ for every $c \in \mathcal{C}$, and $K(\kappa_f^2)=K(\kappa_f)^2 \geq 0$ for every $f \in \mathcal{Q}$. Combining with Equation~\eqref{eq:SDP_constraints} and Equation~\eqref{eq:SDP_compatibility}, it follows that for all $c \in \mathcal{C}$, $f \in \mathcal{Q}$, we have $K(\kappa_c)=K(\kappa_f) = 0$ $\mu$-almost everywhere.
    
    The final step mirrors that of Theorem~\ref{thm:ScalExtHierachyKlep}. By fixing some extremal $\tilde{K}$, we can easily verify that $\Gamma^{\infty}_{\bmu, \bnu} = \tilde{K}(\kappa_{\bmu^{\dagger} \bnu})$ defines a Factorisation Bilocal Moment Matrix, thereby concluding the proof via Theorem~\ref{thm:ConvnonSDPBilocNPAHierarchy}.
\end{proof}

\begin{remark}\label{rmk:ConvergenceProofScaExt}
We offer two observations on the aforementioned proof.
Firstly, the necessity of squaring constraints $c \in \mathcal{C}$ and $f \in \mathcal{Q}$ is evident in the argument. Indeed, this ensures that $K(\kappa_{c}^2)$ and $K(\kappa_f^2)$ are positive, so that we can use the constraint $L(K(\kappa_{c}^2)) = L(K(\kappa_f^2)) = 0$ to conclude that $K(\kappa_{c}^2)=K(\kappa_f^2) = 0$ $\mu$-almost everywhere.

Secondly, Theorem~\ref{thm:ConvScalarBilocNPAHierarchy} serves as an existence proof, suggesting that convergence does not generally allow the extraction of a projector bilocal quantum representation for $\Vec{Q}$, unlike the factorisation bilocal NPA hierarchy. The optimiser within the scalar extension bilocal hierarchy is typically extractable only when the limit point happens to be extremal.
\end{remark}

\section{The inflation-NPA hierarchy}\label{sec:AppendixInflationNPA}
\subsection{Formalism of inflation-NPA hierarchy in bilocal scenario}\label{sec:formalInflationNPA}
We use the noncommutative polynomial formalism to encode the inflational data of inflation order $m$. That is, we associate PVMs $\{\hAinf{i}\}$, $\{\hBinf{j}{k}\}$, and $\{\hCinf{l}\}$ with alphabets $\underline{a}^i$, $\underline{b}^{jk}$, and $\underline{c}^l$, for $i,j,k,l=1,\cdots,m$. We impose the same algebraic rules on all the letters inherited from the underlying operators and inflational structure. 

Therefore, any word $\bomega \in \langle \underline{a}^i, \underline{b}^{jk}, \underline{c}^l \mid i,j,k,l=1,\cdots,m \rangle$ admits a decomposition\footnote{Unlike the decomposition of words in non-inflated scenario, the expression here is not unique. Nonetheless, it is always possible to choose a canonical representation by fixing a rule on the ordering of all $\bbeta^{jk}$ that respect the partial commutativity. We omit proposing such a rule since it is cumbersome and irrelevant to the following discussion.}
\begin{align*}
    \bomega = \balpha^1 \cdots \balpha^m \bbeta \bgamma^1 \cdots \bgamma^m,
\end{align*}
where $\bbeta \in \langle \underline{b}^{jk} \mid j,k=1,\cdots,m \rangle$. Note that we cannot further write $\bbeta = \bbeta^{11} \cdots \bbeta^{jk} \cdots \bbeta^{mm}$ since $\bbeta^{jk}$ does not commute with $\bbeta^{j'k'}$ in general.

The resulting ring of noncommutative polynomials is 
\begin{align*}
    \mathbb{T}^{a^1 \cdots a^m b^{11} \cdots b^{mm} c^1 \cdots c^m} = \mathbb{R}\langle \underline{a}^i, \underline{b}^{jk}, \underline{c}^l \mid i,j,k,l=1,\cdots,m \rangle.
\end{align*}

The product symmetric group $S_m \times S_m$ naturally acts on the inflation words by permuting the superscripts. For example, let $\theta, \theta' \in S_m$ and consider the word $\bomega = \balpha^1 \cdots \balpha^m \bbeta^{11} \cdots \bbeta^{jk} \cdots \bbeta^{mm} \bgamma^1 \cdots \bgamma^m$. Then
\begin{align*}
    (\theta\times\theta')(\bomega) = \balpha^{\theta(1)} \cdots \balpha^{\theta(m)} \bbeta^{{\theta(1)}{\theta'(1)}} \cdots \bbeta^{\theta(j)\theta'(k)} \cdots \bbeta^{{\theta(m)}{\theta'(m)}} \bgamma^{\theta'(1)} \cdots \bgamma^{\theta'(m)}.
\end{align*}

Now we can define the Inflation Bilocal Moment Matrix and inflation-NPA bilocal hierarchy.

\begin{definition}[Inflation Bilocal Moment Matrix]\label{def:InflationMMatrix}
Fix $n, m \in \mathbb{N}$, let $\Xi^{n,m}$ be a square matrix indexed by all words $\bomega \in \langle \underline{a}^i, \underline{b}^{jk}, \underline{c}^l \mid i,j,k,l=1,\cdots,m \rangle$ of length $|\bomega|\leq n$. We say that $\Xi^{n,m}$ is a \emph{$m$th-Inflation Moment Matrix of order $n$} if
\begin{enumerate}[(i)]
\item ${\Xi^{n,m}}_{1,1} = 1$.
\item $\Xi^{n,m}$ is positive.
\item it satisfies all the linear constraints 
\begin{align*}
  {\Xi^{n,m}}_{\bomega,\bnu}={\Xi^{n,m}}_{\bomega',\bnu'}
\end{align*}
whenever $\bomega^\dagger\bnu = \bomega'^\dagger\bnu'$.
\item it satisfies the $S_m \times S_m$ symmetry
\begin{align}
    \Xi^{n,m}_{1, \bomega} = \Xi^{n,m}_{1, (\theta\times\theta')(\bomega)},
\end{align}
for all permutations $\theta, \theta \in S_m$.
\end{enumerate}
We further say that $\Xi^{n,m}$ is compatible with the tripartite distribution $\vec{Q}$ iff 
\begin{align*}
    {\Xi^{n,m}}_{1,\prod^m_{i=1}\Ainf{i}\Binf{i}{i}\Cinf{i}}=\prod^m_{i=1} q(a^ib^{i,i}c^i|x^iy^{i,i}z^i).
\end{align*}
An infinite matrix $\Xi^{\infty,m}$ is said to be a \emph{$m$th-Inflation Moment Matrix} iff all of its principal extracted matrices are $m$th-Inflation Moment Matrix of some finite order.
\end{definition}

\begin{remark}
As noted in the main text, in the limit of $m$, $n$ to infinity, the extended nonlinear constraints
\begin{align}
    \Xi^{\infty,\infty}_{1, \prod^s_{i=t}\balpha^i\bbeta^{ii}\bgamma^i} = \prod^s_{i=t}\Xi^{\infty,\infty}_{1, \balpha^i\bbeta^{ii}\bgamma^i}
\end{align}
are satisfied due to de Finetti Theorem, where $1 \leq t \leq s \leq m$ and arbitrary diagonal words $\bomega^i = \balpha^i\bbeta^{ii}\bgamma^i$ for all $i=1, \cdots n$ with length $\leq n$.
\end{remark}

\begin{definition}[Bilocal inflation-NPA hierarchy] \label{def:InflationNPATest}
Let $\vec{Q}=\{q(abc|xyz)\}$ be a tripartite probability distribution. We say that $\vec{Q}$ passes \emph{bilocal inflation-NPA hierarchy} if for all integers $m \geq 1$ and $n\geq 3$, there exists a $m$th-Inflation Bilocal Moment Matrix $\Xi^{n,m}$ of order $n$ that is compatible with $\vec{Q}$. 
\end{definition}

As the archetype for quantum inflation for bilocal scenario, the Tensor Bilocal Quantum Distribution passes the bilocal inflation-NPA hierarchy, as expected.

\begin{lemma}
All Tensor Bilocal Quantum Distributions $\vec{P}$ passes the bilocal inflation-NPA hierarchy.
\end{lemma}
\begin{proof}
With inflation canonical abstraction we may define 
\begin{align*}
    \Xi^{n,m}_{\bomega,\bnu} = \tr{\tau_n}{\homega^{\dagger} \hnu}
\end{align*}
for suitable words $\bomega, \bnu$. It is clear that $\Xi^{n,m}$ satisfies Definition~\ref{def:InflationMMatrix}.
\end{proof}

Next, we restate and prove Theorem~\ref{thm:InflationNPAIsTighter}, which is the observation that inflation-NPA hierarchy provides a tighter approximation to Tensor Bilocal Quantum correlations. 

\begin{theorem}
If the probability distribution $\{\vec{P}=\{p(abc|xyz)\}$ passes the bilocal inflation-NPA hierarchy, then it also passes the factorisation and scalar extension bilocal NPA hierarchies.
\end{theorem}
\begin{proof}
By Theorem~\ref{thm:ConvnonSDPBilocNPAHierarchy} and Theorem~\ref{thm:ConvScalarBilocNPAHierarchy} it suffices to consider the scalar extension bilocal hierarchy. Consider any $n \geq 3$ and suppose that the size of alphabet $\underline{a} \cup \underline{b} \cup \underline{c}$ is $d$. We show that, for $m > n\sum_{i=0}^n d^i$, we can obtain the Scalar Extension Moment Matrix $\Omega^n$ from $m$th-Inflation Moment Matrix $\Xi^{n,m}$.

First, we identify the alphabets $\underline{a}^1$, $\underline{b}^{11}$, and $\underline{c}^1$ of $\Xi^{n,m}$ with $\underline{a}$, $\underline{b}$, and $\underline{c}$ for $\Omega^n$. 
Furthermore, for each distinct word $\bomega = \balpha\bbeta\bgamma$ with $|\bomega| \leq n$ in $\langle \underline{a}, \underline{b}, \underline{c}, \underline{\kappa_{ABC}} \rangle$, we assign a distinct index $i \geq 2$ and then identify the scalar extension $\kappa_{\bomega}$ with $\bomega^i = \balpha^i\bbeta^{ii}\bgamma^i$. By construction, $\balpha^i$ with $i \geq 2$ commutes with $\balpha^1$, $\bbeta^{1,1}$, and $\bgamma^1$, and since $i$ are distinct for distinct words, all scalar extensions commute with each other. 

Next, suppose that the above assignment used $s$ many inflation levels, we now need to make sense of $\kappa_{\bomega}^l$. For $\kappa_{\bomega}$ identified with $\bomega^i$, the square $\kappa_{\bomega}^2$ is identified with $\bomega^i \bomega^{i+s}$, and inductively for any power lower than $n$. The choice of $m$ is large enough for index assignment.

Now, define $\Omega^n$ by equations such as
\begin{align*}
    \Omega^n_{1, \balpha\bbeta\bgamma\kappa_{\bomega}} = \Xi^{n,m}_{1, \balpha^1\bbeta^{1,1}\bgamma^1\bomega^i}, \\
    \Omega^n_{1, \kappa_{\bomega}^2} = \Xi^{n,m}_{1, \bomega^i \bomega^{i+s}},
\end{align*}
where $i$ is the unique index assigned to $\bomega$. All the other entries can be defined analogously.

There are only two nontrivial facts to be checked while the rest of the conditions are evident by construction.
For the associated linear functional $L_{\Omega^n}$, suppose the unique index of the word $\kappa_{\A \B \C}$ is $i$, using diagonal factoristion we calculate
\begin{align*}
    L_{\Omega^n}(\kappa_{\A \B \C}) = \Xi^{n,m}_{1, \A^i \B^{i,i} \C^i} = q(abc|xyz), \\
    L_{\Omega^n}(\kappa_{\A \B \C}^2) = \Xi^{n,m}_{1, \A^i \B^{i,i} \C^i \A^{i+s} \B^{i+s, i+s} \C^{i+s}} = q(abc|xyz)^2,
\end{align*}
which implies $L_{\Omega^n}(f(\Vec{Q})) = 0$.

Next, without loss of generality we may assume that $\kappa_{\balpha \bgamma}$ is associated with $\balpha^2 \bgamma^2$, $\kappa_{\balpha}$ with $\balpha^3$, and $\kappa_{\bgamma}$ with $\bgamma^4$. Then, we compute
\begin{align*}
    L_{\Omega^n}(h(\balpha, \bgamma)) &= L_{\Omega^n}(\kappa_{\balpha \bgamma}^2) - 2 L_{\Omega^n}(\kappa_{\balpha \bgamma} \kappa_{\balpha} \kappa_{\bgamma}) + L_{\Omega^n}(\kappa_{\balpha}^2 \kappa_{\bgamma}^2) \\
    &= \Xi^{n,m}_{1, \balpha^2 \bgamma^2 \balpha^{2+s} \bgamma^{2+s}} - 2\Xi^{n,m}_{1, \balpha^2 \bgamma^2 \balpha^{3} \bgamma^{4}} + \Xi^{n,m}_{1, \balpha^2 \balpha^{2+s} \bgamma^{4} \bgamma^{4+s}} \\
    &= (\Xi^{n,m}_{1, \balpha^2 \bgamma^2 \balpha^{2+s} \bgamma^{2+s}} - \Xi^{n,m}_{1, \balpha^2 \bgamma^2 \balpha^{3} \bgamma^{4}}) + (\Xi^{n,m}_{1, \balpha^2 \balpha^{2+s} \bgamma^{4} \bgamma^{4+s}} - \Xi^{n,m}_{1, \balpha^2 \bgamma^2 \balpha^{3} \bgamma^{4}}) = 0.
\end{align*}
In the above, we apply the $S_m \times S_m$ symmetry on $\Xi^{n,m}_{1, \balpha^2 \bgamma^2 \balpha^{3} \bgamma^{4}}$. Specifically, in the first bracket, we consider the permutation $(3 \ (2+s))$ for Alice and Bob, and $(4 \ (2+s))$ for Bob and Charlie; in the second bracket, we consider the permutation $(3 \ (2+s))$ for Alice and Bob, and $(2 \ 4 \ (4+s))$ for Bob and Charlie.
\end{proof}

\subsection{Inflation-NPA hierarchy excludes the shared random bit distribution in the triangle scenario}\label{sec:AppendixSharedRandBitInflationNPA}

In this section, we briefly explain why the inflation-NPA hierarchy can be used to prove the impossibility of obtaining a shared random bit distribution in the triangle scenario. 
This proof is directly inspired by the proof of~\cite{QuantumInflation2021}.

Consider the uniform shared random bit distribution $\vec{P}_{ABC}=1/2([000]+[111])$, where $[000]$ (resp. $[111]$) is the deterministic distribution of all parties outputting 0 (resp. 1).
Assume, for contradiction, that it admits an inflation-NPA hierarchy of compatible moment matrices $\Xi^{n,m}$. 
We now show that the existence of $\Xi^{2,2}$ already provides a contradiction.

First, note that the distribution $\vec{Q}_{A^{11}B^{11}C^{12}}$ with the coefficients 
\begin{equation}
q(a_{11}b_{11}c_{12})= \Xi^{2,2}_{A^{11}B^{11},C^{12}} 
\end{equation}
 is a probability distribution (for simplicity in the notation, we omitted the subscripts; e.g., $A^{11}$ should be interpreted as $A^{11}_{a_{11}}$), as these numbers are positive (they can be seen as concrete measurement probabilities in a quantum experiment) and sum to one according to the constraints satisfied by $\Xi^{2,2}$.

According to the diagonalisation constraint, we have the marginal $\vec{Q}_{A^{11}B^{11}}$ over the first two outputs satisfying $\vec{Q}_{A^{11}B^{11}}=\vec{P}_{AB}=1/2([00]+[11])$, which means that we have $a^{11}=b^{11}$ is a shared uniform random bit.
Moreover, we have $\vec{Q}_{B^{11}C^{12}}=\vec{Q}_{B^{11}C^{11}}=1/2([00]+[11])$, where we first used the symmetry condition and then the diagonalisation constraint, meaning that we have $b^{11}=c^{12}$ is a shared uniform random bit.
These two conditions imply that we must have $a^{11}=c^{12}$ is a shared random bit, that is, $\vec{Q}_{A^{11}C^{12}}=1/2([00]+[11])$.

However, the symmetry condition also imposes $\vec{Q}_{A^{11}C^{12}}=\Xi^{2,2}_{A^{11},C^{12}}=\Xi^{2,2}_{A^{11},C^{22}}$, and the diagonalisation constraint implies $\Xi^{2,2}_{A^{11},C^{22}}=\vec{P}_{A}\cdot \vec{P}_{C} = 1/4([00]+[01]+[10]+[11])$, which is contradictory.

\section{Commutator-based definition of quantum correlations in networks}\label{sec:AppendixCommutatorBasedDefinition}
In this section, we propose a general commutator-based definition of accessible quantum correlations $\vec{Q}$ feasible in an arbitrary network, generalising the Definition~\ref{def:CommutatorBilocQDistrib} of Projector Bilocal Quantum Distributions $\vec{Q}$. Then, we discuss potential issues with this definition. Lastly, we show that purification of mixed states cannot be done with this definition. 

\subsection{Definition of Projector Network Quantum Distributions}

We propose the following natural general projector-based definition of the quantum correlations $\vec{Q}$ feasible in an arbitrary network:
\begin{enumerate}[(i)]
    \item The global state, represented by a projector over a pure state $\tau$, acts over a global Hilbert space $\cH$.
    \item Each party $A, B, C, ...$ PVMs operators $\hA, \hB, \hC, ... $ acts over $\cH$, two operators associated with different parties commute.
    \item $q(abc\cdots|xyz\cdots)=\tr{\tau}{\hA\hB\hC\cdots}$.
    \item Each source is associated with a projector $\rho, \sigma, \pi, ...$ acting over $\cH$, and two projectors associated with different sources commute. Moreoever, the product of these operators is $\tau=\rho\cdot\sigma\cdot\pi\cdots$.
    \item For any party (say, $A$) not connected to a source (say, corresponding to source $\rho$), the associated operators commute with the source projector (here, $[\rho,\hA]=0$).
\end{enumerate}

Note that if we restrict to only (i), (ii), and (iii), we recover commutator multipartite quantum correlations as in the standard Bell scenarios.
 
As we show in Theorem~\ref{thm:InflationNPAIsTighter}, inflation-NPA method converges to a subset of the set of Projector Bilocal Quantum Correlations.
The equality of these sets is later answered positively by~\cite[Theorem~16]{ligthart2023inflation}.

However, in networks beyond the bilocal scenario, it is still an ongoing research endeavour on the equivalence of the sets of correlation compatible with the inflation-NPA method and the set of Projector Quantum Correlations.
Strict inclusion would imply that either the definition of Projector Quantum Correlations cannot be taken as the quantum theory postulated correlations in the general network scenario, or that the inflation-NPA method is not the optimal method to characterise quantum correlations in the these general scenarios.
For example, in the triangle scenario in Figure~\ref{fig:GeneralisedScenarios}a, it reads:
\begin{definition}[Commutator Triangle Quantum Distributions]\label{def:CommutatorTriangleQDistrib}
Let $\vec{Q}=\{q(abc|xyz)\}$ be a three-party probability distribution. 
We say that $\vec{Q}$ is a \emph{Commutator Triangle Quantum Distribution} iff there exist a Hilbert space $\cH$, a projector over a pure state $\tau$, three mutually commuting projectors (of possibly infinite trace) $\rho$, $\sigma$, and $\pi$, and some PVMs $\{\hA\}$, $\{\hB\}$, and $\{\hC\}$ over $\cH$ such that
\begin{enumerate}[(i)]
    \item $\tau$ acts on $\cH$.
    \item $[\hA, \hB]=[\hB, \hC]=[ \hC,\hA]=0$ .
    \item $q(abc|xyz)=\tr{\tau}{\hA\hB\hC}$.
    \item  $[\rho, \sigma]=[\sigma, \pi]=[ \pi,\rho]=0$ and $\tau = \rho\cdot\sigma\cdot\pi$.
    \item $[\hA,\rho] = [\hB,\pi] = [\hC, \sigma] = 0$.
\end{enumerate}
\end{definition}

In that case, there exists a weaker method than inflation-NPA to find constraints on the obtainable correlations in the triangle network, called the non-fanout inflation method~\cite{Henson2014,Inflation}. 
Contrary to the factorisation bilocal NPA, scalar extension, and inflation-NPA hierarchies, it is not based on the mathematical Hilbert space formulation of quantum theory. It is solely based on the assumptions of causality and device replication (any device should be replicable in independent copies)~\cite{GisinNSI,BeigiCovariance,CoiteuxRoyPRL,CoiteuxRoyPRA,BancalNLCouplers}, two very weak assumptions.
For network scenarios without loops, the non-fanout inflation does not provide nontrivial constraints beyond the fact that marginal correlations between groups of parties not sharing a source should factorise (as there does not exist any nontrivial non-fanout inflation of no loop networks): this method cannot be used in the bilocal scenario to restrict correlations.

However, in other networks, it implies that a feasible $\vec{Q}$ should be associated with compatible distributions in all the non-fanout inflation of this network. For instance, if $\vec{Q}$ is feasible in the triangle network, it should be associated with an other distribution $\vec{R}$ defined in the hexagon inflation of the triangle, compatible with $\vec{Q}$. As shown in~\cite{Henson2014,Inflation}, this allows one to reject the feasibility of some distributions such as the shared random bit in the triangle scenario.
It is uncertain whether any $\vec{Q}$ satisfying Definition~\ref{def:CommutatorTriangleQDistrib} above is always associated with such distributions $\vec{R}$ in the hexagon (or larger) inflation of the triangle.
In case it is not (which would be the case if this definition allows for the shared random bit distribution), this would imply that such a version of quantum theory based on a commutator-based postulate does not allow one to consider independent copies of systems, which we would see as a strong argument against this version of quantum theory.

\subsection{Pure-mixed state inequivalence in commutator model}\label{sec:AppendixPureMixFormulationInequivalence}
In this section, we show that the condition of $\tau$ being pure in Definition~\ref{def:CommutatorBilocQDistrib} is nontrivial, as stated in Remark~\ref{rem:PureMixFormulationCommtutator}. That is, we cannot reformulate Definition~\ref{def:CommutatorBilocQDistrib} in terms of demanding $\tau$ to be only a mixed state.

Indeed, consider six qubit spaces $\cH_{A_i}$, $\cH_{B_i}$, and $\cH_{C_i}$, for $i=0, 1$.
Let $\cH_i = \cH_{A_i} \otimes \cH_{B_i} \otimes \cH_{C_i}$, and define the global Hilbert space by direct sum decomposition $\cH=\cH_0 \oplus \cH_1$.
Let $\cH_i = \cH_{A_i} \otimes \cH_{B_i} \otimes \cH_{C_i}$, and define the global Hilbert space by direct sum decomposition $\cH=\cH_0 \oplus \cH_1$.
The global state $\tau$ is defined by
\begin{align}
    \tau = \frac{1}{2}\left( \ketbra{0}{0}_{A_0} \otimes \ketbra{0}{0}_{B_0} \otimes \ketbra{0}{0}_{C_0} \oplus \ketbra{1}{1}_{A_1} \otimes \ketbra{1}{1}_{B_1} \otimes \ketbra{1}{1}_{C_1}\right).
\end{align}
In particular, we have a product decomposition $\tau=\rho\sigma$, where
\begin{equation}
    \begin{aligned}
        \sqrt{2} \rho &= \ketbra{0}{0}_{A_0} \otimes \id_{B_0C_0} \oplus \ketbra{1}{1}_{A_1} \otimes \id_{B_1C_1}, \\
        \sqrt{2} \sigma &= \ketbra{0}{0}_{B_0} \otimes \ketbra{0}{0}_{C_0} \otimes \id_{A_0} \oplus \ketbra{1}{1}_{B_1} \otimes \ketbra{1}{1}_{C_1} \otimes \id_{A_1}.
        \end{aligned}
\end{equation}
For measurements, we take the PVMs
\begin{equation}\label{eq:CounterExampleSharedRandomPVM}
    \begin{aligned}
    A_{a=0} &= \ketbra{0}{0}_{A_0} \otimes \id_{B_0C_0} \oplus \ketbra{0}{0}_{A_1} \otimes \id_{B_1C_1}, \\
    A_{a=1} &= \ketbra{1}{1}_{A_0} \otimes \id_{B_0C_0} \oplus \ketbra{1}{1}_{A_1} \otimes \id_{B_1C_1}, \\
    B_{b=0} &= \ketbra{0}{0}_{B_0} \otimes \id_{A_0C_0} \oplus \ketbra{0}{0}_{B_1} \otimes \id_{A_1C_1}, \\
    B_{b=1} &= \ketbra{1}{1}_{B_0} \otimes \id_{A_0C_0} \oplus \ketbra{1}{1}_{B_1} \otimes \id_{A_1C_1}, \\
    C_{c=0} &= \ketbra{0}{0}_{C_0} \otimes \id_{A_0B_0} \oplus \ketbra{0}{0}_{C_1} \otimes \id_{A_1B_1}, \\
    C_{c=1} &= \ketbra{1}{1}_{C_0} \otimes \id_{A_0B_0} \oplus \ketbra{1}{1}_{C_1} \otimes \id_{A_1B_1}.
    \end{aligned}
\end{equation}
It can be easily checked that Definition~\ref{def:CommutatorTriangleQDistrib} is satisfied and that
\begin{align*}
    \tr{\tau}{A_a B_b C_c} = \begin{cases} \frac{1}{2} & a=b=c 
    \\ 0 & \mathrm{otherwise}
    \end{cases}
\end{align*}
is precisely the shared random bit distribution $\vec{P}_{ABC}=1/2([000]+[111])$.

Note that introducing an extra qubit space $\cH_S$ of basis $\ket{0}_S$, $\ket{1}_S$ corresponding to cases $i=0,1$, respectively, one can state this example differently. In the global Hilbert space $\cH=\cH_{A} \otimes \cH_{B} \otimes \cH_{C}\otimes\cH_S$:
\begin{equation}
    \begin{aligned}
    \tau &= \frac{1}{2} \sum_{i=0,1} \ketbra{i}{i}_{A} \otimes \ketbra{i}{i}_{B} \otimes \ketbra{i}{i}_{C}\otimes \ketbra{i}{i}_{S},\\
    \sqrt{2} \rho &= \sum_{i=0,1} \ketbra{i}{i}_A \otimes \ketbra{i}{i}_S \otimes\id_{BC},\\
    \sqrt{2} \sigma &= \sum_{i=0,1} \ketbra{i}{i}_B \otimes \ketbra{i}{i}_C \otimes \ketbra{i}{i}_S \otimes\id_{A},\\
    A_{a=0} &= \ketbra{0}{0}_{A} \otimes \id_{BCS},\\
    A_{a=1} &= \ketbra{1}{1}_{A} \otimes \id_{BCS}, 
    \end{aligned}
\end{equation}
and symmetrically for $B_0, B_1, C_0, C_1$.

The quantum strategies provided above obviously satisfy Definition~\ref{def:CommutatorBilocQDistrib} except for the purity condition. However, it is also clear that $\vec{P}_{ABC}$ cannot be a bilocal quantum distribution because it cannot factorise between Alice and Charlie. 

The deep reason of pure-mixed state inequivalence is that the usual purification does not preserve commutator structure\footnote{Note that in tensor-based model, the purification of a mixed state $\tau = \rho \otimes \sigma$ always assumes the form $\ket{\phi} = \ket{\psi_1} \otimes \ket{\psi_2}$, where \ket{\psi_1} and \ket{\psi_2} are the purification of $\rho$ and $\sigma$, respectively.}, we provide an alternative definition in commutator model.
To this end, consider Definition~\ref{def:CommutatorBilocQDistrib} with distribution $\vec{Q}$, but we instead require that $\rho$, $\sigma$ are only (possibly infinite trace) positive operators for $\tau=\rho\sigma$. We say that $\tau$ is \emph{c-purifiable} if there exists some (pure) density operator $\tau'=\rho'\sigma'$ satisfying all the constraints in Definition~\ref{def:CommutatorBilocQDistrib} with $\vec{Q}$, and $\tau'$ is said to be the \emph{c-purification} of $\tau$. The generalisation of c-purification to an arbitrary network in Appendix~\ref{sec:AppendixCommutatorBasedDefinition} is straightforward. We have the following criterion for c-purifibility in bilocal scenario.

\begin{corollary}\label{cor:CPurifyFactorise}
Consider the notation above. Then $\tau$ is c-purifiable iff all the factorsiation constraints hold, that is,  $\tr{\tau}{\halpha\hgamma}=\tr{\tau}{\halpha}\tr{\tau}{\hgamma}$, for any $\balpha \in \monoid{a}$ and $\bgamma \in \monoid{c}$.
\end{corollary}
\begin{proof}
This is a direct consequence of the proof for Theorem~\ref{thm:ConvnonSDPBilocNPAHierarchy}
\end{proof}

Note that the quantum strategies provided above are examples of Corollary~\ref{cor:CPurifyFactorise}. Also note that, with the notion of c-purifibility, it is possible to formulate commutator-based quantum correlations in general networks in terms of mixed-state. Nevertheless, such a reformulation is cumbersome and redundant; hence we omitted it. Finally, we observe that the inequivalence between the pure state formulation and the mixed state formulation, and the need for an alternative notion of purification, suggest extra subtlety in the commutator-based model compared to the tensor-based model.

\section{The four party star network} \label{AppendixStarNetwork}

The set of distributions that can be obtained in the four-party star network of Figure~\ref{fig:GeneralisedScenarios}a (according to the tensor postulate) is a direct generalisation of Definition~\ref{def:TensorBilocQDistrib}, with an extra state $\pi_{DB_D}$ acting over $\cH_{B_D}\otimes \cH_{D}$ and additional PVMs for $D$.
We now discuss the potential generalisation of our proof in Section~\ref{sec:TripartiteQDistribStandardNPA} to this new scenario. 
We show that when a distribution $\vec{Q}$ admits a \emph{Factorisation Star-shaped Trilocal Moment Matrix} (see Definition~\ref{def:FactorTrilocalMomentM}), then
it is an \emph{Anti-state Projector Star-shaped Quantum Distribution} (which is weaker than what we could expect following the definition of Appendix~\ref{sec:AppendixCommutatorBasedDefinition}). 
As we do not use all the assumptions at our disposal, we conjecture that a more restricted model, maybe the one of Section~\ref{sec:AppendixCommutatorBasedDefinition}, can be proven from the existence of such a hierarchy of Factorisation Star-shaped Trilocal Moment Matrix.

Let us first generalise the factorisation hierarchy by introducing pairwise factorisation constraints:

\begin{definition}[Factorisation Star-shaped Trilocal Moment Matrix]\label{def:FactorTrilocalMomentM}
Fix $n \in \mathbb{N}$, let $\Gamma^n$ be a square matrix indexed by all words $\bomega \in \langle \underline{a}, \underline{b}, \underline{c}, \underline{d} \rangle$ of length $|\bomega|\leq n$. We say that $\Gamma^n$ is a \emph{Factorisation Star-shaped Trilocal Moment Matrix of order $n$} if
\begin{enumerate}[(i)]
\item $\Gamma^n_{1,1} = 1$.
\item $\Gamma^n$ is positive.
\item it satisfies all the linear constraints 
\begin{align*}
  \Gamma^n_{\bomega,\bnu}=\Gamma^n_{\bomega',\bnu'}
\end{align*}
whenever $\bomega^{\dagger} \bnu = {\bomega'}^{\dagger} \bnu'$.
\item it satisfies the pairwise nonlinear factorisation constraints
\begin{align*}
  \Gamma^n_{\balpha, \bgamma}&=\Gamma^n_{\balpha,1}\cdot \Gamma^n_{1,\bgamma}, \\
  \Gamma^n_{\balpha, \bdelta}&=\Gamma^n_{\balpha,1}\cdot \Gamma^n_{1,\bdelta}, \\
  \Gamma^n_{\bgamma, \bdelta}&=\Gamma^n_{\bgamma,1}\cdot \Gamma^n_{1,\bdelta},
\end{align*}
where $\balpha \in \monoid{a}$, $\bgamma \in \monoid{c}$, and $\bdelta \in \monoid{d}$.
\item it satisfies the nonlinear constraints 
\begin{align*}
    \Gamma^n_{\balpha\bgamma, \bdelta}&=\Gamma^n_{\balpha\bgamma,1}\cdot \Gamma^n_{1,\bdelta},
\end{align*}
equivalent to the three-factorisation $\Gamma^n_{\balpha\bgamma, \bdelta} = \Gamma^n_{\balpha,1}\cdot\Gamma^n_{\bgamma,1}\cdot\Gamma^n_{1,\bdelta}$ due to (iv).
\end{enumerate}
An infinite matrix $\Gamma^\infty$ is said to be a \emph{Factorisation Star-shaped Trilocal Moment Matrix} if all of its extracted matrices are Factorisation Trilocal Moment Matrices of some finite order.
We further say that $\Gamma^n$ is compatible with the four-party distribution $\vec{Q}$ iff
\begin{align*}
  \Gamma^n_{1,\A\B\C\D}=q(abcd|xyzw).
\end{align*}
\end{definition}

Observe that condition (iv) in Definition~\ref{def:FactorTrilocalMomentM} does not imply condition (v). In other words, condition (iv) alone does not fully contain the consequences in terms of factorisation of the topology of the star-shaped network.
This definition is associated with the following hierarchy of tests that some distribution $\vec{Q}$ is feasible in the four-party star network:

\begin{definition}[Factorisation star-shaped trilocal NPA hierarchy]
Let $\vec{Q}=\{q(abcd|xyzw)\}$ be a probability distribution for Alice, Bob, Charlie, and Dave. 
We say that $\vec{Q}$ passes the \emph{factorisation star-shaped trilocal NPA hierarchy} if for all integers $n\geq 4$, there exists a Factorisation Star-shaped Trilocal Moment Matrix $\Gamma^n$ of order $n$ that is compatible with $\vec{Q}$. 
\end{definition}

We now show that if $\vec{Q}$ admits such a hierarchy, then it is an Anti-state Projector Star-shaped Quantum Distribution (defined below). 
However, contrary to the bilocal scenario, we do not prove the converse and conjecture that it does not hold, as we suspect that our proof does not use condition (v) in Definition~\ref{def:FactorTrilocalMomentM}.

\begin{definition}[Anti-state Projector Star-shaped Trilocal Quantum Distributions]\label{def:CommutatorQuadStarQDistrib}
Let $\vec{Q}=\{q(abcd|xyzw)\}$ be a probability distribution for Alice, Bob, Charlie, and Dave. 
We say that $\vec{Q}$ is a \emph{Anti-state Projector Star-shaped Trilocal Quantum Distribution} iff there exist a Hilbert space $\cH$, a projector over a pure state $\tau$, three positive operators (possibly infinite trace) $\bar\rho$, $\bar\sigma$, $\bar\pi$ commuting with each other, and PVMs $\{\hA\}$, $\{\hB\}$, $\{\hC\}$, $\{\hD\}$ over $\cH$ such that
\begin{enumerate}[(i)]
    \item $q(abc|xyz)=\Tr{\tau\cdot \hA\cdot\hB\cdot\hC\cdot\hD}$
    \item $[\hA,\hB]=[\hA,\hC]=[\hA,\hD]=[\hB, \hC] = [\hB, \hD] = [\hC, \hD] =0$ 
    \item $\tau = \bar\rho\cdot\bar\sigma=\bar\sigma\cdot\bar\pi=\bar\pi\cdot\bar\rho$
    \item $[\bar\rho,\bar\sigma] = [\bar\rho,\bar\pi] = [\bar\sigma,\bar\pi] = 0$
    \item $[\hA,\bar\rho] = [\hC,\bar\sigma] = [\hD, \bar\pi]=0$
\end{enumerate}
for all $\hA$, $\hB$, $\hC$, and $\hD$.
\end{definition}

Note that this definition is not the natural extension of Definition~\ref{def:CommutatorBilocQDistrib}. 
Indeed, here the operators $\bar\rho,\bar\sigma, \bar\pi$ should not be understood as the original states $\rho, \sigma, \pi$, but as the corresponding ``anti-states'' $\bar\rho= \sigma\cdot\pi, \bar\sigma= \pi\cdot\rho$ and $\bar\pi=\rho\cdot\sigma$.
The bilocal scenario can be seen as a degenerate case: As there are only two states, we could interpret in the proof of Theorem~\ref{thm:ConvnonSDPBilocNPAHierarchy} our constructed $\sigma$ as $\bar\rho$ and $\rho$ as $\bar\sigma$.

\begin{theorem}[Convergence of the factorisation star-shaped trilocal NPA hierarchy]\label{thm:ConvSDPStarShapedNPAHierarchy}
Suppose that $\vec{Q}=\{q(abcd|xyzw)\}$ is a probability distribution for Alice, Bob, Charlie, and Dave. 
If $\vec{Q}$ passes all the factorisation star-shaped trilocal NPA hierarchy tests, then $\vec{Q}$ is a Anti-state Commutator Star-shaped Trilocal Quantum Distribution.

Moreover, if we forget the condition (v) of Definition~\ref{def:CommutatorQuadStarQDistrib} (three-factorisation), then the above two statements are equivalent.
\end{theorem}
\begin{proof}
The proof is completely analogous to the proof of Theorem~\ref{thm:ConvnonSDPBilocNPAHierarchy} in Appendix~\ref{sec:nonSDPhardproof}. Suppose that we have an Anti-state Commutator Star-shaped Trilocal Quantum Distribution $\vec{Q}$ and we wish to show that Definition~\ref{def:CommutatorQuadStarQDistrib} can be satisfied modulo (v). The only nontrivial conditions that need to be shown are (iv).

Condition (iv) is done similarly to Appendix~\ref{sec:nonSDPhardproof}. Indeed,
\begin{align*}
    \tr{\tau}{\halpha^{\dagger}\tau\hgamma} = \Tr{\tau\halpha^{\dagger}\tau\hgamma}
    &= \Tr{\bar{\sigma}\bar{\rho}\halpha^{\dagger}\bar{\sigma}\bar{\rho}\hgamma} \\
    &= \Tr{\bar{\rho}\halpha^{\dagger}\bar{\sigma}\bar{\rho}\hgamma\bar{\sigma}} \\
    &= \Tr{\halpha^{\dagger}\bar{\rho}\bar{\sigma}\bar{\rho}\bar{\sigma}\hgamma} \\
    &= \Tr{\halpha^{\dagger}\tau^2\hgamma} = \tr{\tau}{\halpha^{\dagger}\hgamma},
\end{align*}
where we used the commutativity and cyclicity of the trace. This implies $\Gamma^n_{\balpha,\bgamma}=\Gamma^n_{\balpha,1}\Gamma^n_{1,\bgamma}$ with the same argument as in Appendix~\ref{sec:nonSDPhardproof}. The other two constraints in (iv) follow analogously.

For the converse, we only comment on the construction of $\bar\rho$, $\bar\sigma$, $\bar\pi$ since the rest are the same. Let the vector corresponding to the null word $1$ be $\ket{\phi_1} = \ket{v^{000}}$. Define the subspaces $V_{AB_1} = \Span{P(\{\halpha\})}_P$, $V_{B_2C} = \Span{Q(\{\hgamma\})}_Q$, and $V_{B_3D} = \Span{R(\{\hdelta\})}_R$ each with orthogonal basis $\left\{P_i(\{\halpha\})\ket{v^{000}} = \ket{v^{i00}}, i \in I\right\}$, $\left\{Q_j(\{\hgamma\})\ket{v^{000}} = \ket{v^{0j0}}, j \in J\right\}$, and $\left\{R_k(\{\hdelta\})\ket{v^{000}} = \ket{v^{00k}}, k \in K\right\}$ for some countable index set $I, J, K$, respectively. Using nonlinear factorisation constraints, the analogous calculation as in Appendix~\ref{sec:nonSDPhardproof} shows
\begin{align}
    \bar\rho &= \sum_{i} \ketbra{v^{i00}}{v^{i00}} \\
    \bar\sigma &= \sum_{j} \ketbra{v^{0j0}}{v^{0j0}} \\
    \bar\pi &= \sum_{k} \ketbra{v^{00k}}{v^{00k}}
\end{align}
satisfying all the desired properties.
\end{proof}

\begin{remark}
We remark again that the converse of Theorem \ref{thm:ConvSDPStarShapedNPAHierarchy} is not true in general, i.e. an Anti-state Commutator Star-shaped Trilocal Quantum Distribution may not admit a hierarchy of Factorisation Star-shaped Trilocal Moment Matrices satisfying the three-factorisation constraints. 
Indeed, to this end, we would need to show that the equation
\begin{align*}
    \tr{\tau}{\halpha^{\dagger}\hgamma^{\dagger}\hdelta} \stackrel{?}{=} \Tr{\tau\halpha^{\dagger}\hgamma^{\dagger}\tau\hdelta}
\end{align*}
holds. 
But this is not true since we cannot ``merge'' two $\tau$'s on the right-hand side of the equation with the commutation relations and cyclicity of trace given by the definition of Anti-state Commutator Star-shaped Trilocal Quantum Distributions.

This observation suggests that the factorisation star-shaped trilocal NPA hierarchy converges to a proper subset of all Anti-state Commutator Star-shaped Trilocal Quantum Distributions. It may converge to the model described in Appendix~\ref{sec:AppendixCommutatorBasedDefinition}.
\end{remark}

\section{Bilocal Tsirelson's problem}\label{sec:AppendixBilocalTsirelson}
Note that thanks to~\cite{ligthart2023inflation}, the open problem of obtaining a Tsirelson's Theorem for bilocal networks, which was initially posed by us in the first draft, is fully resolved.
We leave the outdated discussion from the first draft~\cite{renou2022bilocalv1} in Appendix~\ref{sec:AppendixOutdatedTsirelson} for ``historical context'' and discuss the results of~\cite{ligthart2023inflation} in more details in Appendix~\ref{sec:AppendixBilocalTsirelsonSolved}.

\subsection{An outdated discussion on bilocal Tsirelson's theorem as an open question}\label{sec:AppendixOutdatedTsirelson}
We have proposed two alternatives postulates exist to introduce the distributions predicted by quantum theory: the tensor-based postulate of Definition~\ref{def:TensorTripartiteQDist} and the commutator-based postulate of Definition~\ref{def:CommutatorTripartiteQDist}.
While it is trivial that a tensor-based model implies the existence of a commutator-based model, the question of the converse implication was a long-standing problem (equivalent to Kirchberg's Conjecture and Connes' Embedding Conjecture~\cite{Connes_1976, fritz2012tsirelson, Ozawa_2013, Goldbring_2022}) that was only recently disproven by~\cite{ji2021mip}.
However, Tsirelson proved that, restricting to finite-dimensional Hilbert spaces, the two definitions are equivalent~\cite{tsirelsonproblem,tsirelsonNote}. He proved the following theorem:
\begin{theorem}[Tsirelson]\label{thm:tsirelson}
Given a Hilbert space $\cH$, let $\{\hA\}$, $\{\hB\}$ be two finite commuting sets of positive operators in $\mathcal{B}(\cH)$, each generating a finite-dimensional von Neumann algebra $\mathcal{A}$ and $\mathcal{B}$. Let $\vec{Q}=\{q(ab|xy)\}$ be a probability distribution associated with $\{\hA\}$, $\{\hB\}$ and a global state $\tau$.

Then there exists a Hilbert space $\bar{\cH}$ with a decomposition $\bar{\cH} = \cH_A \otimes \cH_B$, such that $\{\hA\}$ can be mapped into $\mathcal{B}(\cH_A)$ and $\{\hB\}$ can be mapped into $\mathcal{B}(\cH_B)$ via injective $^*$-homomorphisms.  
That is, 
\begin{align*}
    \{\hA\} &\mapsto \{\hA \otimes \id_B\} \\
    \{\hB\} &\mapsto \{\id_A \otimes \hB\}.
\end{align*}
Moreover, there exists a state $\bar{\tau}$ on $\bar{\cH}$ recovering the probability distribution $\vec{Q}$ with the operators $\{\hA \otimes \id_B\}$ and $\{\id_A \otimes \hB\}$, namely,
\begin{align*}
    q(ab|xy) = \tr{\tau}{\hA \hB} = \tr{\bar{\tau}}{\hA \otimes \hB}
\end{align*}
\end{theorem}

\begin{remark}\label{rem:InductiveTsirelson}
Observe that for some given Projector Bilocal Quantum Distribution in a finite-dimensional Hilbert space, an inductive application of Theorem~\ref{thm:tsirelson} can easily yield a Tensor Tripartite Quantum model for that distribution.
Thus, Definition~\ref{def:TensorTripartiteQDist} is equivalent to Definition~\ref{def:CommutatorTripartiteQDist} under the assumption of finite-dimensionality.
By the same line of reasoning, we can argue the finite-dimensional equivalence of any multipartite quantum correlations.
\end{remark}

Does the bilocal version of Remark~\ref{rem:InductiveTsirelson} hold? That is, if restricting to finite-dimensional Hilbert spaces, are Definition~\ref{def:TensorBilocQDistrib} and~\ref{def:CommutatorBilocQDistrib} equivalent?

We leave this question open and now discuss why the proof is not a straightforward application of Tsirelson's Theorem~\ref{thm:tsirelson}. 
One of the impasses is due to the fact that the new global state $\bar{\tau}$ constructed does not necessarily admit a bilocal decomposition. Clearly, the equality $\bar{\tau}=\rho'_{AB_L} \otimes \sigma'_{B_RC}$, where $\rho'_{AB_L}$ and $\sigma'_{B_RC}$ are some operators, does not hold. In addition, there is a more subtle problem when embedding the PVMs into the split Hilbert space in bilocal scenarios.

Let us elaborate by discussing the proof of Theorem~\ref{thm:tsirelson}. To this end, we introduce the following known facts in the theory of operator algebra; for more details see~\cite{blackadar2006operator} and~\cite{kadison1997fundamentals}.

\begin{definition}
Given a Hilbert space $\cH$, a von Neumann algebra $\mathcal{A} \subset \mathcal{B}(\cH)$ is called a \emph{factor}\footnote{For our purpose, all factors are automatically of Type-I. We shall omit the terminology for simplicity.} if it has a trivial centre, that is
\begin{align*}
    Z(\mathcal{A}) = \mathcal{A} \cap \mathcal{A}' = \mathbb{C}\id,
\end{align*}
where $\mathcal{A}'$ denotes the commutatant of $\mathcal{A}$.
\end{definition}

The significance of factors is that any von Neumann algebra admits a direct integral decomposition. By assuming finite-dimensionality, we have the following structural result.
\begin{theorem}\label{thm:vonNeumannFiniteDecomposition}
Every finite-dimensional von Neumann algebra $\mathcal{A}$ is a direct sum of $m$ factors, where $m = \dim{Z(\mathcal{A})}$.
\end{theorem}

Let $p$ be a projection on some subspace $\mathcal{K}$ of $\cH$ and $u \in \mathcal{B}(\cH)$, the compression of $u$ to $\mathcal{K}$, denoted by $u_p$ is an element of $\mathcal{B}(\mathcal{K})$ such that $u_p: v \mapsto puv$. Then we have the following lemma.

\begin{lemma}\label{lem:CompressionCommutesCommutant}
Suppose that $\mathcal{A} \subset \mathcal{B}(\cH)$ is a $^*$-algebra and $p \in \mathcal{A}'$ is a projection. Then $p\mathcal{A}p$ and $\mathcal{A}_p = \{u_p \mid u \in \mathcal{A}\}$ are both $^*$-algebra and the map
\begin{align*}
    p\mathcal{A}p \to \mathcal{A}_p, pup \mapsto u_p
\end{align*}
is a $^*$-isomorphism. Moreover, if we also have $p \in \mathcal{A}''$, then
\begin{align*}
    (\mathcal{A}')_p = (\mathcal{A}_p)'.
\end{align*}
\end{lemma}

\begin{theorem}\label{thm:vonNeumannFactorSpatialIsom}
Let $\mathcal{A}$ be a factor in $\mathcal{B}(\cH)$, where $\cH$ is finite-dimensional. Then there exists a spatial isomorphism\footnote{In the theory of operator algebra, it refers to the isomorphism induced by the adjoint action of some unitary.} $\cH \simeq \cH_1 \otimes \cH_2$, represented by a unitary map $u \in \mathcal{B}(\cH)$ such that $u\mathcal{A}u^{\dagger} = \mathcal{B}(\cH_1) \otimes \id_{\cH_2}$.

Furthermore, for the commutator $\mathcal{A}'$ we have $u\mathcal{A}'u^{\dagger} = \id_{\cH_1} \otimes \mathcal{B}(\cH_2)$\footnote{This is a direct consequence of Lemma~\ref{lem:CompressionCommutesCommutant}}.
\end{theorem}

We are ready to give a sketch proof of Theorem~\ref{thm:tsirelson} referencing~\cite{tsirelsonproblem} and~\cite{tsirelsonNote}.
\begin{proof}
The technique is known as ``doubling the centre''. Note that $\mathcal{B} \subset \mathcal{A}$. If $\mathcal{A}$ is a factor, then we are done by Theorem~\ref{thm:vonNeumannFactorSpatialIsom}. Otherwise, since the centre $Z(\mathcal{A})$ is a finite-dimensional unital commutative $C^*$-algebra, the Gelfand--Naimark Theorem gives rise to a finite set of minimal orthogonal central projections $\{p_i\} \subset Z(\mathcal{A})$ that sum up to the identity $\id_{\cH}$.

Therefore, we may identify $\mathcal{A}$ and $\mathcal{A}'$ with direct sum representations $\bigoplus_i \mathcal{A}_i$ and $\bigoplus_i \mathcal{A}'_i$\footnote{Note there is no ambiguity in the notation $\mathcal{A}'_i$ due to Lemma~\ref{lem:CompressionCommutesCommutant}.}, respectively, where $\mathcal{A}_i = \mathcal{A}_{p_i}$. It can be checked that $\mathcal{A}_i$ is a factor for all $i$. Consequently, we have a direct sum decomposition $H \cong \bigoplus_i p_i(\cH) = \bigoplus_i \cH_i$.

Then, applying Theorem~\ref{thm:vonNeumannFactorSpatialIsom} we split $\mathcal{A}_i$ and $\mathcal{A}'_i$ with tensor products and, subsequently, we have $H_i = H^1_i \otimes H^2_i$. It can be checked that the natural embedding of $\cH$ in $\left(\bigoplus_i H^1_i\right) \otimes \left(\bigoplus_i H^1_i\right)$ induces the desired injective $^*$-homomorphism for $\mathcal{A}$ and $\mathcal{A}'$. Moreover, the ($C^*$-algebraic) functional on $\left(\bigoplus_i H^1_i\right) \otimes \left(\bigoplus_i H^1_i\right)$ induced from $\tr{\tau}{\cdot}$ can be checked as a state. Hence, finite-dimensionality implies the existence of an associated density operator $\bar\tau$, which we may purify.
\end{proof}

\begin{remark}
Now we are able to describe the difficulty of showing a bilocal version of Theorem~\ref{thm:tsirelson}, which is due to the existence of operators $\rho$ and $\sigma$. 
Indeed, for example, we may consider the fact that $\mathcal{B}, \mathcal{C}, \{\sigma\} \subset \mathcal{A'}$ and perform a minimal central decomposition of $\mathcal{A}$ and $\mathcal{A}'$. But then we cannot say the same about $\rho$ since $\rho \notin \mathcal{A}'$ in general, and hence it is not clear how to embed $\rho$ isomorphically in the new Hilbert space. A similar problem arises when we attempt to start with the algebra $\mathcal{C}$.

One may also consider the algebra generated by $\mathcal{A}$ and $\rho$, with $\mathcal{C}$ and $\sigma$ lying in its commutant, to obtain a splitting $\cH_{AB_L} \otimes \cH_{B_RC}$. Although seemingly recovering Tensor Bilocal Quantum correlations in Definition~\ref{def:TensorBilocQDistrib}, the algebra $\mathcal{B}$ may not be embedded into the splitting Hilbert space isomorphically. 
The failure of this example is also related to the difficulty of proving the converse of Theorem~\ref{thm:InflationNPAIsTighter} in finite dimension, as it is unclear how to find a sensible splitting of Bob's measurement $\mathcal{B}$.

Furthermore, evident from the proof above, we are not guaranteed that the new global state $\bar\tau$ admits the desired bilocal decomposition.
\end{remark}

\subsection{A new development to the problem}\label{sec:AppendixBilocalTsirelsonSolved}
At the time of the first draft, we left open the question whether the Projector Bilocal Quantum Distribution becomes equivalent to the standard tensor formulation in finite dimension (Appendix~\ref{sec:AppendixOutdatedTsirelson}), which is the case for the standard NPA hierarchy. 
Continuing from our question, the authors of~\cite{ligthart2023inflation} show that our relaxation, the Projector Bilocal Quantum Distribution, is equivalent to the \emph{commuting-observable model} for bilocal network, that is, for $\Vec{Q}$ there are
\begin{enumerate}
    \item mutually commuting $C^*$-subalgebras $\mathcal{A}, \mathcal{B}_L, \mathcal{B}_R, \mathcal{C} \subset \mathcal{D}$ for some global $C^*$-algebra $\mathcal{D}$;
    \item PVMs $\{\A\} \subset \mathcal{A}$, $\{\B\} \subset \mathcal{B} = \mathcal{B}_L \cdot \mathcal{B}_R$, and $\{\C\} \subset \mathcal{C}$;
    \item a state $\rho$ on $\mathcal{D}$ that is product in the sense
    \begin{align*}
        \rho( \alpha \beta_L \beta_R \gamma) = \rho(\alpha \beta_L) \rho(\beta_R \gamma)
    \end{align*}
    for all $\alpha \in \mathcal{A}$, $\beta_L \in \mathcal{B}_L$, $\beta_R \in \mathcal{B}_R$, and $\gamma \in \mathcal{C}$
    and satisfies $q(abc|xyz) = \rho ( \hA \hB \hC)$;
\end{enumerate}
and the \emph{mixed model} for bilocal network, that is, for $\Vec{Q}$ there are
\begin{enumerate}
    \item Hilbert spaces $\cH_{AB_L}$, $\cH_{B_RC}$ (equivalent to $V_{AB_L}$ and $V_{B_RC}$ in our construction);
    \item mutually commuting $C^*$-subalgebras $\mathcal{A}, \mathcal{B}_L \subset B(\cH_{AB_L})$, $\mathcal{B}_R, \mathcal{C} \subset B(\cH_{B_RC})$;
    \item PVMs $\{\A\} \subset \mathcal{A}$, $\{\B\} \subset \mathcal{B} = \mathcal{B}_L \bar{\otimes} \mathcal{B}_R$, and $\{\C\} \subset \mathcal{C}$;
    \item two density operators $\rho$ on $\cH_{AB_L}$ and $\sigma$ on $\cH_{B_RC}$ such that
    \begin{align*}
        q(abc|xyz) = \Tr{ (\rho \otimes \sigma)  (\hA \hB \hC) }.
    \end{align*}
\end{enumerate}
A direct consequence~\cite[Corollary~9]{ligthart2023inflation} is the bilocal Tsirelson's Theorem that establishes the equivalence of these sets of correlations in finite-dimensional settings.

Moreover, they propose another hierarchy, the \emph{polarisation hierarchy} that converges to the same set. Therefore, combining our results, the factorisation bilocal NPA hierarchy, the scalar extension bilocal hierarchy, the inflation-NPA hierarchy, and the polarisation hierarchy all characterise the same commutator relaxation of bilocal scenarios, with finite-dimensional equivalence to tensor formulation.



\bibliographystyle{emss}
\bibliography{reference}

@article{BellTheorem,
  title = {On the Einstein Podolsky Rosen paradox},
  author = {Bell, J. S.},
  journal = {Physics Physique Fizika},
  volume = {1},
  issue = {3},
  pages = {195--200},
  numpages = {6},
  year = {1964},
  month = {Nov},
  publisher = {American Physical Society},
  doi = {10.1103/PhysicsPhysiqueFizika.1.195},
  url = {https://link.aps.org/doi/10.1103/PhysicsPhysiqueFizika.1.195}
}

@article{CHSH,
  title = {Proposed Experiment to Test Local Hidden-Variable Theories},
  author = {Clauser, John F. and Horne, Michael A. and Shimony, Abner and Holt, Richard A.},
  journal = {Phys. Rev. Lett.},
  volume = {23},
  issue = {15},
  pages = {880--884},
  numpages = {0},
  year = {1969},
  month = {Oct},
  publisher = {American Physical Society},
  doi = {10.1103/PhysRevLett.23.880},
  url = {https://link.aps.org/doi/10.1103/PhysRevLett.23.880}
}

@InProceedings{Mayers1998,
  author    = {D. Mayers and A. Yao},
  title     = {{Quantum cryptography with imperfect apparatus}},
  booktitle = {Proc. 39th Symposium on Foundations of Computer Science},
  year      = {1998},
  pages     = {503-509},
  doi       = {10.1109/SFCS.1998.743501},
  note = {(Cat. No.98CB36280)}
}

@article{Vazirani2014,
	doi = {10.1103/physrevlett.113.140501},
  	url = {https://doi.org/10.1103%2Fphysrevlett.113.140501},
  	year = 2014,
	month = {sep},
  	publisher = {American Physical Society ({APS})},
  	volume = {113},
  	number = {14},
  	author = {Umesh Vazirani and Thomas Vidick},
  	title = {Fully Device-Independent Quantum Key Distribution},
  	journal = {Physical Review Letters}
}

@Article{Acin2006,
  author    = {Antonio Ac{\'{\i}}n and Nicolas Gisin and Lluis Masanes},
  title     = {{From Bell's Theorem to Secure Quantum Key Distribution}},
  journal   = {Phys. Rev. Lett.},
  year      = {2006},
  volume    = {97},
  number    = {12},
  pages     = {120405},
  month     = {Sep},
  doi       = {10.1103/PhysRevLett.97.120405},
  issue     = {12},
  numpages  = {4},
  publisher = {American Physical Society},
  url       = {http://link.aps.org/doi/10.1103/PhysRevLett.97.120405},
}

@Article{Pironio2010,
  author    = {S. Pironio and others},
  title     = {{Random numbers certified by Bell’s theorem}},
  journal   = {Nature},
  year      = {2010},
  volume    = {464},
  number    = {7291},
  pages     = {1021--1024},
  doi       = {doi:10.1038/nature09008},
  publisher = {Nature Publishing Group},
  url       = {http://dx.doi.org/10.1038/nature09008},
}

@Article{Colbeck2012,
  author    = {Roger Colbeck and Renato Renner},
  title     = {{Free randomness can be amplified}},
  journal   = {Nat. Phys.},
  year      = {2012},
  volume    = {8},
  number    = {6},
  pages     = {450--453},
  month     = jun,
  issn      = {1745-2473},
  comment   = {10.1038/nphys2300},
  doi       = {10.1038/nphys2300},
  publisher = {Nature Publishing Group},
  url       = {http://dx.doi.org/10.1038/nphys2300},
}

@Article{Arnon2016,
  author  = {Rotem Arnon-Friedman, Renato Renner, Thomas Vidick},
  title   = {{Simple and tight device-independent security proofs}},
  journal = {arXiv:1607.01797},
  year    = {2016},
  pages.hide   = {1607.01797},
  url     = {http://arxiv.org/abs/0911.3814},
}

@Article{ParriloHierarchy,
author  = {Pablo A. Parrilo},
  title   = {{Structured Semidefinite Programs
and Semialgebraic Geometry Methods
in Robustness and Optimization}},
  journal = {PhD thesis},
  year    = {2000},
  url     = {https://thesis.library.caltech.edu/1647/1/Parrilo-Thesis.pdf},
}

@article{LasserreHierarchy,
author = {Lasserre, Jean B.},
title = {A Sum of Squares Approximation of Nonnegative Polynomials},
journal = {SIAM Review},
volume = {49},
number = {4},
pages = {651-669},
year = {2007},
doi = {10.1137/070693709},
URL = { https://doi.org/10.1137/070693709},
eprint = { https://doi.org/10.1137/070693709},
}

@article{Renou2021,
	doi = {10.1038/s41586-021-04160-4},
  	url = {https://doi.org/10.1038%2Fs41586-021-04160-4},
  	year = 2021,
	month = {dec},
  	publisher = {Springer Science and Business Media {LLC}},
  	volume = {600},
  	number = {7890},
  	pages = {625--629},
  	author = {Marc-Olivier Renou and David Trillo and Mirjam Weilenmann and Thinh P. Le and Armin Tavakoli and Nicolas Gisin and Antonio Ac{\'{\i}}n and Miguel Navascu{\'{e}}s},
  	title = {Quantum theory based on real numbers can be experimentally falsified},
  	journal = {Nature}
}

@article{CoiteuxRoyPRL,
	doi = {10.1103/physrevlett.127.200401},
  	url = {https://doi.org/10.1103%2Fphysrevlett.127.200401},
  	year = 2021,
	month = {nov},
  	publisher = {American Physical Society ({APS})},
  	volume = {127},
  	number = {20},
  	author = {Xavier Coiteux-Roy and Elie Wolfe and Marc-Olivier Renou},
  	title = {No Bipartite-Nonlocal Causal Theory Can Explain Nature's Correlations},
  	journal = {Physical Review Letters}
}

@article{CoiteuxRoyPRA,
	doi = {10.1103/physreva.104.052207},
  	url = {https://doi.org/10.1103%2Fphysreva.104.052207},
  	year = 2021,
	month = {nov},
  	publisher = {American Physical Society ({APS})},
  	volume = {104},
  	number = {5},
  	author = {Xavier Coiteux-Roy and Elie Wolfe and Marc-Olivier Renou},
  	title = {Any physical theory of nature must be boundlessly multipartite nonlocal},
  	journal = {Physical Review A}
}

@article{BeigiCovariance,
	doi = {10.1109/tit.2021.3119651},
  	url = {https://doi.org/10.1109%2Ftit.2021.3119651},
  	year = 2022,
	month = {jan},
  	publisher = {Institute of Electrical and Electronics Engineers ({IEEE})},
  	volume = {68},
  	number = {1},
  	pages = {384--394},
  	author = {Salman Beigi and Marc-Olivier Renou},
  	title = {Covariance Decomposition as a Universal Limit on Correlations in Networks},
  	journal = {{IEEE} Transactions on Information Theory}
}

@article{BancalNLCouplers,
	doi = {10.1103/physreva.104.052212},
  	url = {https://doi.org/10.1103%2Fphysreva.104.052212},
  	year = 2021,
	month = {nov},
  	publisher = {American Physical Society ({APS})},
  	volume = {104},
  	number = {5},
  	author = {Jean-Daniel Bancal and Nicolas Gisin},
  	title = {Nonlocal boxes for networks},
  	journal = {Physical Review A}
}

@article{GisinNSI,
	doi = {10.1038/s41467-020-16137-4},
  	url = {https://doi.org/10.1038%2Fs41467-020-16137-4},
  	year = 2020,
	month = {may},
  	publisher = {Springer Science and Business Media {LLC}},
  	volume = {11},
  	number = {1},
  	author = {Nicolas Gisin and Jean-Daniel Bancal and Yu Cai and Patrick Remy and Armin Tavakoli and Emmanuel Zambrini Cruzeiro and Sandu Popescu and Nicolas Brunner},
  	title = {Constraints on nonlocality in networks from no-signaling and independence},
  	journal = {Nature Communications}
}

@article{Weilenmann2020,
	doi = {10.1103/physrevlett.125.060406},
  	url = {https://doi.org/10.1103%2Fphysrevlett.125.060406},
  	year = 2020,
	month = {aug},
  	publisher = {American Physical Society ({APS})},
  	volume = {125},
  	number = {6},
  	author = {Mirjam Weilenmann and Roger Colbeck},
  	title = {Self-Testing of Physical Theories, or, Is Quantum Theory Optimal with Respect to Some Information-Processing Task?},
  	journal = {Physical Review Letters}
}

@article{SelfTestingReview,
	doi = {10.22331/q-2020-09-30-337},
  	url = {https://doi.org/10.22331%2Fq-2020-09-30-337},
  	year = 2020,
	month = {sep},
  	publisher = {Verein zur Forderung des Open Access Publizierens in den Quantenwissenschaften},
  	volume = {4},
  	pages = {337},
  	author = {Ivan {\v{S}
}upi{\'{c}} and Joseph Bowles},
  	title = {Self-testing of quantum systems: a review},
  	journal = {Quantum}
}

@Article{Henson2014,
  author  = {Joe Henson and Raymond Lal and Matthew F. Pusey},
  title   = {{Theory-independent limits on correlations from generalized Bayesian networks}},
  journal = {New J. Phys.},
  year    = {2014},
  volume  = {16},
  number  = {11},
  pages   = {113043},
  url     = {http://stacks.iop.org/1367-2630/16/i=11/a=113043},
}

@article{Scarani2012device,
  title={{The Device-Independent Outlook on Quantum Physics}},
  author={Scarani, Valerio},
  journal={Acta Physica Slovaca},
  volume={62},
  number={4},
  pages={347},
  year={2012},
  publisher={SLOVAK ACAD SCIENCES INST PHYSICS DUBRAVSKA CESTA 9, 842 28 BRATISLAVA, SLOVAKIA},
  url={http://www.physics.sk/aps/pub.php?y=2012&pub=aps-12-04},
  doi={10.2478/v10155-012-0003-4 }
}

@article{Renou2018,
	doi = {10.1103/physrevlett.121.250507},
  	url = {https://doi.org/10.1103%2Fphysrevlett.121.250507},
  	year = 2018,
	month = {dec},
  	publisher = {American Physical Society ({APS})},
  	volume = {121},
  	number = {25},
  	author = {Marc Olivier Renou and J{\k{e}
}drzej Kaniewski and Nicolas Brunner},
  	title = {Self-Testing Entangled Measurements in Quantum Networks},
  	journal = {Physical Review Letters}
}

@article{Bancal2018,
	doi = {10.1103/physrevlett.121.250506},
  	url = {https://doi.org/10.1103%2Fphysrevlett.121.250506},
  	year = 2018,
	month = {dec},
  	publisher = {American Physical Society ({APS})},
  	volume = {121},
  	number = {25},
  	author = {Jean-Daniel Bancal and Nicolas Sangouard and Pavel Sekatski},
  	title = {Noise-Resistant Device-Independent Certification of Bell State Measurements},
  	journal = {Physical Review Letters}
}

@misc{Pavel2022,
  doi = {10.48550/ARXIV.2209.09921},
    url = {https://arxiv.org/abs/2209.09921},
    author = {Sekatski, Pavel and Boreiri, Sadra and Brunner, Nicolas},
    title = {Partial self-testing and randomness certification in the triangle network},
    publisher = {arXiv},
    year = {2022},
    copyright = {arXiv.org perpetual, non-exclusive license}
}

@misc{Supic2022,
  doi = {10.48550/ARXIV.2201.05032},
    url = {https://arxiv.org/abs/2201.05032},
    author = {Šupić, Ivan and Bowles, Joseph and Renou, Marc-Olivier and Acín, Antonio and Hoban, Matty J.},
    title = {Quantum networks self-test all entangled states},
    publisher = {arXiv},
    year = {2022},
    copyright = {Creative Commons Attribution 4.0 International}
}

@Article{BrunnerReview,
  author    = {Brunner, Nicolas and Cavalcanti, Daniel and Pironio, Stefano and Scarani, Valerio and Wehner, Stephanie},
  title     = {Bell nonlocality},
  journal   = {Rev. Mod. Phys.},
  year      = {2014},
  volume    = {86},
  pages     = {419--478},
  month     = {Apr},
  doi       = {10.1103/RevModPhys.86.419},
  issue     = {2},
  numpages  = {60},
  publisher = {American Physical Society},
  url       = {http://link.aps.org/doi/10.1103/RevModPhys.86.419},
}

@article{NetworkNonlocReview,
	doi = {10.1088/1361-6633/ac41bb},
 	url = {https://doi.org/10.1088%2F1361-6633%2Fac41bb},
 	year = 2022,
	month = {mar},
 	publisher = {{IOP} Publishing},
 	volume = {85},
 	number = {5},
 	pages = {056001},
 	author = {Armin Tavakoli and Alejandro Pozas-Kerstjens and Ming-Xing Luo and Marc-Olivier Renou},
 	title = {Bell nonlocality in networks},
 	journal = {Reports on Progress in Physics}
}

@article{Renou2022a,
	doi = {10.1103/physrevlett.128.060401},
  	url = {https://doi.org/10.1103%2Fphysrevlett.128.060401},
  	year = 2022,
	month = {feb},
  	publisher = {American Physical Society ({APS})},
  	volume = {128},
  	number ={6},
  	author = {Marc-Olivier Renou and Salman Beigi},
  	title = {Nonlocality for Generic Networks},
  	journal = {Physical Review Letters}
}

@article{Renou2019,
	doi = {10.1103/physrevlett.123.140401},
  	url = {https://doi.org/10.1103%2Fphysrevlett.123.140401},
  	year = 2019,
	month = {sep},
  	publisher = {American Physical Society ({APS})},
  	volume = {123},
  	number = {14},
  	author = {Marc-Olivier Renou and Elisa Bäumer and Sadra Boreiri and Nicolas Brunner and Nicolas Gisin and Salman Beigi},
  	title = {Genuine Quantum Nonlocality in the Triangle Network},
  	journal = {Physical Review Letters}
}

@article{Branciard2010,
  title = {Characterizing the Nonlocal Correlations Created via Entanglement Swapping},
  author = {Branciard, C. and Gisin, N. and Pironio, S.},
  journal = {Phys. Rev. Lett.},
  volume = {104},
  issue = {17},
  pages = {170401},
  numpages = {4},
  year = {2010},
  month = {Apr},
  publisher = {American Physical Society},
  doi = {10.1103/PhysRevLett.104.170401},
  url = {https://link.aps.org/doi/10.1103/PhysRevLett.104.170401}
}

@article{Fritz2012,
	doi = {10.1088/1367-2630/14/10/103001},
  	url = {https://doi.org/10.1088%2F1367-2630%2F14%2F10%2F103001},
  	year = 2012,
	month = {oct},
  	publisher = {{IOP} Publishing},
  	volume = {14},
  	number = {10},
  	pages = {103001},
  	author = {Tobias Fritz},
  	title = {Beyond Bell{\textquotesingle}s theorem: correlation scenarios},
  	journal = {New Journal of Physics}
}

@article{Lee2018,
  title = {Towards Device-Independent Information Processing on General Quantum Networks},
  author = {Lee, Ciar\'an M. and Hoban, Matty J.},
  journal = {Phys. Rev. Lett.},
  volume = {120},
  issue = {2},
  pages = {020504},
  numpages = {6},
  year = {2018},
  month = {Jan},
  publisher = {American Physical Society},
  doi = {10.1103/PhysRevLett.120.020504},
  url = {https://link.aps.org/doi/10.1103/PhysRevLett.120.020504}
}

@article{NPA2008,
   title={A convergent hierarchy of semidefinite programs characterizing the set of quantum correlations},
   volume={10},
   ISSN={1367-2630},
   url={http://dx.doi.org/10.1088/1367-2630/10/7/073013},
   DOI={10.1088/1367-2630/10/7/073013},
   number={7},
   journal={New Journal of Physics},
   publisher={IOP Publishing},
   author={Navascués, Miguel and Pironio, Stefano and Acín, Antonio},
   year={2008},
   month={Jul},
   pages={073013} }

@book{BratelliRobinson,
   author = {Bratteli, Ola and Robinson, Derek W.},
   address = {Berlin},
   booktitle = {Operator Algebras and Quantum Statistical Mechanics},
   isbn = {978-3-662-03444-6},
   publisher = {Springer Berlin, Heidelberg},
   series = {Theoretical and Mathematical Physics},
   title = {Operator Algebras and Quantum Statistical Mechanics},
   year = {1997},
}

@article{Inflation,
	doi = {10.1515/jci-2017-0020},
  	url = {https://doi.org/10.1515%2Fjci-2017-0020},
  	year = 2019,
	month = {jul},
  	publisher = {Walter de Gruyter {GmbH}},
  	volume = {7},
  	number = {2},
  	author = {Elie Wolfe and Robert W. Spekkens and Tobias Fritz},
  	title = {The Inflation Technique for Causal Inference with Latent Variables},
  	journal = {Journal of Causal Inference}
}

@article{ligthart2023convergent,
  title={A convergent inflation hierarchy for quantum causal structures},
  author={Ligthart, Laurens T and Gachechiladze, Mariami and Gross, David},
  journal={Communications in Mathematical Physics},
  pages={1--42},
  year={2023},
  publisher={Springer},
  doi={10.1007/s00220-023-04697-7}
}

@article{NavascuesCVStandardInflation,
	doi = {10.1515/jci-2018-0008},
	url = {https://doi.org/10.1515%2Fjci-2018-0008},
	year = 2020,
	month = {sep},
	publisher = {Walter de Gruyter {GmbH}},
	volume = {8},
	number = {1},
	pages = {70--91},
	author = {Miguel Navascu{\'{e}}s and Elie Wolfe},
	title = {The Inflation Technique Completely Solves the Causal Compatibility Problem},
	journal = {Journal of Causal Inference}
}

@misc{VictorCVStandardInflation,
  doi = {10.48550/ARXIV.2202.04103},
    url = {https://arxiv.org/abs/2202.04103},
    author = {Gitton, Victor},
    title = {Outer approximations of classical multi-network correlations},
    publisher = {arXiv},
    year = {2022},
    copyright = {Creative Commons Attribution 4.0 International}
}

@ARTICLE{tsirelsonproblem,
       author = {{Scholz}, V.~B. and {Werner}, R.~F.},
        title = "{Tsirelson's Problem}",
      journal = {arXiv e-prints},
         year = 2008,
        month = dec,
          eid = {arXiv:0812.4305},
        pages = {arXiv:0812.4305},
archivePrefix = {arXiv},
       eprint = {0812.4305},
 primaryClass = {math-ph},
       adsurl = {https://ui.adsabs.harvard.edu/abs/2008arXiv0812.4305S}
}

@article{ji2021mip,
  title={Mip*=RE},
  author={Ji, Zhengfeng and Natarajan, Anand and Vidick, Thomas and Wright, John and Yuen, Henry},
  journal={Communications of the ACM},
  volume={64},
  number={11},
  pages={131--138},
  year={2021},
  publisher={ACM New York, NY, USA}
}

@article{ScalarExtension2019,
  title = {Bounding the Sets of Classical and Quantum Correlations in Networks},
  author = {Pozas-Kerstjens, Alejandro and Rabelo, Rafael and Rudnicki, \L{}ukasz and Chaves, Rafael and Cavalcanti, Daniel and Navascu\'es, Miguel and Ac\'{\i}n, Antonio},
  journal = {Phys. Rev. Lett.},
  volume = {123},
  issue = {14},
  pages = {140503},
  numpages = {6},
  year = {2019},
  month = {Oct},
  publisher = {American Physical Society},
  doi = {10.1103/PhysRevLett.123.140503},
  url = {https://link.aps.org/doi/10.1103/PhysRevLett.123.140503}
}

@article{QuantumInflation2021,
  title = {Quantum Inflation: A General Approach to Quantum Causal Compatibility},
  author = {Wolfe, Elie and Pozas-Kerstjens, Alejandro and Grinberg, Matan and Rosset, Denis and Ac\'{\i}n, Antonio and Navascu\'es, Miguel},
  journal = {Phys. Rev. X},
  volume = {11},
  issue = {2},
  pages = {021043},
  numpages = {24},
  year = {2021},
  month = {May},
  publisher = {American Physical Society},
  doi = {10.1103/PhysRevX.11.021043},
  url = {https://link.aps.org/doi/10.1103/PhysRevX.11.021043}
}

@book{kadison1997fundamentals,
  title={Fundamentals of the Theory of Operator Algebras. Volume II},
  author={Kadison, R.V. and Ringrose, J.R.},
  isbn={9780821808207},
  lccn={97020916},
  series={Fundamentals of the Theory of Operator Algebras},
  url={https://books.google.ch/books?id=h5bMkZTnowAC},
  year={1997},
  publisher={American Mathematical Society}
}

@book{blackadar2006operator,
  title={Operator Algebras: Theory of C*-Algebras and von Neumann Algebras},
  author={Blackadar, B.},
  isbn={9783540285175},
  series={Encyclopaedia of Mathematical Sciences},
  url={https://books.google.ch/books?id=7-6MZLdRfdAC},
  year={2006},
  publisher={Springer Berlin Heidelberg}
}

@misc{tsirelsonNote,
  author = {Tsirelson, B.},
  title = {Bell inequalities and operator algebras},
  year={2006},
  url = {https://www.tau.ac.il/~tsirel/Research/bellopalg/main.html}
}

@article{fritz2012tsirelson,
  title={Tsirelson's problem and Kirchberg's conjecture},
  author={Fritz, T.},
  journal={Reviews in Mathematical Physics},
  volume={24},
  number={05},
  pages={1250012},
  year={2012},
  publisher={World Scientific}
}

@article{Ozawa_2013,
	doi = {10.1007/s11537-013-1280-5},
	url = {https://doi.org/10.1007%2Fs11537-013-1280-5},
	year = 2013,
	month = {mar},
	publisher = {Springer Science and Business Media {LLC}},
	volume = {8},
	number = {1},
	pages = {147--183},
	author = {Ozawa, N.},
	title = {About the Connes embedding conjecture},
	journal = {Japanese Journal of Mathematics}
}

@article{Goldbring_2022,
	doi = {10.1090/bull/1768},
	url = {https://doi.org/10.1090%2Fbull%2F1768},
	year = 2022,
	month = {jun},
	publisher = {American Mathematical Society ({AMS})},
	volume = {59},
	number = {4},
	pages = {503--560},
	author = {Goldbring, I.},
	title = {The Connes embedding problem: A guided tour},
	journal = {Bulletin of the American Mathematical Society}
}

@article{Connes_1976,
	doi = {10.2307/1971057},
	url = {https://doi.org/10.2307%2F1971057},
	year = 1976,
	month = {jul},
	publisher = {{JSTOR}},
	volume = {104},
	number = {1},
	pages = {73},
	author = {A. Connes},
	title = {Classification of Injective Factors Cases {II} 1 , {II} $\infty$ , {III} $\uplambda$ , $\uplambda$ 1},
	journal = {The Annals of Mathematics}
}

@book{burgdorf2016optimization,
  title={Optimization of Polynomials in Non-Commuting Variables},
  author={Burgdorf, S. and Klep, I. and Povh, J.},
  isbn={9783319333366},
  series={SpringerBriefs in Mathematics},
  url={https://books.google.ch/books?id=Qe8MkAEACAAJ},
  year={2016},
  publisher={Springer International Publishing}
}

@article{ligthart2023inflation,
  title={The inflation hierarchy and the polarization hierarchy are complete for the quantum bilocal scenario},
  author={Ligthart, Laurens T and Gross, David},
  journal={Journal of Mathematical Physics},
  volume={64},
  number={7},
  year={2023},
  publisher={AIP Publishing}
}

@article{pironio2010convergent,
  title={Convergent relaxations of polynomial optimization problems with noncommuting variables},
  author={Pironio, Stefano and Navascu{\'e}s, Miguel and Ac{\'i}n, Antonio},
  journal={SIAM Journal on Optimization},
  volume={20},
  number={5},
  pages={2157--2180},
  year={2010},
  publisher={SIAM},
  doi={10.1137/090760155}
}

@article{lasserre2001global,
  title={Global optimization with polynomials and the problem of moments},
  author={Lasserre, Jean B},
  journal={SIAM Journal on optimization},
  volume={11},
  number={3},
  pages={796--817},
  year={2001},
  publisher={SIAM},
  doi={10.1137/S1052623400366802}
}

@article{klep2024state,
  title={State polynomials: positivity, optimization and nonlinear Bell inequalities},
  author={Klep, Igor and Magron, Victor and Vol{\v{c}}i{\v{c}}, Jurij and Wang, Jie},
  journal={Mathematical Programming},
  volume={207},
  number={1},
  pages={645--691},
  year={2024},
  publisher={Springer}
}

@article{arveson1969subalgebras,
 author = {Arveson, William B},
 title = {Subalgebras of {{\(C^ *\)}}-algebras},
 fjournal = {Acta Mathematica},
 journal = {Acta Math.},
 issn = {0001-5962},
 volume = {123},
 pages = {141--224},
 year = {1969},
 doi = {10.1007/BF02392388},
 zbMATH = {3308710},
 Zbl = {0194.15701}
}

@misc{renou2022bilocalv1,
  author        = {Renou, Marc-Olivier and Xu, Xiangling},
  title         = {Two convergent NPA-like hierarchies for the quantum bilocal scenario (version 1)},
  archivePrefix = {arXiv},
  eprint        = {2210.09065v1},
  version       = {1},
  year          = {2022},
  primaryClass={quant-ph},
  url           = {https://arxiv.org/abs/2210.09065v1},
  month         = oct
}

\end{document}